\documentclass[12pt]{article}
\usepackage{lipsum}
\usepackage{setspace} 
\usepackage{amsmath, amsfonts, amssymb, amsthm, graphicx, rotating, newtxmath, mathrsfs}
\usepackage{lscape}
\usepackage{booktabs, siunitx}
\usepackage[normalem]{ulem}
\usepackage[colorlinks=true, linkcolor=blue, citecolor=blue, urlcolor=blue]{hyperref}
\usepackage{bbm}
\usepackage[english]{babel}
\usepackage[autostyle]{csquotes}
\usepackage{dsfont}
\usepackage[svgnames,table]{xcolor}
\usepackage[tableposition=above]{caption}
\usepackage{multirow}
\usepackage{makecell}
\usepackage[style=apa, backend=biber, citestyle=authoryear, sorting=nyt, doi=false]{biblatex}
\DeclareLanguageMapping{american}{american-apa}
\usepackage{csquotes}
\usepackage{array}
\usepackage{graphicx}
\usepackage{float}

\newtheorem*{theorem*}{Theorem}
\numberwithin{equation}{section}
\newtheorem{theorem}{Theorem}[section]
\newtheorem{corollary}{Corollary}[section]
\newtheorem{proposition}{Proposition}[section]
\newtheorem{lemma}{Lemma}[section]
\newtheorem*{lemma3}{Lemma 3}
\newtheorem{remark}{Remark}
\theoremstyle{definition} 
\newtheorem{definition}{Definition}[section]
\newtheorem*{definition*}{Definition}
\newtheorem*{proposition*}{Proposition}

\usepackage{mathtools}  
\usepackage{diffcoeff}
\usepackage{algorithm}
\usepackage{algpseudocode}
\usepackage{enumitem}
\usepackage[margin=1in]{geometry}
\usepackage{fancyvrb}
\usepackage{tikz} 
\usetikzlibrary{shapes.geometric, arrows, arrows.meta, shadows, shadows.blur, positioning, calc, decorations.pathmorphing}
\usepackage{stackengine}
\usepackage{float}
\usepackage{placeins}
\usepackage{afterpage}
\usepackage{graphicx}
\usepackage{wrapfig}
\usepackage{ragged2e}
\usepackage{mdframed}
\usepackage{accents}
\usepackage{chngcntr} 
\counterwithin{figure}{section} 
\usepackage[x11names]{xcolor}
\usepackage{appendix} 
\begin{document}
\title{Topological Learning under Sparsity through Huberized Distance Functionals}
\author{\small 
	Tuhin Subhra Mahato \\
	\small IIT Kanpur\\
	\small Department of Statistics and Data Science \\
	\small  Kanpur 208016, India\\
	{\small Email: tuhinsm23@iitk.ac.in }\\
	\and
	\small Subhra Sankar Dhar \\
	\small  IIT Kanpur\\
	\small   Department of Statistics and Data Science \\
	\small Kanpur 208016, India\\
	{\small Email: subhra@iitk.ac.in}\\
}
\date{}
\maketitle
 \begin{center}
     \textbf{Abstract} 
 \end{center} 
Learning topological structure from sparsely contaminated data requires robust statistical methods, as even a small fraction of atypical observations can substantially distort the inferred topology. This article introduces a robust framework for topological learning 
through the Huberized Distance-to-Measure, which provides a robust filtration-generating functional for persistent homology.
We establish large-sample properties of the proposed functional, and show that the associated filtration induces
stable and consistent persistence diagrams. We further derive a maximum bias bound for persistence diagrams in bottleneck distance, thereby providing a quantitative measure of the sparseness associated with topological inference.
The practical utility of the methodology thorough testing of statistical hypothesis problem is illustrated for a seven-circle packing configuration under sparse contamination, where it recovers the underlying topological structure. Its applicability to real data is further demonstrated through an analysis of the M-ferritin protein structure, where the method identifies a persistent $H_2$ feature corresponding to the iron-storage cavity.

\newpage

\section{Introduction}
Recovering the geometric and topological structure of data from noisy finite samples is a fundamental problem in modern statistical learning. Traditional nonparametric methods primarily focus on estimating characteristics of the underlying distribution,
such as densities or regression functions.
In contrast, topological data analysis (TDA) provides a framework for inferring the shape of data
from finite samples in metric spaces. Rather than describing data in terms of distributional functionals, TDA characterizes the topology of the support of the underlying distribution;
see \textcite{carlsson2009topology, edelsbrunner2010computational}. In notation, 
let $\mathbb{P}$ be a probability distribution with compact support $\mathcal{S}_{\mathbb{P}}\subset \mathbb{R}^{\ell}~ (\ell \geq 1)$. The shape of $\mathcal{S}_{\mathbb{P}}$ can be described through its homology, an algebraic invariant which quantifies topological features such as connected components, loops, and higher-dimensional voids. Since $\mathcal{S}_{\mathbb{P}}$ is observed only through a finite sample from $\mathbb{P}$, one instead studies these features using persistent homology, a multiscale extension of homology, which
tracks how these topological features evolve as the scale parameter changes.
The resulting multiscale topological information
is
summarized by the persistence diagram, a multiset of points in $\bar{\mathbb{R}}^2$, where each point records the birth and death scales of a topological feature.
The construction of persistent homology relies on a filtration, which is a nested family of topological spaces indexed by a scale parameter. 
Typically, this filtration is 
generated by the sublevel sets $ \{x:f(x)\leq t\}_{t\in \mathbb{R}}$ associated with a real-valued function $f$ defined on the ambient space. 
Consequently, the resulting persistence diagram is determined by the choice of the filtration-generating function.
Statistical properties of persistence diagrams therefore depend on the stability and robustness of the underlying filtration-generating function.

From a statistical perspective, persistence diagrams are estimated from finite samples and are therefore subject to sampling variability, measurement error, and contamination.
As these perturbations propagate through the underlying filtration, the reliability of topological inference depends on the stability and robustness of the filtration-generating function.
 \subsection{The Distance-to-Measure and the Challenge of Sparsity}
A prominent example of such a filtration-generating function is the distance-to-measure (DTM), introduced by \textcite{chazal2011geometric} 
as a stable alternative to distance-to-set functions in the sense of \textcite{chazal2006sampling}. Since its introduction, DTM has become a cornerstone of topological inference. More precisely, given a probability measure $\mathbb{P}$ on a metric space, the DTM 
defines a nonnegative Lipschitz function whose sublevel sets induce a filtration.
Under some regularity conditions, the associated persistence diagram recovers the homological structure of the support of $\mathbb{P}$ while remaining stable under perturbations of $\mathbb{P}$.
The statistical and computational aspects of DTM-based inference have been extensively studied in \textcite{chazal2018robust}, \textcite{anai2020dtm}, and subsequent works.
Despite its stability,
the empirical DTM remains vulnerable to sparse contamination, where a small fraction of observations lie far from the support of the true distribution. In practice, the empirical DTM at a given point $x$ is computed from average squared distances to its 
$k$-nearest neighbours. 
Since these averages are computed without trimming or downweighting extreme distances, even a small fraction of outlying observations within nearest neighbourhoods can substantially bias the estimate.
This can distort the underlying filtration, and lead to unreliable topological inference such as loss of genuine homological features or creation of spurious ones.
This observation motivates the development of robust distance-based filtration-generating functionals for topological inference.

\subsection{Related Work}
Robust estimation under contamination has been widely studied in statistics.
The pioneering work of \textcite{huber1964robust} introduced the contamination model and proposed robust M-estimators based on bounded influence score functions. Subsequent developments by \textcite{huber2009robust} and \textcite{catoni2012challenging} provided a comprehensive asymptotic theory and finite-sample bounds for robust estimators under contamination. 
A guiding principle of robust statistics is that a small number of atypical observations should not substantially change an estimate.
Under sparse contamination, classical DTM lacks this bounded-influence property.
On the topological side, statistical inference for geometric and topological structure has been extensively studied over the last two decades. Foundational results on homology inference and support recovery were established by \textcite{niyogi2008finding} and \textcite{balakrishnan2012minimax}. Statistical inference for persistence diagrams was further developed through 
confidence sets (\textcite{fasy2014confidence}).
Recent advances in the statistical analysis of persistence diagrams and other topological summaries include \textcite{divol2019choice}, together with testing procedures for homological structures (\textcite{kumar2026testing}), and methods for identifying significant topological patterns in regression curve (\textcite{kumar2026novel}). Broader treatments of statistical TDA are given by \textcite{wasserman2018topological}.

Despite these developments, robust statistical analysis at the level of the filtration-generating function has received comparatively less attention in TDA. Existing robustness results in TDA primarily focus on the stability of topological summaries once the filtration has been specified.  \textcite{cohen2005stability} established the foundational stability theorem for persistence diagrams.
For DTM-based methods, \textcite{buchet2015topological} studied the stability of DTM-based filtrations under Hausdorff-type perturbations. Statistical methodologies based on persistence summaries and vectorized representations, including persistence landscapes (\textcite{bubenik2015statistical}), persistence images (\textcite{adams2017persistence}), and kernel methods for persistence diagrams (\textcite{kusano2016persistence}; \textcite{carriere2017sliced}); have also received considerable attention; see \textcite{ali2023survey} for a recent survey. \textcite{phillips2015geometric} studied the kernel-based topological inference under contamination, while \textcite{adler2014crackle} investigated the homological structure generated by noise. \textcite{bobrowski2017maximally} analyzed random Cech complexes under uncontaminated sampling. More recently, a few robust methodologies for topological inference have been proposed; see, for example, \textcite{vishwanath2022robust}, \textcite{zhang23bb}, and \textcite{mahato2025robustpersistenthomology}. \textcite{bobrowski2023cluster} addressed outlier sensitivity through cluster persistence in weighted graphs. To the best of our knowledge, robustness is typically addressed either at the level of the persistence diagram or through geometric stability arguments after the filtration has been constructed. But something that seems to have received either very little or almost no attention in the existing literature is a systematic study of sparse contamination at the level of the filtration-generating function itself, together with its statistical and topological consequences.

To address sparse contamination, we propose a robust framework for filtration-based topological inference.
The main idea is to generalize the classical DTM by replacing its quadratic distance aggregation
with a Huber-type robust M-estimator. The resulting construction is referred to as the Huberized Distance-to-Measure (Huberized DTM), defines a distance-based filtration whose sublevel sets are suitable for persistent homology. This proposed framework induces a robustness-efficiency tradeoff through a tuning parameter, with the classical DTM is found as a limiting case.

\subsection{Objectives and Contributions}
The objective of this work is to develop a robust framework for topological inference in presence of sparse contamination. Instead of treating persistence diagrams as the primary inferential objects, we focus on robust estimation of the underlying filtration-generating function from which persistent homology is constructed. Since persistence diagrams inherit their geometric and statistical properties from the underlying filtration, stability theorems translate perturbations of the filtration function into perturbations of the associated persistence diagrams. Thus, statistical consistency and robustness of the filtration function carry over to the persistence diagrams through the stability theory of persistent homology. Robust estimation of the filtration function thus constitutes the key to robust topological inference under contamination.
The proposed approach therefore addresses the effect of sparse contamination 
at the level of the distance functional, before it propagates to the topological summary.

The contributions of this work are as follows. First, we formulate topological inference under sparse contamination as a problem of robust estimation of distance functionals that generate filtrations. This formulation places robustness at the level of filtration function from which persistent homology is constructed. Second, we introduce the Huberized DTM, a robust distance functional obtained by incorporating the Huber loss into the DTM construction. The resulting estimator has bounded influence, 
which attenuates the effect of sparse contamination
while preserving the multiscale geometric structure underlying persistent homology. We further establish that the Huberized DTM is well-defined and Lipschitz continuous on the metric space $(\mathcal{M},d)$, providing the regularity required for the induced filtration and its subsequent topological analysis. Third, we establish uniform consistency of the empirical Huberized DTM over the ambient space $\mathcal{M}$,
and show that this consistency carries over to the induced persistence diagrams in bottleneck distance.
Together these results ensure reliable topological inference from finite samples.
We also study robustness under Huber's contamination model through maximum bias analysis. The maximum bias of the Huberized DTM is of order $\mathcal{O}\left(\frac{\alpha}{1-\alpha} \right)$ in the contamination proportion $\alpha\in (0, \frac{1}{2})$, uniformly over all contaminating distributions. For the induced persistence diagrams, the maximum bias in bottleneck distance consists of a functional component, inherited from the Huberized DTM and a geometric component 
arising from the perturbation of the contaminated support relative to the signal support.
Under a geometric sparsity condition on the contamination, both components are of order $\mathcal{O}\left(\frac{\alpha}{1-\alpha} \right)$. These results provide a rigorous quantitative understanding of how sparse contamination affects both the filtration function and the resulting topological summaries.

As a part of our contributions, we analyze a molecular data to demonstrate the practical utility of the proposed Huberized DTM framework. In the course of this study, it is observe that the proposed framework isolates the biologically meaningful iron-storage cavity as a dominant $H_2$ feature. For details, see \hyperref[section 7]{\textcolor{blue}{Section \ref{section 7}}}.

\subsection{Challenges}
Robust topological inference under sparse contamination involves several interconnected statistical, geometric, and topological challenges. A key statistical challenge is to establish uniform convergence of the filtration-generating distance functional over the entire metric space $\mathcal{M}$, as pointwise convergence does not ensure a reliable approximation of the induced filtration for topological inference. For the Huberized DTM, this analysis requires a function class indexed jointly by the 
evaluation point $x \in \mathcal{M}$ and its associated robust parameter $\theta \in \mathbb{R}$ within the Huber estimating equation (see, \hyperref[section 3.2]{\textcolor{blue} Section \ref{section 3.2}}), while accounting for the quantile-based trimming inherent in the DTM construction. 
Additionally, the distance functional must satisfy the regularity conditions needed for the induced filtration to be well defined and amenable to topological analysis.
On the geometric side, contamination may change the support on which the filtration is constructed. As a result, its effect on the persistence diagrams arises through both the perturbations of the distance functional as well as the perturbations of the support. Addressing these challenges requires combining ideas from robust statistics, empirical-process theory, and metric geometry with the stability theory of persistence diagrams (\textcite{cohen2005stability}). From a computational perspective, persistent homology becomes challenging as sample size (denoted by $n$) increase. Storing the pairwise distance matrix requires $\mathcal{O}(n^2)$ memory, while computing the $k$-the homology group $H_k$ requires constructing simplices up to dimension $(k+1)$, whose number increases as $\mathcal{O}(n^{k+2})$. Consequently, both the complex construction and the persistent homology computation become expensive as the sample size $n$ and/or homology dimension $k$ increases.


\subsection{Organization} 
The remainder of the article is organized as follows. \hyperref[section 2]{\textcolor{blue}{Section \ref{section 2}}} introduces the distance-to-measure framework, the contamination model, and the standing assumptions. \hyperref[section 3]{\textcolor{blue}{Section \ref{section 3}}} presents the proposed Huberized DTM framework and the associated filtration for topological inference. The large-sample properties of the proposed Huberized DTM including statistical consistency and the corresponding stability and consistency results for the associated persistence diagrams are established in \hyperref[section 4]{\textcolor{blue}{Section \ref{section 4}}}. \hyperref[section 5]{\textcolor{blue}{Section \ref{section 5}}} studies the robustness of the proposed Huberized DTM under sparse contamination, including maximum bias bounds and topological robustness of the associated persistence diagrams. The finite-sample performance of the proposed method is evaluated through simulation studies in \hyperref[section 6]{\textcolor{blue}{Section \ref{section 6}}}, while \hyperref[section 7]{\textcolor{blue}{Section \ref{section 7}}} illustrates the practical relevance of our approach by applying it to a real-world dataset related to cellular biology. \hyperref[section 8]{\textcolor{blue}{Section \ref{section 8}}} concludes with a discussion of the main findings, limitations, and directions for future research. Technical proofs together with some standing notation and definitions are provided in \hyperref[Appendix]{\textcolor{blue}{Appendix}}.

\subsection{Notation}
We summarize the notation used throughout the article. $(\mathcal{M},d)$ denotes a compact metric space, and $\mathscr{P}(\mathcal M)$ denotes the collection of Borel probability measures on $\mathcal{M}$. For $\mathbb{P} \in \mathscr{P}(\mathcal M)$, we write $\mathcal{S}$ for its support. For any nonempty set $\mathcal{A} \subseteq \mathcal{M}$, its diameter is defined by
$\text{diam}(\mathcal{A}):= \sup_{x,y\in \mathcal{A}} d(x,y)$. For $x \in \mathcal{M}$ and $\varepsilon>0$, $B_\varepsilon(x):= \{y \in \mathcal{M}: d(x, y)< \varepsilon \}$ denotes the open ball of radius $ \varepsilon $ centered at $x$, and $\mathbbm{1}_{\mathcal{A}}$ denotes the indicator function of a measurable set $\mathcal{A}$. For two sets, $A,B\subseteq \mathcal{M}$, their symmetric difference is denoted by $A \triangle B := (A \cap B^c) \cup (A^c \cap B)$.
For a real-valued function $f:\mathcal{M} \rightarrow \mathbb{R}$, its supremum norm is defined by $\|f\|_\infty:=\sup_{x\in \mathcal{M}}|f(x)|$. The notation $C^1(\mathbb{R})$ denotes the class of continuously differentiable real-valued functions on $\mathbb{R}$. For two sequences $\{a_n\}_{n \geq 1}$ and $\{b_n\}_{n \geq 1}$,  we write $a_n=\mathcal{O}(b_n)$ if there exists constant $R>0$ such that $|\frac{a_n}{b_n}|<R$, for all $n\ge N_0(R)\in \mathbb{N}$. For a sequence of random variables $\{X_n\}_{n\geq 1}$ and and a random variable $X$, $X_n\xrightarrow{P}X$ denotes convergence in probability; that is, for every $\varepsilon>0$, $Pr(d(X_n,X)>\varepsilon)\rightarrow 0$ as $n\rightarrow \infty$. Furthermore, for a sequence of positive constants
$\{a_n\}_{n \geq 1}$, $X_n=\mathcal{O}_{P}(a_n)$ indicates that for every $\varepsilon>0$, there exist constants $R_{\varepsilon}>0$ and $N_{\varepsilon}\in \mathbb{N}$ such that $Pr\left(\left|\frac{X_n}{a_n}\right|\geq R_\varepsilon\right)<\varepsilon$ for all $n \geq N_{\varepsilon}$. For $X\sim \mathbb{P}$ and $x\in \mathcal{M}$, let $\mathbb{S}_x(X):=[d(x,X)]^2$, with distribution function $F_x$ and quantile function $F_x^{-1}$. The mass parameter is denoted by $m\in(0,1)$. Under Huber's contamination model, $\mathbb{P}_s$, $\mathbb{P}_c$, and $\mathbb{P}_\alpha=(1-\alpha)~\mathbb{P}_s+\alpha~ \mathbb{P}_c$ denote the signal, contaminating, and contaminated distributions, respectively, where $\alpha \in[0,1/2)$ is the contamination proportion. The persistence diagram in homological dimension $p$ is denoted by $\mathrm{Dgm}_p(\cdot)$, and $d_B$ denotes the bottleneck distance between persistence diagrams.

\section{Mathematical Framework and Preliminaries}\label{section 2}
This section introduces the mathematical framework for robust topological learning under sparse contamination through distance functionals. The metric and probabilistic framework is described in \hyperref[section 2.1]{\textcolor{blue}{Section \ref{section 2.1}}}, followed by filtrations and persistent homology in \hyperref[section 2.2]{\textcolor{blue}{Section \ref{section 2.2}}}. \hyperref[section 2.3]{\textcolor{blue}{Section \ref{section 2.3}}} introduces the population distance-to-measure. The sampling and contamination model is described in \hyperref[section 2.4]{\textcolor{blue}{Section \ref{section 2.4}}}, while the empirical distance-to-measure construction 
in \hyperref[section 2.5]{\textcolor{blue}{Section \ref{section 2.5}}}.
\subsection{Metric and Probabilistic Framework} \label{section 2.1}
Let $\mathcal{M}\subset \mathbb{R}^{\ell}$ $(l\geq 1)$ be a compact metric space equipped with a distance function $d:\mathcal{M}\times\mathcal{M}\rightarrow [0,\infty)$, and let its diameter be $\text{diam}(\mathcal{M})= \sup_{x,y\in \mathcal{M}} d(x,y) := D <\infty$. 
For a Borel probability measure $\mathbb{P}$ on $\mathcal{M}$, let $B_\varepsilon(x):= \{y \in \mathcal{M}: d(x, y)< \varepsilon \}$ denotes the open ball of radius $ \varepsilon $ centered at $x$, and let $\mathcal{S}_{\mathbb{P}}$
= $\{x \in \mathcal{M}: \mathbb{P}(B_\varepsilon(x))>0, ~ \text{for all}~ \varepsilon>0\}$ denote the support of $\mathbb{P}$. To ensure geometric regularity of the support and to facilitate stable topological inference,
 we restrict attention to the probability measures satisfying the following condition.
 
\textbf{$(a,b)$-standard measure:}
A Borel probability measure $\mathbb{P}$ on $\mathcal{M}$ is said to be an $(a,b)$-standard measure,
and belongs to the class $\mathscr{P}_{a,b}(\mathcal{M})$, if there exist constants $a, b > 0$, and $\varepsilon_0 \in (0,1)$, such that for all $x \in \mathcal{S}_{\mathbb{P}}$ and for all $0 < \varepsilon < \varepsilon_0$,
\begin{equation}\label{eq:2.1}
    \mathbb{P}(B_\varepsilon(x)) \geq \min(1, a \varepsilon^b).
\end{equation}
It is further assumed that $\mathcal{S}_{\mathbb{P}}$ is compact and has positive reach 
in the sense of \textcite{federer1959curvature} (details on reach are provided in \hyperref[Appendix]{\textcolor{blue}{Appendix}}). This condition ensures that $\mathcal{S}_{\mathbb{P}}$ is geometrically well behaved for topological inference by precluding self-intersections and sharp cusps.

Moreover, the uniform lower mass-bound condition imposed by assumption (A2)
prevents the probability measure from being arbitrarily sparse on its support. It underlies uniform convergence in support estimation and the stability of topological descriptors, including those arising in persistent homology. Such assumptions are standard in nonparametric support estimation and geometric inference. See, for instance, the works of  \textcite{tsybakov1997nonparametric}, \textcite{cuevas2007nonparametric}, \textcite{genovese2012minimax}, \textcite{chazal2011geometric}, \textcite{aamari2019estimating}.


The construction of both the classical DTM and the proposed Huberized DTM are based on the quantiles of the distance distribution.
Accordingly, for $x\in \mathcal{M}$, $\mathbb{P}\in \mathscr{P}_{a,b}(\mathcal{M})$ and $m\in(0,1)$, define the $m$-th quantile of $\mathbb{P}$ as
\begin{equation}\label{eq:2.2}
    q_{m}(x):=q_{\mathbb{P},m}(x)=\inf\left\{\varepsilon \geq 0 : \mathbb{P}(B_\varepsilon(x)) \geq m\right\}.
\end{equation}
 Moreover, observe that as $\mathcal{M}$ is compact with diameter $D$, we have $q_{m}(x)\leq D$ for all $x\in \mathcal{M}$ and $m\in (0,1)$. Further, for a fixed $m\in (0,1)$ satisfying $m< a\varepsilon_0^b$, 
 choosing $\varepsilon = \left(\frac{m}{a}\right)^{\frac{1}{b}}\leq \varepsilon_0$ in condition \eqref{eq:2.1} yields $\mathbb{P}\left(B_{\left(m/a\right)^{\frac{1}{b}}}(x)\right) \geq m$, for all $x\in \mathcal{S}_{\mathbb{P}}$, and hence, $q_{m}(x) \leq \left(\frac{m}{a}\right)^{\frac{1}{b}}$. Consequently, the $m$-th distance quantile is uniformly bounded over $\mathcal{S}_{\mathbb{P}}$ by $\left(\frac{m}{a}\right)^{\frac{1}{b}}$.

\subsection{Filtrations and Persistent Homology}\label{section 2.2}
Topological inference in our framework is carried out through filtrations generated by distance-based functionals.
We briefly recall the notions of filtrations, persistent homology, and persistence diagrams, which provide the mathematical connection between estimation of the filtration-generating function and multiscale topological summaries.
For comprehensive treatments, one can see \textcite{edelsbrunner2002topological, zomorodian2004computing}, \textcite{  chazal2016structure}.
\paragraph{Filtration:}
A filtration of $\mathbb{R}^{\ell}$ is a nested family of subsets $\{\mathcal{X}_t\}_{t\in \mathbb{R}}$, and ~$\mathcal{X}_t\subseteq \mathcal{M}$ satisfying $\mathcal{X}_s\subseteq \mathcal{X}_t$ for all $s\leq t$. Filtrations are typically generated through the sublevel sets of a real-valued function $f: \mathbb{R}^{\ell} \rightarrow \mathbb{R}$, defined by $\mathcal{X}_t(f)=\{x\in \mathbb{R}^{\ell}:f(x)\leq t\}$, $t \in \mathbb{R}$. Thus, once $f$ is specified, the filtration and every persistence summary derived from it are determined. In many applications, different estimators of $f$ can lead to substantially different topological conclusions under contamination. Therefore, from a statistical point of view, the construction of $f$ becomes an inferential problem. Since persistence diagrams are induced by the filtration generated by $f$, any instability or contamination in the estimator of $f$ can directly affect the resulting topological summaries. In this sense, the statistical problem of topological learning reduces to the estimation of a robust filtration-generating distance functional.

\paragraph{Persistent Homology:} Persistent homology provides a mathematical framework for quantifying the evolution of the topology across the filtration $\{\mathcal{X}_t(f) \}_{t\in \mathbb{R}}$.
As the filtration parameter $t$ increases, the topology of the sublevel sets $\{\mathcal{X}_t(f) \}_{t\in \mathbb{R}}$ evolves across scales through the birth and death of connected components, loops, voids, and higher-dimensional topological features.
Persistent homology captures this evolution through the homology groups $\mathbb{H}_p \left(\mathcal{X}_t(f), \mathbb{F} \right)$, where $\mathbb{F}$ is a coefficient field and $p$ denotes the homological dimension. For each $t$, the homology group $\mathbb{H}_p\left(\mathcal{X}_t(f), \mathbb{F} \right)$ encode $p$-dimensional topological features present at scale $t$. Since the filtration is nested, the inclusions $\mathcal{X}_s(f) \hookrightarrow \mathcal{X}_t(f)$ for $s\leq t$ induce linear maps 
\begin{equation*}   \phi_{s,t}:\mathbb{H}_p\left(\mathcal{X}_s(f), \mathbb{F} \right) \rightarrow \mathbb{H}_p\left(\mathcal{X}_t(f), \mathbb{F} \right),
\end{equation*}
which track the evolution of topological features across filtration scales. A $p$-dimensional topological feature $ x_{b,d}^{(p)} \in \mathbb{H}_p\left(\mathcal{X}_b(f), \mathbb{F}\right)$ is said to born at scale $b$ if $x_{b,d}^{(p)} \notin \mathbb{H}_p\left(\mathcal{X}_{b-\delta}(f), \mathbb{F}\right)$ for any $\delta >0$, and said to die at scale $d>b$ if $x_{b,d}^{(p)} \in \mathbb{H}_p\left(\mathcal{X}_{d-\delta}(f), \mathbb{F}\right)$ for all $\delta >0$ satisfying $b\leq d-\delta_2$ and $x_{b,d}^{(p)} \notin \mathbb{H}_p\left(\mathcal{X}_{d}(f), \mathbb{F}\right)$. A feature will also die if it merges with a feature born earlier in the filtration. The quantity $(d-b)>0$ is the persistence of the feature, which measures the range of scales over which it survives.

\paragraph{Persistence Diagram:} The $p$-dimensional persistence diagram associated with a filtration-generating function $f$ is the multiset $\mathrm{Dgm}_p(f)=\left\{(b_{i}, d_{i}):b_{i}< d_{i}\right\} \subset \bar{\mathbb{R}}^2$, augmented by the diagonal $\Delta = \big\{(u,u): u\in \mathbb{R} \big\}$ with infinite multiplicity, where each point $(b_{i}, d_{i})$ corresponds to the birth and death scales of the $i$-th $p$-dimensional feature.

\noindent The persistence diagram provides a multiscale summary of the homological structure generated by the filtration $\{\mathcal{X}_t(f) \}_{t\in \mathbb{R}}$. Consequently, when the filtration-generating function is estimated from data, the resulting persistence diagram becomes a random topological summary induced by the underlying probability distribution. To ensure that persistence diagrams to be well-defined, $f$ is assumed to be tame, that is,
the homology groups $\mathbb{H}_p\left(\mathcal{X}_t(f), \mathbb{F} \right)$ are finite-dimensional for all $t$ and $p$, and may change only finitely many times as $t$ varies. The bottleneck distance between two persistence diagrams $\mathrm{Dgm}_{p}(f)$ and $\mathrm{Dgm}_{p}(g)$ associated with tame functions $f,g:\mathbb{R}^{\ell} \rightarrow \mathbb{R}$, is defined as 
\begin{equation*}
W_{\infty}\Big(\mathrm{Dgm}_{p}(f),\mathrm{Dgm}_{p}(g)\Big)=\inf_{\gamma_{p}} \sup_{x\in \mathrm{Dgm}_{p}(f)} \|x-\gamma_{p}(x)\|_{\infty} ,
\end{equation*}
where $\gamma_{p}:\mathrm{Dgm}_{p}(f) \rightarrow ~\mathrm{Dgm}_{p}(g)$ is a bijection. The relationship between perturbations of $f$ and perturbations of $\mathrm{Dgm}_{p}(f)$ is formalized by the following stability theorem, which underpins the topological inference results developed in this article.

\begin{theorem}[Stability of Persistence Diagrams; \textcite{cohen2005stability}]\label{theorem 2.1}
Let $f,g:\mathbb{R}^{\ell} \rightarrow \mathbb{R}$ be tame functions. Then for all homological dimension
$p\geq 0$,
\begin{equation}
    W_{\infty}\Big(\mathrm{Dgm}_{p}(f),\mathrm{Dgm}_{p}(g)\Big)\leq \|f-g\|_{\infty}.
\end{equation}
\end{theorem}

As an immediate consequence of \hyperref[theorem 2.1]{\textcolor{blue}{Theorem \ref{theorem 2.1}}}, if any estimator $\widehat{f_n}$ satisfies $\|\widehat{f_n}-f\|_{\infty}\leq \zeta_n$ with probability at least $(1-\eta)$, where $\eta\in (0, 1)$, then $W_{\infty}\Big(\mathrm{Dgm}_{p}(\widehat{f_n}),\mathrm{Dgm}_{p}(f)\Big)\leq \zeta_n$ with the same probability. In other words, a bound on the sup-norm difference between the estimated and population filtration-generating function provides an identical bound on the bottleneck distance between the corresponding persistence diagrams.
Thus, robust topological inference reduces to robust estimation of the underlying filtration-generating distance functional.
\subsubsection{Assumptions}\label{Assumptions}
 The following regularity conditions are imposed on the probability measure $\mathbb{P}$, and the ambient metric space $(\mathcal{M},d)$ in order to establish the technical results.
 \begin{itemize}
\item[(A1)] \textbf{Geometric regularity:} The space $(\mathcal{M},d)$ is a compact metric space. Furthermore, the support $\mathcal{S}_{\mathbb{P}}\subseteq \mathcal{M}$ of the measure $\mathbb{P}$ is compact with reach($\mathcal{S}_{\mathbb{P}}$) $\geq \rho >0$. 
\label{(A1)}
\begin{remark}
The compactness assumption prevents the support from being unbounded; otherwise, data could be spread out infinitely, making it difficult to infer meaningful topological features. 
The positive reach condition ensures the support has a regular geometric structure with bounded curvature and excludes geometric degeneracies such as cusps and self-intersections.
\end{remark}

\item[(A2)] \textbf{$(a,b)$-standard measure:} The probability measure $\mathbb{P}$ is an $(a,b)$-standard measure (for details see, \eqref{eq:2.1} and related discussions).
\label{(A2)}

\begin{remark}
     This assumption provides a uniform lower bound on the local mass around any point in the support, which enforces a uniform non-degeneracy of mass at small scales. This assumption is standard and widely adopted in the theory of support estimation.
\end{remark}

\item[(A3)] \textbf{Absolute continuity with bounded density:} The probability measure $\mathbb{P}$ is absolutely continuous with respect to the Lebesgue measure $\lambda_{\ell}$ on $\mathbb{R}^{\ell}$ with 
density $f$ satisfying $0< f_{\inf} \leq f(x) \leq f_{\sup} <\infty $ for all $x \in \mathcal{S}_{\mathbb{P}}$.
\begin{remark}
These density bounds ensure that the underlying distribution is neither too sparse nor too concentrated on its support. This condition is essential for the consistency and stability results developed later.
\end{remark}
\item[(A4)] \textbf{Uniform density bounds for the squared-distance distribution:} 
$\mathbb{P}\circ\mathbb{S}^{-1}_{x}$ is absolutely continuous with respect to Lebesgue measure $\lambda$ on $\mathbb{R}_+$ with density
$h_{\mathbb{S}_x}:=\frac{d\mu_x}{d\lambda}$ satisfying $0<h_{\inf}\leq h_{\mathbb{S}_x}\leq h_{\sup}<\infty$, $\lambda$-almost everywhere on $\mathcal{S}_{\mu_x} \subseteq [0, (\text{diam}(\mathcal{M}))^2]$, where $h_{\inf},~h_{\sup}$ are constants independent of $x$. Here, for each $x \in \mathcal{M}$, $\mathbb{P}\circ\mathbb{S}^{-1}_{x}$ is the law of $\mathbb{S}_x(X):=[d(x,X)]^2,~ X\sim \mathbb{P}$ with distribution function $F_x(t):=\mathbb{P}\circ\mathbb{S}^{-1}_{x}((-\infty,t])$.
\begin{remark}
This assumption provides uniform regularity of the squared-distance distribution, which underlies the continuity, stability, and robustness theory developed throughout the article.
\end{remark}
\end{itemize}

\subsection{Population Distance-to-measure} \label{section 2.3}
The distance-to-measure (DTM) of \textcite{chazal2011geometric} is a filtration-generating distance functional that replaces the distance-to-support with a distributionally averaged distance. Since the proposed Huberized DTM is constructed as a robust modification of the classical DTM, we first recall its population formulation together with the regularity properties that underpin the subsequent theoretical developments.

Here, note that under Assumption (A4), $\mathbb{S}_x(X)$ admits a density. Consequently, 
$F_x$ is continuous and $F_x(F_{x}^{-1}(m))=m$, where $F_{x}^{-1}(u):=\inf\{t\geq 0:F_x(t)\geq u\}$, $u\in(0,1)$ is the quantile function. Using this fact, the population distance-to-measure $\delta_{\mathbb{P},m}:\mathcal{M} \rightarrow [0,\infty)$ with parameter $m \in (0,1)$ is defined as,
\begin{equation}\label{eq:2.4}
    \delta_{\mathbb{P},m}(x)=\left(\frac{1}{m}\int\limits_{0}^{m}F_x^{-1}(u)~du\right)^{\frac{1}{2}}.
\end{equation}

\noindent For each $ t\geq 0$, define the sublevel-sets
  $\mathcal{M}_t=\left\{x \in \mathcal{M}:\delta_{\mathbb{P},m}(x) \leq t \right\}$.
Since $\mathcal{M}_{t_i}\subseteq \mathcal{M}_{t_j}$ whenever $t_i\leq t_j$, the collection $\{\mathcal{M}_t\}_{t\geq0}$ forms a filtration of $\mathcal{M}$.  Under some regularity conditions on $\mathbb{P}$ (see \cite{anai2020dtm}), the associated persistent homology recovers the homological structure of $\mathcal{S}_{\mathbb{P}}$. Now, we study a few regularity conditions, which enables us to investigate the theoretical properties of Huberized DTM.
\begin{proposition}\label{proposition 2.2}
  Let $\mathcal{M}$ be a compact metric space with diameter $D<\infty$, and define $\mathbb{S}_x(X):=[d(x,X)]^2,~ X \sim \mathbb{P}$. Then
  for all $x,y\in\mathcal{M}$,
\begin{equation*}
\mathbb{P}\bigl(|\mathbb{S}_x(X)-\mathbb{S}_y(X)| \leq 2D~d(x,y)\bigr)=1.
\end{equation*}
In particular, the process $\{\mathbb{S}_x(X):x\in \mathcal{M}\}$ has $2D$- Lipschitz sample paths with respect to $d$.
\end{proposition}
\noindent \textbf{Proof:} See the \hyperref[Proof of proposition 2.2]{\textcolor{blue}{Proof of proposition 2.1}} in \hyperref[Appendix A]{\textcolor{blue} Appendix A}.

\hfill $\blacksquare$

\begin{proposition}\label{proposition 2.3}
 Let $\mathcal{M}$ be a compact metric space with diam($\mathcal{M}$)$=D<\infty$. Then, for every $m\in(0,1)$, and for all $x,y\in\mathcal{M}$,
 \begin{equation*}
     |F_x^{-1}(m)-F_y^{-1}(m)| \leq 2D~d(x,y).
 \end{equation*}
 In particular, the map $x\in \mathcal{M}$ is $2D$-Lipschitz on $(\mathcal{M},d)$, uniformly over $m\in (0,1)$.
\end{proposition}
\noindent \textbf{Proof:} See the \hyperref[Proof of proposition 2.3]{\textcolor{blue}{Proof of proposition 2.2}} in \hyperref[Appendix A]{\textcolor{blue} Appendix A}.

\hfill $\blacksquare$




\begin{remark}\label{remark 5}
Proposition \ref{proposition 2.3} asserts that for every $m \in (0,1)$, the population DTM $\delta_{\mathbb{P},m}$ is 1-Lipschitz on $(\mathcal{M},d)$, that is, $|\delta_{\mathbb{P},m}(x)-\delta_{\mathbb{P},m}(y)|\leq d(x,y),~ \text{for all}~ x,y \in \mathcal{M}.$ This is a distribution-free property of the DTM functional, holds for any $\mathbb{P}$, and it implies that $\delta_{\mathbb{P},m}$ is uniformly continuous on $\mathcal{M}$.
\end{remark}
\begin{remark}\label{remark 6}
    Observe that \eqref{eq:2.4} implies that the population DTM $\delta_{\mathbb{P},m}(x)$ admits the representation
    \begin{equation*}
        \delta^2_{\mathbb{P},m}(x)=\mathbb{E}_{\mathbb{P}} \left[\mathbb{S}_x(X)\mid \mathbb{S}_x(X)\leq F_x^{-1}(m)\right],
    \end{equation*}
     or equivalently,
    \begin{equation*}
        \delta^2_{\mathbb{P},m}(x) = \arg\min_{\theta} \mathbb{E}_{\mathbb{P}} \left[\big(\mathbb{S}_x(X)-\theta\big)^2\mid \mathbb{S}_x(X)\leq F_x^{-1}(m)\right].
    \end{equation*}
        Moreover, as $\mathcal{M}$ is compact, $\mathbb{E}_{\mathbb{P}} \left[\big(\mathbb{S}_x(X)\big)^2\mid \mathbb{S}_x(X)\leq F_x^{-1}(m)\right]<\infty$.  Thus, the classical DTM can be interpreted as a conditional least-squares estimator based on the lower tail of the distribution of the distance function. This least-squares formulation motivates the Huberized version developed in \hyperref[section 3]{\textcolor{blue}{Section \ref{section 3}}}.
    \end{remark}
\subsection{ Contamination Model} \label{section 2.4}
Suppose that we observe a random sample $\mathscr{X} = \{X_{1},\cdots, X_{n}\}$ drawn independently from a probability measure $\mathbb{P} \in \mathscr{P}_{a,b}(\mathcal{M})$, where the class $\mathscr{P}_{a,b}(\mathcal{M})$ is defined in \hyperref[section 2.1]{\textcolor{blue}{Section \ref{section 2.1}}}. The empirical measure associated with this sample is defined as $ \mathbb{P}_{n} := \frac{1}{n} \sum_{i=1}^{n}\delta_{X_i}$, where $\delta_{x}$ denotes the Dirac measure at $x$, i.e., $\delta_{x}(A)= 1_{\{x\}}(A)$, for every Borel set $A\subseteq \mathcal{M}$, with support $\mathcal{S}_{\mathbb{P}_n}$.
The empirical DTM and its Huberized counterpart introduced later are constructed from the observed sample
and are used to infer the underlying topological structure of $\mathcal{S}_{\mathbb{P}}$.
%

In many practical situations, however, the observed data may contain contaminated observations. To model this situation, we consider Huber’s contamination model (\textcite{huber1964robust})
\begin{equation}
  \mathbb{P}_{\alpha} = (1 - \alpha)~ \mathbb{P}_{s} + \alpha~ \mathbb{P}_{c}, \quad \alpha \in \left[0,\frac{1}{2}\right),
\end{equation}
where $\mathbb{P}_{s} \in \mathscr{P}_{a,b}(\mathcal{M})$ denotes the underlying signal distribution of interest, and $\mathbb{P}_{c}$ is an arbitrary contaminating distribution. Under this model, the objective of robust topological inference is to recover the topology of 
$\mathcal{S}_{\mathbb{P}_{s}}$
rather than that of $\mathcal{S}_{\mathbb{P}_{\alpha}}$.
\subsection{Empirical Distance-to-measure} \label{section 2.5}
Based on the observed sample $\mathscr{X}=\{X_{1},\cdots, X_{n}\}$, let $\mathbb{P}_{n}$ denote the corresponding empirical measure. For each $x \in \mathcal{M}$, define the sample squared distances $\mathbb{S}_x(X_j):=[d(x,X_j)]^2$, for  $j=1,\cdots,n$. Let $\mathbb{S}_{x}(X_{(1)})\leq \mathbb{S}_{x}(X_{(2)})\leq \cdots \leq \mathbb{S}_{x}(X_{(n)})$ be the corresponding order statistics and $k=\lceil mn \rceil$. The empirical DTM is defined by
\begin{equation}
\widehat{\delta}_{\mathbb{P}_n,m}(x)=\left(\frac{1}{k}\sum_{j=1}^{k}\mathbb{S}_{x}(X_{(j)})\right)^{\frac{1}{2}}.
\end{equation}

Note that, squared of $\widehat{\delta}_{\mathbb{P}_n,m}(x)$ admits the following variational characterization. Define 
\begin{equation} \label{eq:2.6}
    Q_n(x;\theta)=\frac{1}{k}\sum_{j=1}^{k}\frac{1}{2}\left(\mathbb{S}_{x}(X_{(j)})- \theta\right)^2,\quad \theta \in \mathbb{R}.
\end{equation}
As $Q_n(x;\theta)$ is a strictly convex function with respect to $\theta$, it admits the unique minimizer
\begin{equation*}
    \widehat{\theta}_{\mathbb{P}_n,m}(x):=\arg\min_{\theta\in \mathbb{R}}~Q_n(x;\theta)=\frac{1}{k}\sum_{j=1}^{k}\mathbb{S}_{x}(X_{(j)})= \widehat{\delta}^2_{\mathbb{P}_n,m}(x).
\end{equation*}
Further, observe that $\widehat{\delta}_{\mathbb{P}_n,m}(x)$ can be viewed as the square root of the least-squares estimator based on the $k$-nearest-neighbour squared distances (see, \hyperref[remark 7]{\textcolor{blue} Remark \ref{remark 7}}).
\begin{remark}\label{remark 7}
    If $N_k(x)$ denotes the set of $k$-nearest neighbours of $x$ in $\mathcal{S}_{\mathbb{P}_n}$, then 
    \begin{equation*}
    \widehat{\delta}_{\mathbb{P}_n,m}^{2}(x)=\frac{1}{k}\sum_{X_j \in N_k(x)}\mathbb{S}_{x}(X_{j}).
\end{equation*}
\end{remark}
Under uncontaminated sampling, $\widehat{\delta}_{\mathbb{P}_n,m}$ is a consistent estimator of the population DTM
(see, \hyperref[section 2.3]{\textcolor{blue}{Section \ref{section 2.3}}} and Theorem 3.5, \cite{chazal2018robust}). However, because it is based on a quadratic loss, its underlying least-squares formulation has unbounded influence. Consequently, even a small number of contaminated in nearest neighbours may substantially distort the estimated distance functional.
This sensitivity motivates replacing the quadratic loss function in \eqref{eq:2.6} with a bounded-influence loss function, leading to Huberized DTM introduced in \hyperref[section 3]{\textcolor{blue}{Section \ref{section 3}}}.


\section{Huberized Distance-to-measure} \label{section 3}
In this section, we introduce a robust analogue of the population distance-to-measure defined in \hyperref[section 2.3]{\textcolor{blue} Section \ref{section 2.3}} by incorporating the Huber loss function.
\subsection{Huber Loss Function} \label{section 3.1}
Let $c>0$ be a tuning parameter. The Huber loss function $\rho_c:\mathbb{R}\rightarrow \mathbb{R}$ is defined by $\rho_c(u)=\frac{1}{2} u^2\cdot \mathbbm{1}_{\{|u| \leq c\}} + \left(c|u|-\frac{1}{2}c^2   \right) \mathbbm{1}_{\{|u| > c\}}$
with associated score function $\psi_c(u)=u\cdot \mathbbm{1}_{\{|u| \leq c\}} + c~\text{sign}(u)\cdot \mathbbm{1}_{\{|u| > c\}}$
The function $\rho_c$ is convex and continuously differentiable on $\mathbb{R}$, while $\psi_c$ is non-decreasing, 1-Lipschitz, and satisfies $\|\psi_c\|_{\infty}\leq c$.
As a result, $\psi_c$ truncates the contribution of extreme observations, a property that underlies the bounded-influence estimator introduced in the next subsection.
\subsection{Huberized DTM estimator} \label{section 3.2}
For $x \in \mathcal{M}$, define the empirical Huber objective function by
\begin{equation}\label{eq:3.3}
    L_{n}(x;\theta):= \frac{1}{k}\sum_{j=1}^{k}\rho_c\left(\mathbb{S}_{x}(X_{(j)})- \theta\right),\quad \theta\in \mathbb{R}.
\end{equation}
Observe that the objective function $L_{n}(x;\theta)$ is strictly convex in $\theta$, and $L_{n}(x;\theta)\rightarrow \infty$ as $|\theta|\rightarrow \infty$ almost surely, since $\rho_c\left(\mathbb{S}_{x}(X_{(j)})- \theta\right) \geq c|\theta|-\left(c~\mathbb{S}_{x}(X_{(j)}) +\frac{1}{2} c^2\right)$. Therefore, the minimum of $L_{n}(x;\theta)$ with respect to $\theta$ exists. Hence,
$\widehat{\theta}_{\mathbb{P}_n,m,c}(x)\in \displaystyle\arg \min_{\theta\in \mathbb{R}} L_{n}(x;\theta)$.
Since $\rho_c$ is continuously differentiable, any minimizer satisfies the estimating equation $\sum\limits_{j=1}^{k}\psi_c\left(\mathbb{S}_{x}(X_{(j)})- \theta\right)=0$.
\noindent Moreover, $L_{n}(x,\cdot)$ is decreasing on $\left(-\infty, \mathbb{S}_{x}(X_{(1)})\right)$ and increasing on $\left(\mathbb{S}_{x}(X_{(k)}),\infty \right)$, see \hyperref[Lemma 3]{\textcolor{blue} Lemma 3}.
Consequently, every minimizer of $L_{n}(x,\cdot)$ lie in the interval $[\mathbb{S}_{x}(X_{(1)}),\mathbb{S}_{x}(X_{(k)})] \subset [0,\infty)$. In Particular $\widehat{\theta}_{\mathbb{P}_n,m,c}(x) \geq 0$, as $\mathbb{S}_{x}(X_{(1)})\geq 0$. The empirical Huberized DTM at $x \in \mathcal{M}$ is therefore defined as
\begin{equation}\label{eq:3.5}
    \widehat{\delta}_{\mathbb{P}_n,m,c}(x)=\sqrt{\widehat{\theta}_{\mathbb{P}_n,m,c}(x)}.
\end{equation}
Similar to classical empirical DTM $\widehat{\delta}_{\mathbb{P}_n,m}(x)$, this estimator $\widehat{\delta}_{\mathbb{P}_n,m,c}(x)$ 
depends on the sample only through the first $k$ order statistics $\mathbb{S}_{x}(X_{(1)}),\cdots,\mathbb{S}_{x}(X_{(k)})$. In particular, observations beyond the $k$-th nearest neighbour of $x$ do not contribute to $\widehat{\delta}_{\mathbb{P}_n,m,c}(x)$.
\subsection{Limiting regimes and the Robustness-Efficiency Tradeoff} \label{section 3.3}
The Huberized DTM  bridges between classical estimators for varying choice of $c$.

\vspace{0.1 in}

\noindent \textbf{1. Mean-type DTM:} As $c\rightarrow \infty$, $\psi_c(u)\rightarrow u$ pointwise and, the estimating equation
    reduces to $\sum_{j=1}^{k}\left(\mathbb{S}_{x}(X_{(j)})- \theta\right)=0$, whose unique solution is $\frac{1}{k}\sum_{j=1}^{k}\mathbb{S}_{x}(X_{(j)})$. Hence, $\widehat{\delta}_{\mathbb{P}_n,m,c}(x) \rightarrow \widehat{\delta}_{\mathbb{P}_n,m}(x)$, as $c\rightarrow \infty$. For instance, one can see \eqref{eq:2.6}.

    \vspace{0.1 in}
    
\noindent \textbf{2. Median-type behaviour:} As $c\rightarrow0$, $\psi_c(u)\rightarrow c~\text{sign}(u)$, and the estimating equation reduces to
$\sum\limits_{j=1}^{k}\text{sign}\left(\mathbb{S}_{x}(X_{(j)})- \theta\right)=0$, which characterizes the minimizers of $\sum\limits_{j=1}^{k}\left|\mathbb{S}_{x}(X_{(j)})- \theta\right|=0$. In this regime, $\widehat{\theta}_{\mathbb{P}_n,m,c}(x)$ converges to a median-type estimator.

\subsection{Population Huberized DTM}
The variational characterization of the population DTM in \hyperref[remark 6]{\textcolor{blue} Remark \ref{remark 6}} leads to the following Huberized analogue based on Huber loss. Accordingly, for $\mathbb{P} \in \mathscr{P}_{a,b}(\mathcal{M})$, $m \in (0,1)$ and $c>0$, define the population Huber risk
\begin{equation}\label{eq:3.50}
    L(x;\theta):=\mathbb{E}\left[\rho_c \left(\mathbb{S}_x(X)-\theta\right)\mid \mathbb{S}_x(X)\leq F_x^{-1}(m)\right],
\end{equation}
$\theta\in \mathbb{R}$, where $\mathbb{S}_x(X)$ and $F_x^{-1}(m)$ are as introduced in \hyperref[section 2.3]{\textcolor{blue} Section \ref{section 2.3}}. The corresponding population Huber M-functional is
\begin{equation} \label{eq:3.6}
\theta_{\mathbb{P},m,c}(x)  =\arg\min_{\theta\in \mathbb{R}} L(x;\theta). 
\end{equation} 
\noindent The following lemma justifies differentiation under the conditional expectation in (\ref{eq:3.6}).
\begin{lemma} \label{Lemma 3.1}
    Let $Z$ be a real-valued random variable satisfying $\mathbb{E}(|Z|)<\infty$, and let $q\in \mathbb{R}$ be such that $\Pr(Z\leq q)>0$. For $\theta \in \mathbb{R}$, define $\mathcal{L}(\theta)=\mathbb{E}\left[\rho_c\left(Z-\theta\right)\mid Z\leq q\right]$. 
    Then, $\mathcal{L}$ is continuously differentiable on $\mathbb{R}$,
    and its derivative is given by $\mathcal{L}'(\theta)=-\mathbb{E}\left[\psi_c\left(Z-\theta\right)\mid Z\leq q\right]$, where $\rho_c$ and $\psi_c$ are as introduced in \hyperref[section 3.1]{\textcolor{blue} Section \ref{section 3.1}}.
\end{lemma}
\noindent \textbf{Proof:} See the \hyperref[Proof of Lemma 3.1]{\textcolor{blue}{Proof of Lemma 3.1}} in \hyperref[Appendix A]{\textcolor{blue} Appendix A}.

 \hfill$\blacksquare$

\hyperref[Lemma 3.1]{\textcolor{blue} Lemma \ref{Lemma 3.1}} ensures that $L(x;\theta) \in C^1(\mathbb{R})$ with respect to $\theta$.
Consequently, any minimizer $\theta_{\mathbb{P},m,c}(x)$
 satisfies
\begin{equation} \Phi_{\mathbb{P},x}(\theta):=\mathbb{E}_{\mathbb{P}}\left[\psi_c\left(\mathbb{S}_x(X)-\theta\right)\mid \mathbb{S}_x(X)\leq F_x^{-1}(m)\right]=0.
\end{equation}
The following two propositions establish existence, uniqueness, and admissible bounds for the population Huber functional. For brevity, throughout this article $\Phi_{x}(\theta)$ denotes $\Phi_{\mathbb{P},x}(\theta)$.

\begin{proposition}[Existence and uniqueness]\label{proposition 3.1}
Suppose Assumption (A4) holds. Then, for every $x \in \mathcal{M}$ the estimating equation $\Phi_x(\theta)=0$ admits a unique solution $\theta_{\mathbb{P},m,c}(x) \in \mathbb{R}$.
\end{proposition}
\noindent \textbf{Proof:} See the \hyperref[Proof of proposition 3.1]{\textcolor{blue}{Proof of proposition 3.1}} in \hyperref[Appendix A]{\textcolor{blue} Appendix A}.

\hfill $\blacksquare$

\begin{proposition}[Admissible Bounds]\label{proposition 3.2}
Under Assumption (A2)-(A3), for every $x \in \mathcal{M}$ and $m\in (0,1)$, $\theta_{\mathbb{P},m,c}(x)\in(0,~ F_x^{-1}(m)]\subseteq \left(0,\left(\frac{m}{a}\right)^{\frac{2}{b}} \right]$, where $\mathbb{S}_x(X):=[d(x,X)]^2,~X\sim \mathbb{P}$ and $F_x^{-1}(m):=\inf\{t\geq 0: F_x(t)\geq m\}$ is the $m$-th quantile of $F_x(t):=\mathbb{P}(\mathbb{S}_x(X)\leq t)$.
\end{proposition}
\noindent \textbf{Proof:} See the \hyperref[Proof of proposition 3.2]{\textcolor{blue}{Proof of proposition 3.2}} in \hyperref[Appendix A]{\textcolor{blue} Appendix A}.

\hfill $\blacksquare$

\hyperref[proposition 3.2]{\textcolor{blue} Proposition \ref{proposition 3.2}} asserts that $\theta_{\mathbb{P},m,c}(x)$ is strictly positive and satisfies $\theta_{\mathbb{P},m,c}(x) \leq F_x^{-1}(m)$.
The corresponding population Huberized DTM is therefore defined as 
\begin{equation} \label{eq:3.8}
    \delta_{\mathbb{P},m,c}(x)=\sqrt{\theta_{\mathbb{P},m,c}(x)}.
\end{equation}
As an immediate consequence of \hyperref[proposition 3.1]{\textcolor{blue} Proposition \ref{proposition 3.1}} and \hyperref[proposition 3.2]{\textcolor{blue} Proposition \ref{proposition 3.2}}, $\delta_{\mathbb{P},m,c}(x)$ is well defined for all $x\in \mathcal{M}$ and satisfies $\delta_{\mathbb{P},m,c}(x)\in \left(0,\left(\frac{m}{a}\right)^{\frac{1}{b}} \right]$.
\begin{remark}
   Assumption (A4) in \hyperref[proposition 3.1]{\textcolor{blue} Proposition \ref{proposition 3.1}} excludes 
   degenerate conditional distributions of $\mathbb{S}_x(X)=[d(x,X)]^2$. However, uniqueness of the $\theta_{\mathbb{P},m,c}(x)$ (see, \eqref{eq:3.6}) can also be established in this degenerate case. More precisely, suppose, for some $x\in \mathcal{M}$ and $s_0\geq 0$,
   \begin{equation*}
       Pr\left(\mathbb{S}_x(X)=s_0 \mid \mathbb{S}_x(X)\leq F_x^{-1}(m)\right)=1.
   \end{equation*}
   Then $L(x;\theta)$ in \eqref{eq:3.50} reduces to $L(x;\theta)=\rho_c(s_0-\theta),~\theta \in \mathbb{R}$, which is uniquely minimized at $\theta=s_0$. Consequently, $\theta_{\mathbb{P},m,c}(x)=s_0$,  $\delta_{\mathbb{P},m,c}(x)=\sqrt{s_0}$.
   For example, if $\mathbb{P}$ is the uniform distribution on a sphere of radius $a>0$ centered at $x$, then $\mathbb{S}_x(X)=a^2$ almost surely, which is independent of the dimension of the sphere. Hence, for every $m\in (0,1)$, $\theta_{\mathbb{P},m,c}(x)=a^2$, and therefore,  $\delta_{\mathbb{P},m,c}(x)=a$. Thus, the $\theta_{\mathbb{P},m,c}(x)$ remains uniquely defined irrespective of whether the distribution of $\mathbb{S}_x(X)$ is degenerate.
\end{remark}

\subsubsection{Properties of \texorpdfstring{$\delta_{\mathbb{P},m,c}$}{delta\_{P,m,c}(x)}}\label{section 3.4.1}
The $\delta_{\mathbb{P},m,c}(.)$ serves as the filtration function underlying the proposed persistent homology framework. Its statistical and topological properties therefore
depend on the regularity of the map $x \mapsto \delta_{\mathbb{P},m,c}(x)$. In this subsection, we establish continuity and Lipschitz bound for $\delta_{\mathbb{P},m,c}$, which underpin the statistical consistency and topological stability results derived in \hyperref[section 4]{\textcolor{blue}{Section \ref{section 4}}}.
\begin{theorem}\label{Theorem 3.1}

Under Assumptions (A3) and (A4), for every $m\in (0,1)$ and $c>0$, 
$\delta_{\mathbb{P},m,c}(.)$ defined in \eqref{eq:3.8} is continuous on $(\mathcal{M},d)$.
\end{theorem}
\noindent \textbf{Proof:} See the \hyperref[Proof of Theorem 3.1]{\textcolor{blue}{Proof of Theorem 3.1}} in \hyperref[Appendix B]{\textcolor{blue} Appendix B}.

\hfill $\blacksquare$

\begin{remark}
     \hyperref[Theorem 3.1]{\textcolor{blue} Theorem \ref{Theorem 3.1}} together with the compactness of $\mathcal{M}$ ensures that the map $x \mapsto \delta_{\mathbb{P},m,c}(x)$ is uniformly continuous on $\mathcal{M}$. Consequently, nearby points in $\mathcal{M}$ have uniformly similar Huberized DTM values. This uniform regularity enables us in establishing uniform consistency of the empirical Huberized DTM and the convergence of the associated persistence diagrams.
     \end{remark}

 We next establish that the population Huberized DTM is uniformly bounded away from zero over $\mathcal{M}$. This ensures a positive lower bound for the filtration-generating function 
throughout the domain.
\begin{proposition}\label{proposition 3.3} 
Let $$\theta_{\inf}:= \inf_{x\in\mathcal{M}} \theta_{\mathbb{P},m,c}(x),$$ 
where $\theta_{\mathbb{P},m,c}(x)$ is as defined in \eqref{eq:3.6}. Then, under Assumptions (A1), (A3) and (A4), $\theta_{\inf}>0$, and consequently, $\displaystyle\inf_{x\in \mathcal{M}}\delta_{\mathbb{P},m,c}(x)>0$.
\end{proposition}

\noindent \textbf{Proof:} See the \hyperref[Proof of proposition 3.3]{\textcolor{blue}{Proof of proposition 3.3}} in \hyperref[Appendix A]{\textcolor{blue} Appendix A}.

\hfill $\blacksquare$



To study stability of the persistence diagrams induced by the Huberized DTM filtration, the filtration function must be lipschitz continuous, which ensures that small spatial perturbations in $\mathcal{M}$ lead to small changes in the filtration values. We therefore establish Lipschitz continuity of $\delta_{\mathbb{P},m,c}$, which provides the regularity needed for the subsequent stability analysis of the associated persistence diagrams, following the stability framework of \textcite{cohen2005stability}.

\begin{theorem}
\label{Theorem 3.2}
Under Assumptions (A1)-(A4), for every $m\in (0,1)$ and $c>0$, the map $x\mapsto\delta_{\mathbb{P},m,c}(x)$ defined in \eqref{eq:3.8} is Lipschitz continuous on $(\mathcal{M},d)$; that is,
there exists a constant $L_*>0$ such that $$\left|\delta_{\mathbb{P},m,c}(x)-\delta_{\mathbb{P},m,c}(y)\right|\leq L_*~ d(x,y),$$ for all $x,y\in \mathcal{M}$.
\end{theorem}
\noindent \textbf{Proof:} See the \hyperref[Proof of Theorem 3.2]{\textcolor{blue}{Proof of Theorem 3.2}} in \hyperref[Appendix B]{\textcolor{blue} Appendix B}. 


\hfill $\blacksquare$

\begin{remark}
    As $\delta_{\mathbb{P},m,c}$ is Lipschitz continuous on the compact space $\mathcal{M}$, it has finitely many topological critical values in view of proposition 3.5, \textcite{buchet2016efficient}.
    Consequently, its sublevel sets exhibit finitely generated homology groups that change at only finitely many levels.
    This property ensures that $\delta_{\mathbb{P},m,c}$ is suitable for defining the weighted Vietoris–Rips filtration introduced in the following.
\end{remark}
\subsection{Huberized DTM Filtration}\label{section 3.5}
The regularity properties established in \hyperref[section 3.4.1]{\textcolor{blue}{Section \ref{section 3.4.1}}} motivate the use of $\delta_{\mathbb{P},m,c}$ as a weight function in the weighted Vietoris-Rips filtration, as used in \textcite{anai2020dtm}. Following the weighted Vietoris–Rips 
filtration framework, we now define the corresponding filtration. Let $\delta_{\mathbb{P},m,c}: \mathcal{M} \rightarrow [0,\infty) $ 
be the population Huberized DTM defined in \eqref{eq:3.8}.
For each $x \in \mathcal{M}$ and filtration parameter $t\geq 0$, define the scale-dependent radius 
\begin{align*}
    r_x(t)=
    \begin{cases}
-\infty, & t<\delta_{\mathbb{P},m,c}(x),\\[1ex]
\left(t^2-\delta_{\mathbb{P},m,c}^2(x)\right)^{1/2},
& t\ge\delta_{\mathbb{P},m,c}(x),
\end{cases}
\end{align*}
which determines the radius of the weighted ball associated with $x$ at scale $t$. The weighted Vietoris-Rips complex at scale $t$  is defined as
\begin{align}
\mathbb{VR}_{t}\Big( \mathcal{S}_{\mathbb{P}},\delta_{\mathbb{P},m,c} \Big) = \bigg\{ \sigma_{\kappa}^{(i)} =& [x_{i_0}, x_{i_1}, \cdots, x_{i_{\kappa}}]: \nonumber \{ x_{i_0}, x_{i_1}, \cdots, x_{i_{\kappa}}\} \subseteq \mathcal{S}_{\mathbb{P}}, \nonumber \\ & \qquad \qquad
d(x_{i_{j}},x_{i_{l}})\leq r_{x_{i_j}}(t)+ r_{x_{i_l}}(t),~ \forall j,l \in \{0,\cdots, \kappa\}\bigg\},
\end{align}
where $\sigma_{\kappa}^{(i)}$ denotes the $i$-th $\kappa$-simplex with vertex set $\{x_{i_0}, x_{i_1}, \cdots, x_{i_\kappa}\}\subseteq \mathcal{S}_{\mathbb{P}}$ for $i\in \{0,\cdots,N_\kappa\}$, and $N_\kappa$ denotes the total number of $\kappa$-simplices.
Since $r_x(t)$ is non-decreasing in $t$, every simplex present at scale $t$ remains present at all larger scales. Consequently,
\begin{equation*}
    \mathbb{VR}_{t_i}\Big( \mathcal{S}_{\mathbb{P}},\delta_{\mathbb{P},m,c} \Big)\subseteq \mathbb{VR}_{t_j}\Big( \mathcal{S}_{\mathbb{P}},\delta_{\mathbb{P},m,c} \Big),\quad t_i\leq t_j.
\end{equation*}
Hence, the collection  
\begin{equation}  \mathbb{VR}\Big(\mathcal{S}_{\mathbb{P}},\delta_{\mathbb{P},m,c} \Big):=\bigg\{\mathbb{VR}_{t}\Big( \mathcal{S}_{\mathbb{P}},\delta_{\mathbb{P},m,c} \Big)\bigg\}_{t\geq 0}
\end{equation}
forms a filtration, referred to as the Huberized DTM filtration. 

Analogously, substituting the population functional  $\delta_{\mathbb{P},m,c}$ with the empirical estimator $\widehat{\delta}_{\mathbb{P}_n,m,c}$ based on the sample $\mathscr{X}=\{X_1,\cdots,X_n\}$ yields the empirical Huberized DTM filtration 
\begin{equation}
    \widehat{\mathbb{VR}}\Big( \mathcal{S}_{\mathbb{P}_n},\delta_{\mathbb{P}_n,m,c} \Big):=\bigg\{\mathbb{VR}_{t}\Big( \mathcal{S}_{\mathbb{P}_n},\widehat{\delta}_{\mathbb{P}_n,m,c} \Big)\bigg\}_{t\geq 0}.
\end{equation}
Finally, for each homological dimension $p\ge0$, $\mathrm{Dgm}_p\big(\delta_{\mathbb{P},m,c}\big)$ and $\widehat{\mathrm{Dgm}}_p\big(\delta_{\mathbb{P}_n,m,c}\big)$ denote the persistence diagrams associated with the population and empirical Huberized DTM filtrations, respectively.

\section{Large Sample Properties} \label{section 4}
In this section, we study the asymptotic behaviour and robustness properties of the empirical Huberized DTM and its associated persistence diagrams. We first prove its uniform consistency over the metric space $\mathcal{M}$
in \hyperref[section 4.1]{\textcolor{blue}{Section \ref{section 4.1}}}. Thereafter, we combine this result with the stability theory of persistence diagrams to establish consistency of the associated persistence diagrams in \hyperref[section 4.2]{\textcolor{blue}{Section \ref{section 4.2}}}. 
 \subsection{Consistency of the Huberized DTM}\label{section 4.1}
 We first establish the uniform consistency of the empirical Huberized DTM over the entire metric space $\mathcal{M}$. This convergence result is subsequently combined with the stability theory of persistence diagrams to establish the consistency of the associated persistence diagrams.
\begin{theorem}\label{Theorem 4.1}
    Under Assumptions (A1)-(A4), for every $m\in (0,1)$ and $c>0$,
\begin{equation*}
    \sup_{x\in \mathcal{M}} \left| \widehat{\delta}_{\mathbb{P}_n,m,c}(x) - \delta_{\mathbb{P},m,c}(x)\right| \xrightarrow{P} 0,
\end{equation*}
as $n \rightarrow \infty$, where $\widehat{\delta}_{\mathbb{P}_n,m,c}(x)$ and $\delta_{\mathbb{P},m,c}(x)$ are same as defined in \eqref{eq:3.5} and \eqref{eq:3.8}, respectively.
\end{theorem}
\noindent \textbf{Proof:} See the \hyperref[Proof of Theorem 4.1]{\textcolor{blue}{Proof of Theorem 4.1}} in \hyperref[Appendix B]{\textcolor{blue} Appendix B}.

\hfill $\blacksquare$

\begin{remark}
    \hyperref[Theorem 4.1]{\textcolor{blue} Theorem \ref{Theorem 4.1}} yields the uniform convergence of the empirical Huberized DTM to its population counterpart. Under the weighted Vietoris–Rips framework of \hyperref[section 3.5]{\textcolor{blue} Section \ref{section 3.5}}, this uniform consistency allows us to apply the stability results. 
\end{remark}



\subsection{Huberized DTM based Persistence Diagrams} \label{section 4.2}
We now translate the uniform consistency of the empirical Huberized DTM established in \hyperref[Theorem 4.1]{\textcolor{blue} Theorem \ref{Theorem 4.1}} into consistency of the associated persistence diagrams.
This follows from the stability theory of weighted Vietoris–Rips filtrations developed by \textcite{anai2020dtm}.
 The stability of weighted Vietoris–Rips filtrations is quantified through the interleaving distance, denoted by $d_I$, which measures the discrepancy between two filtrations; see \textcite{oudot2015persistence}. Since the empirical and population filtrations differ through both the filtration-generating function and the underlying support, we recall the corresponding stability results.
 \begin{proposition}[Stability under perturbation of the filtration-generating function, \textcite{anai2020dtm}, Proposition 3] \label{proposition 4.1}
Let $\mathcal{M}$ be a compact metric space, and let $f,g:\mathcal{M}\rightarrow [0,\infty)$ be two weight functions. Then the weighted Vietoris-Rips filtrations generated by $f$ and $g$ satisfy
\begin{equation}
d_I\left(\mathbb{VR}(\mathcal{M},f), \mathbb{VR}(\mathcal{M},g) \right) \leq \|f-g\|_{\infty},
\end{equation}
where $d_I$ denotes the interleaving distance between filtrations.
\end{proposition}
\begin{proposition}[Stability under perturbation of the support, \textcite{anai2020dtm}, Proposition 4] \label{proposition 4.2}
Let $\mathcal{X}$ and $\mathcal{Y}$
be compact subsets of $\mathcal{M}$, and let $f:\mathcal{\mathcal{M}} \rightarrow [0,\infty)$ be a $L$-Lipschitz weight function. Then
\begin{equation}
d_I\left(\mathbb{VR}(\mathcal{X},f), \mathbb{VR}(\mathcal{Y},f) \right) \leq (1+L) ~d_H(\mathcal{X},\mathcal{Y}),
\end{equation}
where $d_H$ denotes the Hausdorff distance.
\end{proposition}
The following theorem is an consequence of \hyperref[Theorem 4.1]{\textcolor{blue}{Theorem \ref{Theorem 4.1}}} and the preceding stability results.

\begin{theorem}\label{Theorem 4.2}
Under Assumptions (A1)-(A4), for every homological dimension $p\geq0$,
\begin{equation}\label{eq:4.19}
        d_B\left( \widehat{\mathrm{Dgm}}_p\big(\delta_{\mathbb{P}_n,m,c}\big), \mathrm{Dgm}_p\left(\delta_{\mathbb{P},m,c}\right) \right) \leq \|\widehat{\delta}_{\mathbb{P}_n,m,c} - \delta_{\mathbb{P},m,c} \|_{\infty}+ (1+L_*)~ d_H\left( \mathcal{S}_{\mathbb{P}_n},\mathcal{S}_{\mathbb{P}}\right),
    \end{equation}
where $\widehat{\mathrm{Dgm}}_p\big(\delta_{\mathbb{P}_n,m,c}\big)$ and $\mathrm{Dgm}_p\left(\delta_{\mathbb{P},m,c}\right)$ denote the persistence diagrams in homological dimension $p$, associated with the weighted Vietoris–Rips filtrations induced by $\widehat{\delta}_{\mathbb{P}_n,m,c}$ and $\delta_{\mathbb{P},m,c}$, respectively, as defined in \hyperref[section 3.5]{\textcolor{blue} Section \ref{section 3.5}}; $L_*$ is s the Lipschitz constant of $\delta_{\mathbb{P},m,c}$ introduced in \hyperref[Theorem 3.2]{\textcolor{blue}{Theorem \ref{Theorem 3.2}}}; and 
$d_B$, $d_H$ denote the bottleneck and Hausdorff distances, respectively.
    \noindent Consequently,
    \begin{equation}\label{eq:4.20}
        d_B\left( \widehat{\mathrm{Dgm}}_p\big(\delta_{\mathbb{P}_n,m,c} \big), \mathrm{Dgm}_p\left(\delta_{\mathbb{P},m,c}\right)\right) \xrightarrow{P} 0, \quad n\rightarrow \infty.
    \end{equation}
\end{theorem}
\noindent\textbf{Proof:} Recall from \hyperref[section 3.5]{\textcolor{blue} Section \ref{section 3.5}} that $\mathbb{VR}\Big(\mathcal{S}_{\mathbb{P}},\delta_{\mathbb{P},m,c} \Big)$ and
$\widehat{\mathbb{VR}}\Big( \mathcal{S}_{\mathbb{P}_n},\delta_{\mathbb{P}_n,m,c} \Big)$ denote the weighted Vietoris–Rips filtrations induced by $\delta_{\mathbb{P},m,c}$ on $\mathcal{S}_{\mathbb{P}}$ and by $\widehat{\delta}_{\mathbb{P}_n,m,c}$ on $\mathcal{S}_{\mathbb{P}_n}$, respectively. Using the triangle inequality for the interleaving distance, we have
\small
\begin{align}
 d_I\left( \widehat{\mathbb{VR}}\Big( \mathcal{S}_{\mathbb{P}_n} ,\delta_{\mathbb{P}_n,m,c} \Big),  \mathbb{VR}\Big(\mathcal{S}_{\mathbb{P}},\delta_{\mathbb{P},m,c} \Big) \right)
&\leq d_I\left( \widehat{\mathbb{VR}}\Big(\mathcal{S}_{\mathbb{P}_n},\delta_{\mathbb{P}_n,m,c} \Big), \widehat{\mathbb{VR}}\Big(\mathcal{S}_{\mathbb{P}_n},\delta_{\mathbb{P},m,c} \Big) \right) \nonumber\\
&\quad + d_I\left(\widehat{\mathbb{VR}}\Big(\mathcal{S}_{\mathbb{P}_n},\delta_{\mathbb{P},m,c} \Big), \mathbb{VR}\Big(\mathcal{S}_{\mathbb{P}},\delta_{\mathbb{P},m,c} \Big) \right).
\end{align}
\normalsize
\hyperref[proposition 4.1]{\textcolor{blue}{Proposition \ref{proposition 4.1}}} and \hyperref[proposition 4.2]{\textcolor{blue}{Proposition \ref{proposition 4.2}}} yields
\begin{align}
    d_I\left( \widehat{\mathbb{VR}}\Big(\mathcal{S}_{\mathbb{P}_n},\delta_{\mathbb{P}_n,m,c} \Big),  \mathbb{VR}\Big(\mathcal{S}_{\mathbb{P}},\delta_{\mathbb{P},m,c} \Big) \right) \leq \|\widehat{\delta}_{\mathbb{P}_n,m,c} - \delta_{\mathbb{P},m,c} \|_{\infty} + (1+L_{*})~ d_H\left(\mathcal{S}_{\mathbb{P}_n},\mathcal{S}_{\mathbb{P}}\right),
\end{align}
where $L_{*}$ is the Lipschitz constant of $\delta_{\mathbb{P},m,c}$. The desired bottleneck bound follows  from the Algebraic Stability Theorem (see, \textcite{chazal2009proximity}, Theorem 4.4), which states that $ d_B \leq d_I$, thereby proving the inequality in \eqref{eq:4.19}. Finally, since
$ \left\| \widehat{\delta}_{P_n,m,c} - \delta_{P,m,c} \right\|_{\infty} \xrightarrow{P} 0 $ as $n\rightarrow \infty$ in view of the assertion in \hyperref[Theorem 4.1]{\textcolor{blue}{Theorem \ref{Theorem 4.1}}} and
$ d_H\left( \mathcal{S}_{\mathbb{P}_n}, \mathcal{S}_{\mathbb{P}} \right) \xrightarrow{P} 0$ as $n\rightarrow \infty$, the convergence in \eqref{eq:4.20} follows from Slutsky's theorem. 

\hfill $\blacksquare$

\hyperref[Theorem 4.2]{\textcolor{blue}{Theorem \ref{Theorem 4.2}}} shows that consistency of the filtration estimator implies consistency of the assocated persistence diagrams. Consequently, within this framework, robust topological inference reduces to the robust estimation of the underlying filtration function.

\begin{corollary}[Finite-Sample Bottleneck Bound]\label{corollary 4.1}
Under Assumptions (A1)-(A4), if 
$$\|\widehat{\delta}_{\mathbb{P}_n,m,c} - \delta_{\mathbb{P},m,c} \|_{\infty}=\mathcal{O}_{\mathbb{P}}(r_n),$$ and $$d_H\left(\mathcal{S}_{\mathbb{P}_n}, \mathcal{S}_{\mathbb{P}}\right)=\mathcal{O}_{\mathbb{P}}(s_n),$$
for some deterministic sequence $r_n, ~s_n \rightarrow 0$, then, for every homological dimension $p\geq0$,
\begin{equation*}
    d_B\left( \widehat{\mathrm{Dgm}}_k\big(\delta_{\mathbb{P}_n,m,c} \big), \mathrm{Dgm}_k\left(\delta_{\mathbb{P},m,c}\right)\right)=\mathcal{O}_{\mathbb{P}}(r_n+s_n).
\end{equation*}
\end{corollary}
\noindent\textbf{Proof:} The result follows from the inclusion
\small
\begin{align*}
    \left\{d_B\left( \widehat{\mathrm{Dgm}}_p\big(\delta_{\mathbb{P}_n,m,c} \big), \mathrm{Dgm}_p\left(\delta_{\mathbb{P},m,c}\right)\right)>r_n+s_n \right\}\subseteq & \left\{\|\widehat{\delta}_{\mathbb{P}_n,m,c} - \delta_{\mathbb{P},m,c} \|_{\infty}>\frac{r_n}{2}  \right\} \\
    & \qquad \bigcup \left\{(1+L_{*})~ d_H\left(\mathcal{S}_{\mathbb{P}_n},\mathcal{S}_{\mathbb{P}}\right)>\frac{s_n}{2}  \right\},
\end{align*}
\normalsize
which is a consequence of (\ref{eq:4.19}).
\hfill $\blacksquare$
\begin{remark}
    The rate $r_n$ measures the error in estimating the population Huberized DTM, whereas $s_n$ measures the error in estimating the support in Hausdorff distance. Both rates depend on the metric $d$, the regularity of $\mathbb{P}$, and the geometry of $\mathcal{S}_{\mathbb{P}}$. The corollary therefore shows how estimation errors in the Huberized DTM and the support are reflected in the corresponding persistence diagrams. For example, if $r_n=\mathcal{O}_{\mathbb{P}}(n^{-\frac{1}{2}})$ and $s_n=\mathcal{O}_{\mathbb{P}}(n^{-\frac{1}{3}})$, then the bottleneck distance is order of $\mathcal{O}_{\mathbb{P}}(n^{-\frac{1}{3}})$. Hence, the slower support estimation rate determines the overall topological estimation error.
\end{remark}


 \section{Robustness Study: Sparsity}\label{section 5}
 The preceding results establish the large-sample properties of the Huberized DTM and the associated persistence diagrams under uncontaminated sampling. We now study the effect of sparse contamination on the Huberized DTM and its induced persistence diagrams. In this section, we quantify its robustness under Huber contamination through a maximum-bias analysis and derive the corresponding robustness results for persistence diagrams under sparse contamination. In order to capture this feature, the
 observations are assumed to obtain from well known Huber's contamination model
 \begin{equation}\label{eq:5.1}
   \mathbb{P}_{\alpha} = (1 - \alpha)~ \mathbb{P}_{s} + \alpha~ \mathbb{P}_{c}, \quad \alpha \in \left[0,\frac{1}{2}\right),
\end{equation}
where $\mathbb{P}_{s} \in \mathscr{P}_{a,b}(\mathcal{M})$ satisfies Assumptions (A1)–(A4), while $\mathbb{P}_{c}$ is an arbitrary probability measure. As $\mathbb{P}_{c}$ is arbitrary, $\mathbb{P}_{\alpha}$ may not lie in $\mathscr{P}_{a,b}(\mathcal{M})$. Hence, the previous results may not be directly applicable to sparse data.
Despite this model misspecification, our inferential objective remains to recover the topological structure of the support of $\mathbb{P}_s$, i.e., $\mathcal{S}_{\mathbb{P}_s}$.
 \subsection{Maximum Bias}\label{section 5.1}
 The robustness analysis presented here is based on the maximum bias of the population Huberized DTM under Huber contamination. The maximum bias measures the worst-case contamination-induced bias
of the population Huberized DTM over all contaminating distributions at a given contamination level $\alpha$; see \textcite{Huber1981}, \textcite{Hampel2005} and \textcite{dhar2016study}. 
\begin{definition}[Maximum bias]\label{def 4.1}
    The maximum bias of the population Huberized DTM at contamination level $\alpha$ is defined by
    \begin{equation*}
        B(\alpha):=\sup_{\mathbb{P}_{c}\in \mathscr{P}(\mathcal M)}\|\delta_{\mathbb{P}_{\alpha},m,c}- \delta_{\mathbb{P}_{s},m,c} \|_{\infty},
    \end{equation*}
    where  $\mathscr{P}(\mathcal M)$ denotes the set of all Borel probability measures on $\mathcal{M}$.
\end{definition}
\begin{remark}
    It is an appropriate place to mention that unlike finite-sample error bounds, which quantify randomness due to sampling, the maximum bias is a deterministic measure of robustness that captures only the effect of contamination. Together, they characterize the statistical accuracy and robustness of the Huberized DTM.
\end{remark}
\begin{theorem}[Maximum bias bound]\label{Theorem 4.3}
    Under Assumptions (A1)–(A4), there exists a constant $C^*>0$ such that for every $0\leq \alpha<\frac{1}{2}$,
    \begin{equation*}
        B(\alpha) \leq C^*~\frac{\alpha}{(1-\alpha)},
    \end{equation*}
    where $B(\alpha)$ denotes the maximum bias (see Definition \ref{def 4.1}) of the $\delta_{\mathbb{P},m,c}$ for the model, as defined in \eqref{eq:5.1}.
\end{theorem}
\noindent \textbf{Proof:} See the \hyperref[Proof of Theorem 4.3]{\textcolor{blue}{Proof of Theorem 5.1}} in \hyperref[Appendix B]{\textcolor{blue} Appendix B}.

\hfill $\blacksquare$

\begin{remark}\label{remark 13}
The bound $B(\alpha)\leq C^*\cdot\frac{\alpha}{(1-\alpha)}$ has two implications. First, $B(\alpha)\rightarrow 0$ as $\alpha \rightarrow 0$, while the upper bound converges to $C^*$ as $\alpha \rightarrow \frac{1}{2}$, ensuring bounded maximum bias up to 50\% contamination. Second, since $C^{*}$ is independent of $\mathbb{P}_c$, the bound holds uniformly over all contaminating distributions, including point masses, heavy-tailed distributions, and adversarial outliers.
\end{remark}
To study the robustness of the associated persistence diagrams, we restrict attention to the case where the filtration generated by the contaminated population Huberized DTM is tame. This assumption ensures that the corresponding persistence diagram is well defined. This regularity condition is required only for the subsequent topological analysis. Since Huber's contamination model allows arbitrary contaminating distributions, the filtration may not be tame in general. For example, tameness may fail when the contaminating distribution has noncompact support with heavy-tail or consists of infinitely many isolated outliers.
Such cases lie beyond the scope of the present topological analysis.

\subsection{Topological Robustness under Sparse Contamination}\label{section 5.2}
The preceding subsection established a maximum bias bound for the population Huberized DTM. Now, we study the robustness result for the associated persistence diagrams.
\begin{definition}[Maximum bias of persistence diagrams]
    For each homological dimension $p\geq 0$, the maximum bias of the persistence diagram functional at contamination level $\alpha$ is defined as
    \begin{equation*}
        B_{p}^{PD}(\alpha):=\sup_{\mathbb{P}_{c}\in \mathscr{P}(\mathcal M)} d_B\left( \mathrm{Dgm}_p\big(\delta_{\mathbb{P}_{\alpha},m,c} \big), \mathrm{Dgm}_p\left(\delta_{\mathbb{P}_{s},m,c}\right)\right),
    \end{equation*}
    where $\mathrm{Dgm}_p\big(\delta_{\mathbb{P}_{\alpha},m,c}\big)$ and $\mathrm{Dgm}_p\left(\delta_{\mathbb{P}_s,m,c}\right)$ denote the persistence diagrams in homological dimension $p$, associated with the weighted Vietoris–Rips filtrations induced by $\delta_{\mathbb{P}_{\alpha},m,c}$ and $\delta_{\mathbb{P}_s,m,c}$, respectively, as constructed in \hyperref[section 3.5]{\textcolor{blue} Section \ref{section 3.5}}; and
    $\mathscr{P}(\mathcal M)$ denotes the set of all Borel probability measures on $\mathcal{M}$.
\end{definition}
\begin{theorem}[Maximum bias bound for persistence diagrams]\label{Theorem 4.4}
Under Assumptions (A1)–(A4), for every homological dimension $p \geq 0$ and $\alpha \in [0,\frac{1}{2})$,
\begin{equation}\label{4.8}
    B_{p}^{PD}(\alpha) \leq B(\alpha) + (1+L_{*})~\sup_{\mathbb{P}_{c}\in \mathscr{P}(\mathcal M)} d_H\left(\mathcal{S}_{\mathbb{P}_{\alpha}},\mathcal{S}_{\mathbb{P}_s}\right),
\end{equation}
where $B(\alpha)$ is the maximum bias (as defined in Definition \ref{def 4.1}) of the $\delta_{\mathbb{P},m,c}$,  $L_*$ is the Lipschitz constant of $\delta_{\mathbb{P}_s,m,c}$ established in \hyperref[Theorem 3.2]{\textcolor{blue} Theorem \ref{Theorem 3.2}}, and  $d_H$ denotes the Hausdorff distance. In particular, if $$\sup_{\mathbb{P}_{c}\in \mathscr{P}(\mathcal M)} d_H\left(\mathcal{S}_{\mathbb{P}_{\alpha}},\mathcal{S}_{\mathbb{P}_s}\right)=\mathcal{O}\left(\frac{\alpha}{1-\alpha}\right),$$
then, there exists a constant $C^{**}>0$ such that
\begin{equation}\label{4.9}
     B_{p}^{PD}(\alpha) \leq C^{**}~\frac{\alpha}{(1-\alpha)}.
\end{equation}
\end{theorem}
\noindent \textbf{Proof:} See the \hyperref[Proof of Theorem 4.4]{\textcolor{blue}{Proof of Theorem 5.2}} in \hyperref[Appendix B]{\textcolor{blue} Appendix B}.

\hfill $\blacksquare$

\begin{remark}
    \hyperref[Theorem 4.4]{\textcolor{blue} Theorem \ref{Theorem 4.4}} extends the robustness results of \hyperref[Theorem 4.3]{\textcolor{blue} Theorem \ref{Theorem 4.3}} from the Huberized DTM to the level of persistence diagrams. The bound in \eqref{4.8} reveals two interpretable components. The first is the functional bias $B(\alpha)$, whose properties are discussed in \hyperref[remark 13]{\textcolor{blue} Remark \ref{remark 13}}. The second is the geometric distortion $d_H\left(\mathcal{S}_{\mathbb{P}_{\alpha}},\mathcal{S}_{\mathbb{P}_s}\right)$, which captures the Hausdorff deviation of the contaminated support $\mathcal{S}_{\mathbb{P}_{\alpha}}:=\mathcal{S}_{\mathbb{P}_s} \cup \mathcal{S}_{\mathbb{P}_c}$ from $\mathcal{S}_{\mathbb{P}_s}$, arising from the huber contamination model $\mathbb{P}_{\alpha} = (1 - \alpha)~ \mathbb{P}_{s} + \alpha~ \mathbb{P}_{c},~ \alpha \in \left[0,\frac{1}{2}\right) $. If the geometric distortion is of the same order as the functional bias, then the maximum bias of persistence diagrams is of order $\mathcal{O} \left(\frac{\alpha}{1-\alpha}\right)$ in bottleneck distance. 
\end{remark}

\section{Numerical Experiments} \label{section 6}
In this section, we illustrate the empirical performance of our proposed procedure for topological inference under sparse contamination. Contaminating observations typically introduce both external outliers (points far from the signal support) and internal outliers (points inside void regions). 
As illustrated in \hyperref[Figure 6.0]{\textcolor{blue}{Figure \ref{Figure 6.0}}}, such observations often distort the underlying topological structure.
For instance, external outliers may induce spurious loops or connected components, whereas internal outliers may fill existing voids and thereby obscure genuine topological features. 

\begin{figure}[!htbp]
        \centering       
        \includegraphics[scale=0.6]{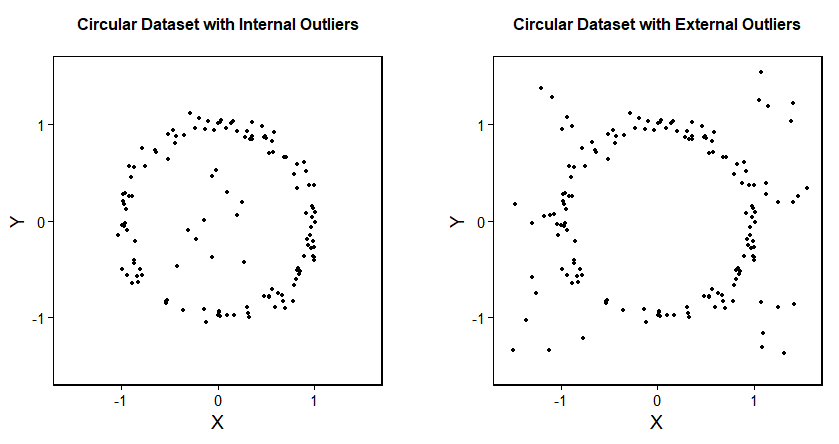}
        \caption{Illustration of internal and external outliers for a circular signal. (Left) internal outliers lie within the enclosed hole; (right) external outliers lie outside the signal support.}
        \label{Figure 6.0}
        \noindent
\end{figure}

Here, the observations are drawn independently from a mixture distribution
\begin{equation*}
  \mathbb{P}_{\alpha} = (1 - \alpha)~\mathbb{P}_{s} + \alpha~ \mathbb{P}_{c}, \quad \alpha \in \left[0,\frac{1}{2}\right),
\end{equation*}
where $\mathbb{P}_{s}$ is the signal distribution of interest, and $\mathbb{P}_{\text{c}}$ is an arbitrary contaminating distribution. Now, the goal is to recover the persistent topological features associated with the signal support  $\mathcal{S}_{\mathbb{P}_{s}}$ from 
the contaminated observations $ X_1, \cdots, X_n \stackrel{\mathrm{i.i.d.}}{\sim} \mathbb{P}_{\alpha} $.

\subsection*{Example: Seven-Circle Packing Configuration with Sparsity}
We consider a two-dimensional point pattern consisting of seven neighbouring circular structures in $\mathbb{R}^2$. This configuration provides a benchmark for examining whether the proposed method can recover persistent loop structures from sparse observations in which the underlying $H_1$ structure is known.

\textbf{Data generation from Signal component:} The signal distribution $\mathbb{P}_{s}$ is supported on the circumferences of seven circles of unit radius,
arranged in a hexagonal packing configuration. The central circle is centered at $C_1=(0,0)$, while the six surrounding circles have centers 
\begin{equation*}
      C_j = (c_{j1},c_{j2}) = \left( 2\cos\frac{2\pi(j-2)}{6},~ 2\sin\frac{2\pi(j-2)}{6} \right), \quad 
       j=2,\cdots,7.
\end{equation*}
This geometry ensures the six outer circles are equally spaced around the central circle, with each being mutually tangent to the central circle and to its two neighbouring outer circles. 
For each circle $j=1,\cdots,7$, we generate 80 noisy observations as
\begin{equation*}
    X_{ij} = \left(c_{j1}+ \cos\Theta_{ij} + \varepsilon_{ij1}, ~c_{j2}+\sin\Theta_{ij}+\varepsilon_{ij2} \right) \in\mathbb{R}^2,\quad i=1,\cdots,80,
\end{equation*}
where $C_j=(c_{j1},c_{j2})$ is the center of the $j$-th circle,
$\Theta_{ij}\stackrel{\mathrm{i.i.d.}}{\sim}\text{Unif}~(0,2\pi)$ is the angular coordinate of the $i$-th point on the $j$-th circle's circumference, and the additive noise terms satisfy $\varepsilon_{ij1},\varepsilon_{ij2} \stackrel{\mathrm{i.i.d.}}{\sim}\operatorname{Unif}(-0.1,0.1)$.
Consequently, the resulting 560
observations (i.e., 80 observations for each circle) form a noisy sample of the union of the seven circumferences. The underlying signal support has seven independent $H_1$ loops corresponding to the seven circles.
\begin{remark}
    The hexagonal packing arrangement represents canonical geometric patterns of spatial organization such as cellular organization in tissue biology (\textcite{gibson2006emergence}), the packing-number and lattice-packing structure in condensed matter and mathematical physics (\textcite{conway1999sphere}). It provides a simple as well as insightful benchmark for studying topological structure in spatial data.
\end{remark}

\textbf{Data generation from Contamination component:} To examine robustness in presence of sparse contamination, we generate $350$
contaminating observations independently from $\text{Unif}([-3,3]^2)$. As the contamination region contains the entire seven-circle configuration, contaminating points may lie both outside the signal support as well as in the regions between the circular components (\hyperref[Figure 6.1]{\textcolor{blue}{Figure \ref{Figure 6.1}}}, left). The resulting sample therefore contains both external and internal outliers, which provides a challenging testbed for recovering the underlying circular structures.

\begin{figure}[!htbp]
        \centering       
        \includegraphics[scale=0.45]{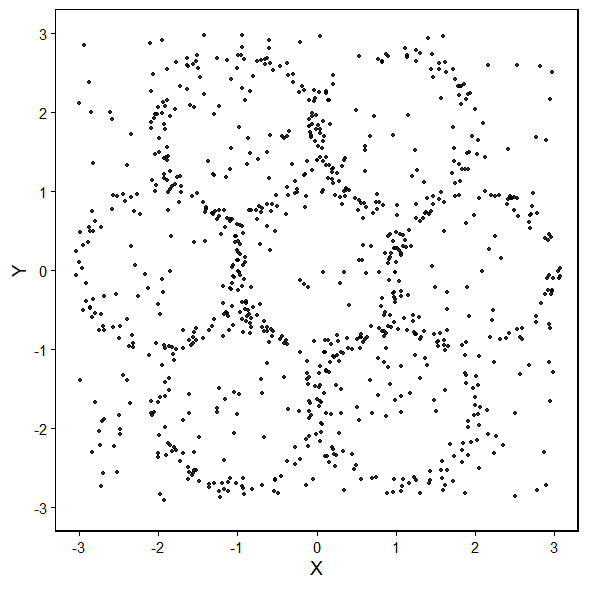}~
        \includegraphics[scale=0.45]{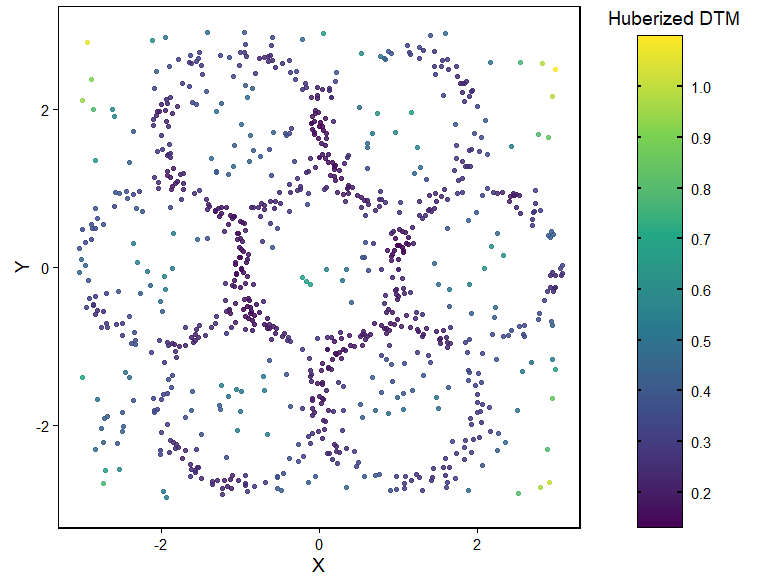}
        \caption{Contaminated observations from the seven-circle configuration (left) and the same sample with the empirical Huberized DTM values at the observed points represented by a color scale (right), for $m=0.02$ and $c=0.05$.}
        \label{Figure 6.1}
        \noindent
\end{figure}
\begin{remark}
    At $\alpha=350/910 \approx 0.385$ in the model designed in \eqref{eq:5.1}, internal outliers obscure genuine loops while external outliers create spurious ones. This dual distortion provides a challenging testbed for robust topological inference.
\end{remark}
To quantify the underlying topological structure in presence of noise, we first compute the empirical Huberized DTM $\widehat{\delta}_{\mathbb{P}_n,m,c} (\cdot)$ at the observed sample points with mass parameter $m =0.02$ and Huber tuning parameter $c=0.05$ (see, \hyperref[Figure 6.1]{\textcolor{blue}{Figure \ref{Figure 6.1}}}, right diagram). 
The resulting Huberized DTM values are then used as vertex weights in the Huberized DTM weighted Vietoris-Rips filtration, from which persistence diagrams are computed in homology dimensions $k=0,1$.
Precisely speaking, the $H_0$ persistence diagram captures cluster structure, while the $H_1$ persistence diagram records one-dimensional loop features.
To illustrate the resulting topological structure, \hyperref[Figure 6.2]{\textcolor{blue}{Figure \ref{Figure 6.2}}} shows a representative persistence diagrams obtained from the noisy signal sample and from the same sample after contamination.
The seven most persistent features correspond to the seven loop structures in the signal configuration and are well-separated from the remaining features. To characterize the sampling variability of these $H_1$ features across samples, we 
perform $R=100$ independent Monte Carlo replications. 
For each replication $r=1,\cdots, R$,
let $\widehat{\mathrm{Dgm}}_{1}^{(r)}$ denote the empirical $H_1$ persistence diagram obtained from the Huberized DTM-weighted Vietoris–Rips filtration constructed using the $r$-th replication. Each topological feature is represented by a birth time-death time pair  $\big(b_i^{(r)}, d_i^{(r)}\big)\in \widehat{\mathrm{Dgm}}_{1}^{(r)},~ i=1,2.\cdots,$ with persistence $\ell^{(r)}=d_i^{(r)}-b_i^{(r)}$. 
\begin{figure}[!htbp]
        \centering       
        \includegraphics[scale=0.375]{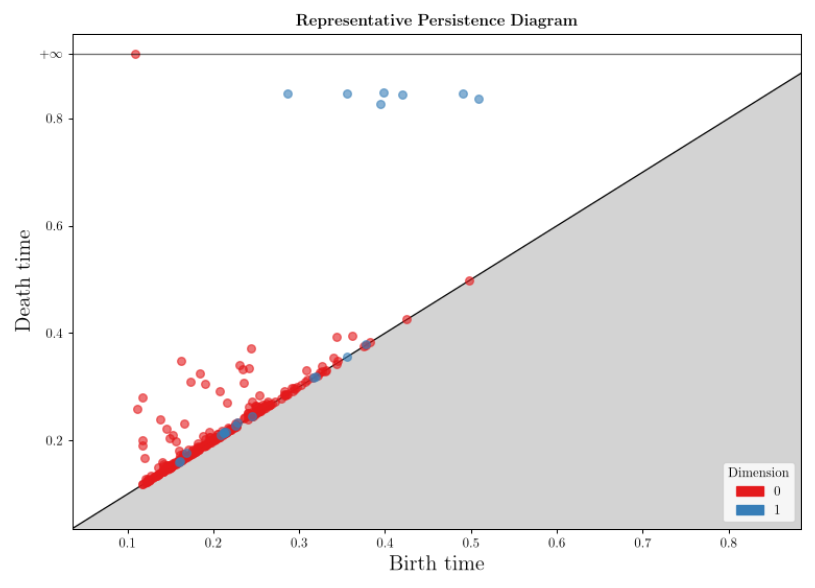}~
        \includegraphics[scale=0.385]{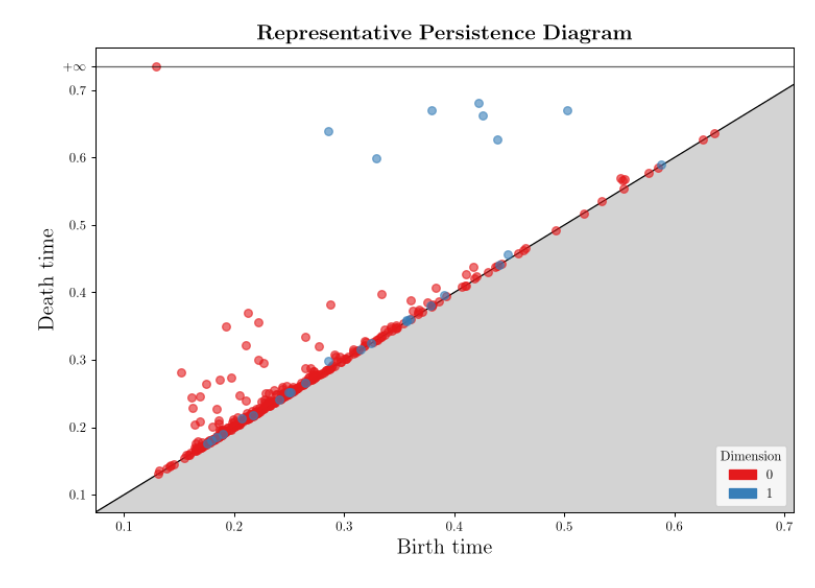}
        \caption{Representative Persistence diagram computed from the Huberized DTM weighted Vietoris-Rips filtration for 560 noisy signal points generated from the seven-circle configuration (left diagram), and for the same signal points after augmented with 350 uniformly distributed contaminating points (right diagram), with $m=0.02$ and $c=0.05$.}
        \label{Figure 6.2}
        \noindent
\end{figure}
To distinguish dominant topological signals from noise, these persistence values are ordered in descending order, $\ell_{(1)}^{(r)} \geq \ell_{(2)}^{(r)} \geq \cdots$.
As the signal configuration contains seven target loop features,
we restrict attention to the first eight order statistics,
thereby including the seven target features together with the next most persistent feature for comparison.
For each $j \in \{1, \dots, 8\}$, the corresponding persistence across the $R$ replications are summarized by the Monte Carlo mean $\bar \ell_{(j)}=R^{-1}\sum_{r=1}^R \ell_{(j)}^{(r)}$ together with the corresponding Monte Carlo standard error. The resulting Monte Carlo summaries are reported in \hyperref[Table 6.1]{\textcolor{blue}{Table \ref{Table 6.1}}}.
\begin{table}[htbp]
\centering
\setlength{\tabcolsep}{6pt}
\renewcommand{\thetable}{6.1}
\caption{Monte Carlo mean persistence and standard error (SE) for the first eight $H_1$ persistence order statistics at $m=0.02$ and $c=0.05$ based on $R=100$ replications.}
\label{Table 6.1}
\renewcommand{\arraystretch}{1.3} 
\begin{tabular}{l cccccccc}
\midrule
\makecell{Mean Persistence \\ order statistic} & $\bar\ell_{(1)}$ & $\bar\ell_{(2)}$ & $\bar\ell_{(3)}$ & $\bar\ell_{(4)}$ & $\bar\ell_{(5)}$ & $\bar\ell_{(6)}$ & $\bar\ell_{(7)}$ & $\bar\ell_{(8)}$ \\
\midrule
Mean persistence & 0.3337 & 0.2812 & 0.2530 & 0.2276 & 0.2010 & 0.1748 & 0.1316 & 0.0238 \\

Monte Carlo SE & 0.0042 & 0.0034 & 0.0033 & 0.0032 & 0.0032 & 0.0033 & 0.0038 & 0.0013 \\
\bottomrule
\end{tabular}
\end{table}

At $m=0.02$ and $c=0.05$, the mean persistence of $\bar\ell_{(1)}$ to $\bar\ell_{(7)}$ ranges from $0.3337$ to $0.1316$, while the $\bar\ell_{(8)}$ has mean persistence $0.0238$. The difference between the mean values of $\bar\ell_{(7)}$ and $\bar\ell_{(8)}$ is $0.1078$, which is significantly larger than other mean values. This separation is also reflected in the persistence diagram. 
This fact indicates an evidence of a separation after the seventh feature. Moreover, the corresponding Monte Carlo standard errors range from $0.0013$ to $0.0042$, indicating that the reported Monte Carlo means are estimated with high precision over the $R=100$ replications. The observed separation between the seventh and eighth features is consistent with the seven-loop structure specified by the simulation design. Overall, these findings in numerical study favours of using persistence lengths to characterize the target loop structure.

\subsection{Empirical Power Study}
We now study whether sparsity leads to a detectable change in the estimated $H_1$ persistence structure, as measured by the bottleneck distance from a fixed signal reference. In this study, we consider $n=910$,
$m=0.02$, and $c=0.05$, giving $k=\lceil mn \rceil=18$. Here note that $910$ is a multiple of $7$, as required. The contaminated distribution is
\begin{equation*}
  \mathbb{P}_{\alpha} = (1 - \alpha)~ \mathbb{P}_{s} + \alpha~ \mathbb{P}_{c},
\end{equation*}
where $\mathbb{P}_{s}$ and $\mathbb{P}_{c}$ denote the signal and contamination distributions, respectively. Here observe that $\alpha =350/910$. We now would like to test
\begin{equation*}
    \mathcal{H}_{null}: \mathbb{P}_{\mathcal{Y}}=\mathbb{P}_{s}\quad \text{against}\quad \mathcal{H}_{alt}: \mathbb{P}_{\mathcal{Y}}=\mathbb{P}_{\alpha},
\end{equation*}
where $\mathbb{P}_{\mathcal{Y}}$ denotes the data-generating distribution of the 
test sample $\mathcal{Y}$. Now, it should be explained how this test is associated with the sparsity. Note that under null hypothesis, the data are obtained from regular distribution like $\mathbb{P}_{s}$ satisfying the regularity condition like $(a, b)$ standard measure whereas under alternative hypothesis, the data obtained from contaminated distribution like $\mathbb{P}_{\alpha}$ for any $\alpha > 0$, which does not satisfy $(a, b)$ measure criteria. In other words, under $\mathbb{P}_{s}$, the data are located outside the main data cloud, which creates the sparsity in the whole data set. Next, the implementation of the test is as follows. 

Let, $\mathcal{X}_{\text{ref}}=\{x_1,\cdots, x_{910}  \}$ be a reference sample of $910$ observations, generated independently from $\mathbb{P}_{s}$. Using the Huberized-DTM-weighted Vietoris-Rips filtration, we compute its $H_1$ persistence diagram $\widehat{\mathrm{Dgm}}_{1}\left(\mathcal{X}_{\text{ref}}\right)$, which is kept fixed in this study.
For any test sample $\mathcal{Y}$, the test statistic is defined as,
\begin{equation}
    T (\mathcal{Y})=d_B\left(\widehat{\mathrm{Dgm}}_{1}\left(\mathcal{Y}\right), ~\widehat{\mathrm{Dgm}}_{1}\left(\mathcal{X}_{\text{ref}}\right)\right),
\end{equation}
where $d_B$ denotes the bottleneck distance.

To calibrate the test at the $5\%$ significance level, we independently generate $R_0=500$ signal samples $\mathcal{Y}_1^{(0)},\cdots,\mathcal{Y}_{R_0}^{(0)}$, each consisting of $910$ observations from $\mathbb{P}_{s}$ (i.e., under $\mathcal{H}_{null}$), and compute 
\begin{equation}
    T_r^{(0)}=T \left(\mathcal{Y}_r^{(0)}\right),\quad r=1,\cdots,R_0.
\end{equation}
The resulting values statistic ${\small \left\{T_r^{(0)}\right\}_{r=1}^{R_0}}$ are used to approximate the null distribution of the bottleneck-distance statistic conditional on $\widehat{\mathrm{Dgm}}_{1}\left(\mathcal{X}_{\text{ref}}\right)$. Let $\widehat{F}_0$ denote the corresponding empirical distribution function obtained by
$ \left\{T_r^{(0)}\right\}_{r=1}^{R_0}$, and the critical value is estimated by $0.95$-th quantile of the empirical version of null distribution, which is denoted by $\widehat{c}_{0.95}$.
The test reject $\mathcal{H}_{null}$ if $T (\mathcal{Y})>\widehat{c}_{0.95}$.

\textbf{Empirical Power Evaluation:} 
Here we compute the empirical rejection rates (i.e., empirical power) when $\alpha \in \{0,0.10,0.20,0.30,0.385\}$. For each $\alpha$, we generate $R=100$ independent test samples ${\small \left\{\mathcal{Y}_r^{\alpha}\right\}_{r=1}^{R}}$, with $\lfloor \alpha n \rfloor$ observations from $\mathbb{P}_c$ and $n-\lfloor \alpha n \rfloor$ observations from $\mathbb{P}_s$ (i.e., under $\mathcal{H}_{alt}$). The $\widehat{\mathrm{Dgm}}_{1}\left(\mathcal{X}_{\text{ref}}\right)$ and $\widehat{c}_{0.95}$ are kept fixed across all contamination levels. Finally, the empirical power function at $\alpha$ is defined as,
\begin{equation*}
\widehat{\pi}(\alpha) = \frac{1}{R} \sum_{r=1}^{R} \mathbbm{1}_{\left\{ T\left(\mathcal{Y}_r^{(\alpha)}\right) > \widehat{c}_{0.95} \right\}}.
\end{equation*}
In this context, observe that $\widehat{\pi}(0)$ yields the empirical size of the test, while $\widehat{\pi}(\alpha)$ for $\alpha>0$ evaluates statistical power of the test. The estimated size and power across $\alpha \in \{0,0.10,0.20,0.30,0.385\}$ are reported in \hyperref[Table 6.2]{\textcolor{blue}{Table \ref{Table 6.2}}}.
\begin{table}[htbp]
\centering
\setlength{\tabcolsep}{10pt}
\renewcommand{\arraystretch}{1.3}
\renewcommand{\thetable}{6.2}
\caption{Empirical size ($\alpha=0$) and power ($\alpha>0$) of the bottleneck-distance based test across contamination levels, based on $100$ replications. The critical value $\widehat{c}_{0.95}$ is calibrated from $500$ null samples and is kept fixed across contamination levels.}
\label{Table 6.2}
\begin{tabular}{cccc}
\midrule
\makecell{Contamination \\ proportion $(\alpha)$} \qquad &
\makecell{Contaminating \\ observations} \qquad &
\makecell{Signal \\ observations} \qquad &
\makecell{Empirical  \\ size/power $(\widehat{\pi}(\alpha))$} \\
\midrule
$0$    \qquad  & $0$   & $910$ & $0.06$ \\
$0.10$ \qquad  & $91$  & $819$ & $0.77$ \\
$0.20$ \qquad  & $182$ & $728$ & $1.00$ \\
$0.30$ \qquad  & $273$ & $637$ & $1.00$ \\
$0.385$ \qquad & $350$ & $560$ & $1.00$ \\
\bottomrule
\end{tabular}
\end{table}

As shown in \hyperref[Table 6.2]{\textcolor{blue}{Table \ref{Table 6.2}}}, the empirical size is $0.06$ at $\alpha=0$, which is close to the $5\%$ significance level. The rejection rate increases from $0.77$ at $\alpha=0.10$ to $1.00$ at $\alpha=0.20$, remaining at $1.00$ thereafter. Therefore, for this data, it is evident that the proposed test detects changes in the estimated $H_1$ persistence structure. Strictly speaking, these results indicate that $\alpha=0.20$ introduces sufficient topological changes in the $H_1$ features for the test to reliably reject the null hypothesis in this example.


\section{Real-data Analysis} \label{section 7}
In this section, we illustrate the practical utility of our proposed methodology through an analysis of 
the experimentally determined three-dimensional structure of the protein M-ferritin.
Ferritin stores iron in a cavity enclosed by its protein shell. This makes it a natural testbed for testing whether the proposed filtration identifies the cavity as a persistent $H_2$ feature. Since the atomic coordinates are obtained from X-ray crystallography, they are subject to measurement uncertainty. Our objective is to demonstrate how the proposed methodology provides robust insights in such real-world scenarios.

\subsubsection*{Topological Characterization of the Iron-Storage Cavity in M-Ferritin}
In cellular biology, a protein's biological function is governed not only by its linear sequence of amino acids but also by its three-dimensional conformation. Such conformations often contain enclosed cavities, which may serve as sites for the binding, storage, or transport of small molecules and ions. Characterizing these cavities is essential for understanding protein function. As a case study, we consider the three-dimensional structure of M-ferritin, a multimeric iron-storage protein that stores iron within an enclosed internal cavity (see, e.g., \textcite{theil1987ferritin}, \textcite{chasteen1999mineralization}). From geometric point of view, the iron-storage cavity is a hollow region within the protein structure. Topologically, such a cavity corresponds to a $H_2$ feature, which can be detected and summarize using persistent homology. The corresponding atomic coordinates are available in the Protein Data Bank (\href{https://www.rcsb.org/structure/1MFR}{PDB ID: 1MFR}), and the structure was determined by X-ray crystallography (\textcite{ha1999crystal}).

To analyze the structure topologically, we represent the M-ferritin Protein by the Cartesian coordinates $(x,y,z)$ of
its $C_{\alpha}$ atoms, excluding water molecules, ions and other heteroatoms. The resulting dataset consists of 4104 sample points.
To reduce the computational cost of computing $H_2$ persistence of $4,104$ points, farthest-point sampling (maxmin) technique (see, \textcite{gonzalez1985clustering}) is applied to obtain a representative subsample of $400$ points. 
To account for variability induced by the initialization, we perform farthest-point sampling over $12$ replications, each initiated with an independently selected starting point. For each subsample,
we then construct the Huberized-DTM-weighted Vietoris–Rips filtration with $m=0.02$ and $c=0.05$ on the subsampled dataset, and compute $H_2$ persistence to characterize the enclosed voids in the protein structure. 
\hyperref[Figure 7.1]{\textcolor{blue}{Figure \ref{Figure 7.1}}} shows a representative Huberized DTM based persistence diagram. To characterize the dominant topological features and quantify their variability across replications, Monte Carlo means and standard errors of the birth times, death times, and persistence of the three largest persistent $H_2$ features in each replication are reported in \hyperref[Table 7.1]{\textcolor{blue}{Table \ref{Table 7.1}}}, with features ranked by persistence.

\begin{figure}[!htbp]
        \centering       
        \includegraphics[scale=0.45]{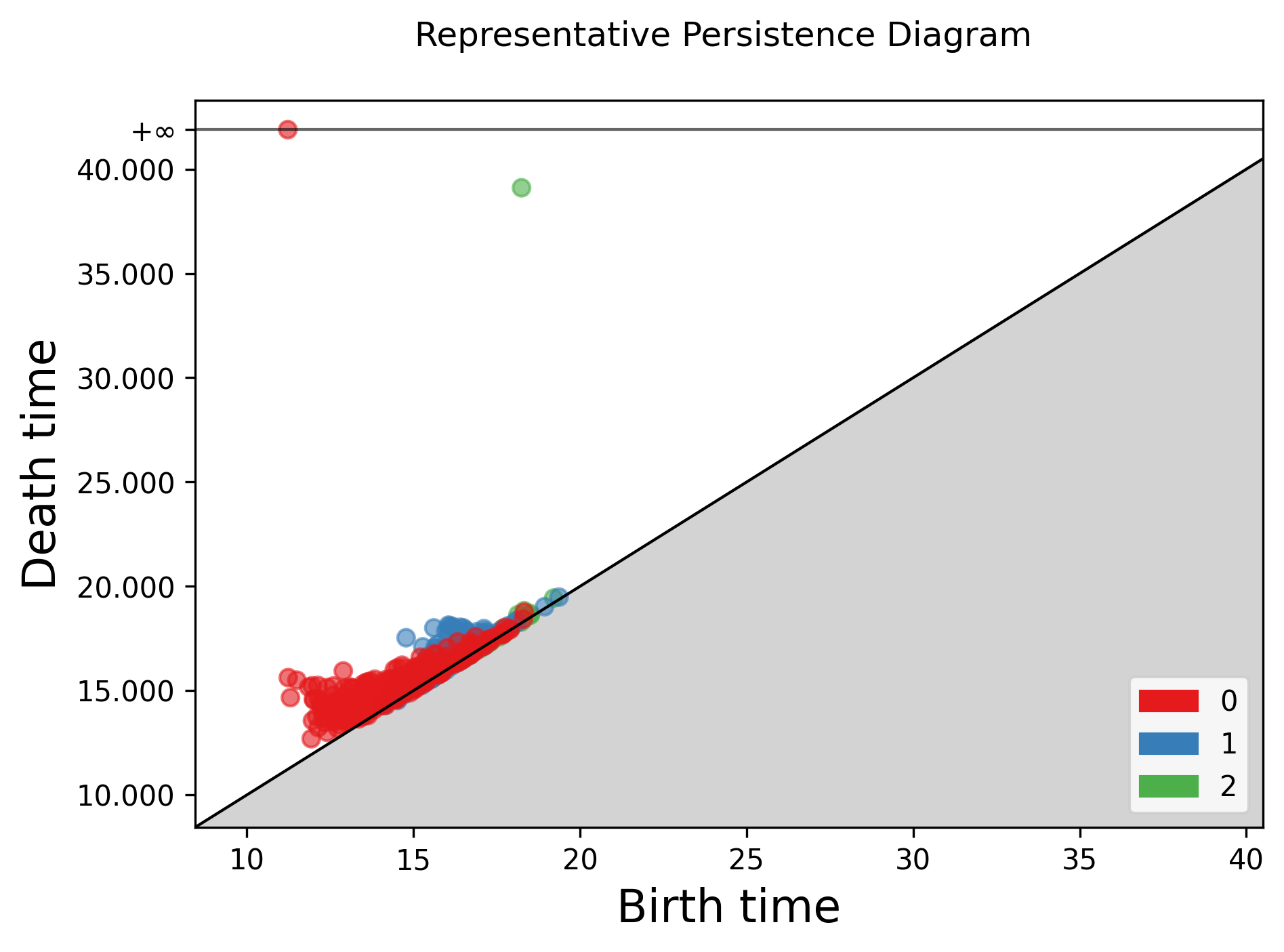}
        \caption{Representative Huberized DTM based Persistence diagram for the M-ferritin structure ($n=400$, $m=0.02$, $c=0.05$).}
        \label{Figure 7.1}
        \noindent
\end{figure}
\begin{table}[htbp]
    \centering
    \setlength{\tabcolsep}{10pt}
    \renewcommand{\arraystretch}{1.3}
    \renewcommand{\thetable}{7.1}
    \caption{Monte Carlo means of the birth time, death time, and persistence of the three most persistent $H_2$ classes for the subsampled M-ferritin structure
    ($n=400$, $m=0.02$, $c=0.05$). Standard errors (SE) are in parentheses.}
    \label{Table 7.1}
    \begin{tabular}{cccc}
        \toprule
        Rank (by persistence) & Birth time & Death time & Persistence \\
        \midrule
        1 & $18.7446$ & $38.7149$ & $19.9703 ~ (0.5203)$ \\
        2 & $16.3466$ & $17.0583$ & $0.7117 ~(0.1190)$ \\
        3 & $15.8679$ & $16.4360$ & $0.5681~ (0.0756)$ \\
        \bottomrule
    \end{tabular}
\end{table}

In the representative Huberized DTM based persistence diagram (see, for instance, \hyperref[Figure 7.1]{\textcolor{blue}{Figure \ref{Figure 7.1}}}), one $H_2$ feature is well separated from the diagonal, whereas the remaining features lie close to it. \hyperref[Table 7.1]{\textcolor{blue}{Table \ref{Table 7.1}}} shows that this separation is stable across replications. The dominant $H_2$ feature has mean birth time $18.7446$, mean death time $38.7149$, and mean persistence $19.9703$ (SE $0.5203$). In contrast, the second and third most persistent $H_2$ features have substantially smaller mean persistence of $0.7117$ (SE $0.1190$) and $0.5681$ (SE $0.0756$), respectively. This pronounced separation in persistence indicates that the dominant $H_2$ feature is more persistent than the other observed $H_2$ features. The remaining $H_2$ features are short persistence and lie close to the diagonal, suggesting that they may arise due noise. Taken together, its substantially longer persistence and consistency with the enclosed cavity lead us to identify the dominant feature as the iron-storage cavity of M-ferritin.



\section{Concluding Remarks} \label{section 8}
This article studies robust topological learning under sparse contamination through Huberized distance functionals. Even a small fraction of contamination can distort persistent homology by perturbing both the filtration-generating functional and the support on which the filtration is constructed. To address these effects, we construct a filtration based on the Huberized Distance-to-Measure, which reduces the influence of atypical observations while retaining the geometric information relevant to topological inference. 
We establish uniform consistency of the empirical Huberized DTM, and together with stability theorems for persistence diagrams, show that the associated persistence diagrams are consistent. We also derive a maximum-bias bound in the bottleneck distance under sparse contamination. As an illustration, we consider a seven-circle packing configuration in $\mathbb{R}^2$, whose signal support has seven independent $H_1$ loops. Even at contamination level $\alpha=0.385$,  the proposed procedure successfully recovers all seven loops. The empirical power study further shows that the bottleneck distance based test detects contamination-induced changes in the estimated $H_1$ persistence structure with high power.

Several extensions would strengthen the framework developed in this work. A natural extension would be to develop adaptive choices of $c$ and $m$, together with corresponding theoretical results. Our theoretical results are derived under the Huber contamination model with independent observations. Extending these results to dependent data and structured contamination (e.g., clustered or spatially concentrated outliers) would provide a broader robustness framework. Finally, weighted Vietoris-Rips filtrations can become computationally expensive for large datasets and higher homological dimensions. Developing scalable approximations through sparsification, subsampling, or landmark points, together with explicit bounds on the resulting bottleneck distance error, could extend the method to large datasets and higher homological dimensions.

\vspace{0.2 in}

\noindent {\bf Acknowledgment:} The authors thank Professor Steve Oudot and Professor Henry Adams for stimulating discussions on Topological Data Analysis (TDA), which helped shape many of the ideas developed in this work. The first author would also like to thank the participants of the workshop “OneMath World School 2026” at the Chennai Mathematical Institute (CMI), whose insightful questions provided him with a deeper understanding of TDA. The authors sincerely thank Nivesh Mishra, a PhD student in the Department of Mechanical Engineering, IIT Kanpur, for his assistance with the high-performance computing (HPC) computations.

\vspace{0.1in}

\noindent {\bf Code Availability:} The codes of numerical studies are available to the first author upon request. 

\vspace{0.1in}

\noindent {\bf AI Disclosure:} The methodological development and mathematical results were carried out by the authors without AI assistance. AI was used solely to assist with the literature survey in the Introduction, language polishing, and some aspects of the R coding. To ensure the accuracy and quality of the article, the authors themselves thoroughly reviewed the manuscript, refined the language, and independently verified the code.

\section*{Appendix} \label{Appendix}
\subsection*{Notation and Definitions}
\subsubsection*{Geometric Regularity and Positive Reach}
The stability of topological inference depends on geometric regularity of the support $\mathcal{S}_{\mathbb{P}}$. This regularity is captured by the concept of reach (see \textcite{federer1959curvature}). To define the reach of a compact set $S \subset \mathcal{M}$, we first introduce its metric projection and medial axis.
\begin{definition*}[Metric Projection]
Let $(\mathcal{M},d)$ be the metric space, and suppose that $S \subset \mathcal{M}$ is a nonempty compact set. For $y\in \mathcal{M}$,  $d(y,S):=\displaystyle\inf_{x\in S} d(y,x)$.
Then, the metric projection of $y$ onto $S$ is defined as
\begin{equation*}
    \Pi_{S}(y):=\arg \min_{x\in S} d(y,x)=\{x\in S:d(y,x)=d(y,S) \}.
\end{equation*}
Therefore, $\Pi_{S}(y)$ consists of all points of $S$, which minimize the distance to $y$. Since $S$ is compact, $\Pi_{S}(y) \neq \varnothing$. 
\end{definition*}

\begin{definition*}[Medial Axis]
    Let $(\mathcal{M},d)$ be the metric space, and suppose that $S \subset \mathcal{M}$ is a nonempty compact set. The medial axis of $S$ is defined as 
    \begin{equation*}
        \text{Med}(S):=\{y \in \mathcal{M}:|\Pi_{S}(y)|\geq 2\},
    \end{equation*}
    where $|\Pi_{S}(y)|$ denotes the cardinality of $\Pi_{S}(y)$. Thus, $\text{Med}(S)$ consists of all points in $\mathcal{M}$ having at least two distinct nearest points in $S$.
\end{definition*}

\begin{definition*}[Reach]
     Let $(\mathcal{M},d)$ be the metric space, and suppose that $S \subset \mathcal{M}$ is a nonempty compact set. The reach  of $S$ is defined as
     \begin{equation*}
         \text{reach}(S)=\inf_{z\in \text{Med}(S)}d(z,S).
     \end{equation*}
\end{definition*}
\noindent Accordingly, $S$ is said to have positive reach if $\text{reach}(S)>0$. For the support $\mathcal{S}_{\mathbb{P}}$, we assume $\text{reach}(\mathcal{S}_{\mathbb{P}})>0$.

\subsection*{Appendix A} \label{Appendix A}
\subsubsection*{A.1.\quad \hyperref[proposition 2.2]{\textcolor{blue}{Proof of Proposition \ref{proposition 2.2}}}} \label{Proof of proposition 2.2}
\noindent For $x,y \in \mathcal{M}$, observe that
\small
    \begin{align*}
        \left|\mathbb{S}_x(X)-\mathbb{S}_y(X)\right|  =\left|[d(x,X)]^2-[d(y,X)]^2\right|
        =|d(x,X)-d(y,X)|\cdot|d(x,X)+d(y,X)|
        \leq2D~d(x,y),
    \end{align*}
    \normalsize
where the inequality follows from the triangle inequality, and in view of the fact that $\mathcal{S}_{\mathbb{P}}$ is bounded set with diameter $D$. Therefore, the process, $\{\mathbb{S}_x(X):x\in \mathcal{M}\}$ has $2D$-Lipschitz sample paths with respect to $d$.
It completes the proof. 
 
\hfill $\blacksquare$

\subsubsection*{A.2.\quad \hyperref[proposition 2.3]{\textcolor{blue}{Proof of Proposition \ref{proposition 2.3}}}} \label{Proof of proposition 2.3}
\noindent By \hyperref[proposition 2.2]{\textcolor{blue}{Proposition \ref{proposition 2.2}}}, $\mathbb{S}_x(X)\leq \mathbb{S}_y(X)+2D~d(x,y),$ almost surely. 
Hence,
\begin{equation}\label{eq:8.1}
    \{\mathbb{S}_y(X)\leq F_y^{-1}(m)\}\subseteq\{\mathbb{S}_x(X)\leq F_y^{-1}(m)+2D~d(x,y)\},
\end{equation}
and therefore, $F_x(F_y^{-1}(m)+2D~d(x,y))\geq F_y(F_y^{-1}(m))$. As $F_y(F_y^{-1}(m))\geq m$, it follows that $F_x(F_y^{-1}(m)+2D~d(x,y))\geq m$. Now, by definition of the quantile, $F_x^{-1}(m)=\inf\{t:F_x(t) \geq m\}$, and hence, $F_x^{-1}(m)\leq F_y^{-1}(m)+2D~d(x,y).$
Finally, by symmetry, 
\begin{equation}\label{eq:8.2}
    F_y^{-1}(m)\leq F_x^{-1}(m)+2D~d(x,y).
\end{equation}
Combining \eqref{eq:8.1} and \eqref{eq:8.2}, we have
\begin{equation*}
    |F_x^{-1}(m)-F_y^{-1}(m)|\leq 2D~d(x,y).
\end{equation*}
Hence, the quantile map $x\mapsto F_x^{-1}(m)$ is $2D$-Lipschitz continuous on $\mathcal{M}$. This completes the proof. 
 
\hfill $\blacksquare$

\subsection*{A.3.}
\begin{lemma3}\label{Lemma 3}
Let $(\mathcal{M},d)$ be a metric space and fix $x\in \mathcal{M}$. Define $\mathbb{S}_x(X)=[d(x,X)]^2$,
and suppose that $\mathbb{S}_{x}(X_{(1)})\leq \mathbb{S}_{x}(X_{(2)})\leq \cdots \leq \mathbb{S}_{x}(X_{(n)})$ denote the order statistics of $\left\{\mathbb{S}_x(X_j): j=1,\ldots,n\right\}$. For $c>0$, let $L_{n}(x;\theta)$ be the empirical Huber objective as defined in \eqref{eq:3.3}. Then
\begin{equation*}
    \frac{\partial}{\partial \theta}L_{n}(x;\theta)<0 ~\text{for all}~\theta<\mathbb{S}_{x}(X_{(1)}),~\text{and}~\frac{\partial}{\partial \theta}L_{n}(x;\theta)>0~\text{for all}~\theta>\mathbb{S}_{x}(X_{(k)}).
\end{equation*}
 Consequently, ${\displaystyle\arg \min_{\theta\in \mathbb{R}} L_{n}(x;\theta) \subseteq [\mathbb{S}_{x}(X_{(1)}),\mathbb{S}_{x}(X_{(k)})] \subset [0,\infty).}$
\end{lemma3}
\noindent \textbf{Proof:} Fix $x\in \mathcal{M}$. As $\rho_c$ is continuously differentiable with $\psi_c=\rho_c'$,
\begin{equation*}
    \frac{\partial}{\partial \theta}L_{n}(x;\theta)=-\frac{1}{k}\sum_{j=1}^{k}\psi_c\left(\mathbb{S}_{x}(X_{(j)})- \theta\right).
\end{equation*}
If $\theta < \mathbb{S}_{x}(X_{(1)})$, then
$\left(\mathbb{S}_{x}(X_{(j)})-\theta\right)>0$ for $j=1,\cdots,k$. Hence, $\psi_c\left(\mathbb{S}_{x}(X_{(j)})-\theta\right)>0$ for every $j$, and therefore, 
\begin{equation*}
    \frac{\partial}{\partial \theta}L_{n}(x;\theta) <0.
\end{equation*}
Thus,
$L_{n}(x,\cdot)$ is strictly decreasing on $\left(-\infty,~\mathbb{S}_{x}(X_{(1)})\right)$. Similarly, if $\theta > \mathbb{S}_{x}(X_{(k)})$, then 
$\left(\mathbb{S}_{x}(X_{(j)})-\theta\right)<0$ for $j=1,\cdots,k$, so that $\psi_c\left(\mathbb{S}_{x}(X_{(j)})-\theta \right)<0$ for every $j$. Consequently,\begin{equation*}
    \frac{\partial}{\partial \theta}L_{n}(x;\theta)>0.
\end{equation*} 
Hence, $L_{n}(x,\cdot)$ is strictly increasing on $\left(\mathbb{S}_{x}(X_{(k)}),\infty\right)$. It follows that every minimizer of $L_{n}(x,\cdot)$ lies in the interval $[\mathbb{S}_{x}(X_{(1)}),\mathbb{S}_{x}(X_{(k)})]$. Since $\mathbb{S}_x(X_j)=[d(x,X_j)]^2\geq 0$,
\begin{equation*}
    \arg \min_{\theta\in \mathbb{R}} L_{n}(x;\theta) \subseteq [\mathbb{S}_{x}(X_{(1)}),\mathbb{S}_{x}(X_{(k)})] \subset [0,\infty).
\end{equation*}
This completes the proof. 
 
\hfill $\blacksquare$

\subsubsection*{A.4.\quad \hyperref[Lemma 3.1]{\textcolor{blue}{Proof of Lemma \ref{Lemma 3.1}}}}\label{Proof of Lemma 3.1}
By definition of conditional expectation,
\begin{equation*} 
\mathcal{L}(\theta):=\mathbb{E}\left[\rho_c\left(Z-\theta\right)\mid Z\leq q\right]=\frac{\mathbb{E}\left[\rho_c\left(Z-\theta\right)\mathbbm{1}_{\{Z\leq q\}}\right]}{\Pr(Z\leq q)}.
\end{equation*}
Now, define $g(z,\theta)= \rho_c\left(z-\theta \right) \cdot\mathbbm{1}_{\{z \leq q\}}$. Note that for every fixed $z$, the map $\theta \mapsto g(z,\theta)$ is continuously differentiable with $\frac{\partial}{\partial \theta} \left(g(z,\theta)\right)=- \psi_c(z-\theta)\cdot\mathbbm{1}_{\{z \leq q\}}$. As $\left|\psi_c(u)\right|\leq c,$ for all $u\in \mathbb{R}$, we have $\left|\frac{\partial}{\partial \theta}~g(z,\theta)\right|=\left|\psi_c(z-\theta)\mathbbm{1}_{\{Z\leq q\}}\right|\leq c < \infty$.
Therefore, in view of the Dominated Convergence Theorem (DCT) (\textcite{folland1999modern}, Theorem 2.24, p. 54), we have,
\begin{equation*}
    \frac{d}{d\theta} ~\mathcal{L}(\theta)= -\frac{1}{\Pr(Z\leq q)}\mathbb{E}\left[\psi_c\left(Z-\theta\right)\mathbbm{1}_{\{Z\leq q\}}\right]=-\mathbb{E}\left[\psi_c\left(Z-\theta\right)\mid Z\leq q\right].
\end{equation*}

Next, we study on the continuity of $\mathcal{L}'$.
Let, $\theta_n \rightarrow \theta$ as $n\rightarrow\infty$, and since $\psi_c$ is continuous, $$\psi_c\left(Z-\theta_n\right) \rightarrow \psi_c\left(Z-\theta\right)$$ pointwise, as $n\rightarrow\infty$. As $\left|\psi_c\left(Z-\theta_n\right)\right| \leq c$ for all $n$, 
using Dominated Convergence Theorem 
\begin{equation*}
    \mathcal{L}'(\theta_n)=-\mathbb{E}\left[\psi_c\left(Z-\theta_n\right)\mid Z\leq q\right] \rightarrow -\mathbb{E}\left[\psi_c\left(Z-\theta\right)\mid Z\leq q\right]=\mathcal{L}'(\theta).
\end{equation*}
Hence, $\mathcal{L}':\mathbb{R}\rightarrow \mathbb{R}$ is continuous, establishing $\mathcal{L} \in C^1(\mathbb{R})$. It completes the proof. 
 
\hfill $\blacksquare$

\subsubsection*{A.5.\quad \hyperref[proposition 3.1]{\textcolor{blue}{Proof of Proposition \ref{proposition 3.1}}}} \label{Proof of proposition 3.1}
\noindent For fixed $x \in \mathcal{M}$, the map $\theta \mapsto \psi_c\left(\mathbb{S}_x(X)-\theta\right)$ is continuous and uniformly bounded by $c$. As $\psi_c$ is non-decreasing, for $\theta_1<\theta_2$,~ $\psi_c\left(\mathbb{S}_x(X)-\theta_1\right)\geq \psi_c\left(\mathbb{S}_x(X)-\theta_2\right)$, with strict inequality with positive probability whenever $\mathbb{S}_x(X)$ is non-degenerate, under Assumption (A4). Hence, $\Phi_x$ is strictly decreasing in $\theta$, and continuity of $\Phi_x$ follows from \hyperref[Lemma 3.1]{\textcolor{blue} Lemma \ref{Lemma 3.1}}.


    Further, $$\psi_c\left(\mathbb{S}_x(X)-\theta\right)\rightarrow c$$ as $\theta \rightarrow -\infty$, and $$\psi_c\left(\mathbb{S}_x(X)-\theta\right)\rightarrow -c$$ as $\theta \rightarrow \infty$. As $|\psi_c(u)|\leq c$ for all $u$, the Dominated Convergence Theorem yields, 
    \begin{equation*}
        \lim_{\theta \rightarrow -\infty}\Phi_x(\theta)=c>0,  \quad\lim_{\theta \rightarrow \infty}\Phi_x(\theta)=-c<0.
    \end{equation*} 
Consequently, in view of Intermediate Value Theorem (\textcite{rudin1976}, Theorem 4.23, p. 93) , there exists $\theta_{\mathbb{P},m,c}(x)\in \mathbb{R}$ such that $\Phi_x(\theta)=0$.
Finally, the uniqueness of the solution follows from the strict monotonicity of $\Phi_x$. This completes the proof. 
 
\hfill $\blacksquare$

\subsubsection*{A.6.\quad \hyperref[proposition 3.2]{\textcolor{blue}{Proof of Proposition \ref{proposition 3.2}}}} \label{Proof of proposition 3.2}
\noindent As $\mathbb{S}_x(X):=[d(x,X)]\geq 0$, for every $\theta<0$, we have $\bigl(\mathbb{S}_x(X)-\theta\bigr)>0$, and since $\psi_c(u)>0$ for all $ u>0$, it follows that $\Phi_x(\theta)>0$ for all $\theta<0$.
Moreover, since $\theta_{\mathbb{P},m,c}(x)$ is the unique solution of $\Phi_x(\theta)=0$, therefore, in view of strict monotonicity of $\Phi_x$ (\hyperref[proposition 3.1]{\textcolor{blue} Proposition \ref{proposition 3.1}}), $\theta_{\mathbb{P},m,c}(x)\geq 0$.
Now, under Assumption (A3), absolutely continuity of $\mathbb{P}$ implies that $\mathbb{P}(X=x)=0$ for every $x \in \mathcal{M}$. Hence, $\mathbb{S}_x(X) =[d(x,X)]^2>0 $.
Consequently, $\psi_c\left(\mathbb{S}_x(X)\right)>0$, and in particular,
\begin{equation*}
    \Phi(0)= \mathbb{E}_{\mathbb{P}}\left[\psi_c\left(\mathbb{S}_x(X)\right)\mid \mathbb{S}_x(X)\leq F_x^{-1}(m)\right]>0.
\end{equation*}
Furthermore, since $\Phi_x$ is decreasing using \hyperref[proposition 3.1]{\textcolor{blue} Proposition \ref{proposition 3.1}} and has a unique root, we have, $\theta_{\mathbb{P},m,c}(x)> 0$.

For $\theta >F_x^{-1}(m)$, $\bigl(\mathbb{S}_x(X)-\theta\bigr)<0$ when $\mathbb{S}_x(X) \leq F_x^{-1}(m)$. This fact implies that, $\Phi_x(\theta)<0$.
Again using the monotonicity of $\Phi_x$, we obtain $\theta_{\mathbb{P},m,c}(x)\leq F_x^{-1}(m)$, and for any $m\in(0,1)$,
\begin{align*}
    F_x^{-1}(m)=\inf \{u:F_x(u)\geq m  \}
    &=\inf \{u:\mathbb{P}\left([d(x,X)]^2\leq u \right)\geq m  \}\\ &= \inf \{\varepsilon^2: \mathbb{P}\left(d(x,X)\leq \varepsilon \right)\geq m  \}\\
    &=\left(\inf \{\varepsilon^2: \mathbb{P}\left(d(x,X)\leq \varepsilon \right)\geq m  \} \right)^2=q^2_m(x).
\end{align*}
Next, under Assumption (A2), $F_x^{-1}(m)=q^2_m(x)\leq (\frac{m}{a})^{\frac{2}{b}}$ uniformly over $x\in \mathcal{M}$. Therefore,
\begin{equation*}
    \theta_{\mathbb{P},m,c}(x)\in(0, F_x^{-1}(m)] \subseteq \left(0,\left(\frac{m}{a}\right)^{\frac{2}{b}} \right], 
\end{equation*}
and $\delta_{\mathbb{P},m,c}(x)=\sqrt{\theta_{\mathbb{P},m,c}(x)} \leq \left(\frac{m}{a}\right)^{\frac{1}{b}}$ uniformly over $x\in \mathcal{M}$.
It completes the proof. 
 
\hfill $\blacksquare$


\subsubsection*{A.7.\quad \hyperref[proposition 3.3]{\textcolor{blue}{Proof of Proposition \ref{proposition 3.3}}}} \label{Proof of proposition 3.3}
\noindent Under Assumptions (A3) and (A4), $x\rightarrow \theta_{\mathbb{P},m,c}(x)$ is continuous on $(\mathcal{M},d)$ (see \hyperref[Theorem 3.1]{\textcolor{blue} Theorem \ref{Theorem 3.1}}). As $\mathcal{M}$ is a compact metric space by Assumption (A1),
the Extreme Value Theorem (\textcite{rudin1976}, Theorem 4.16) yields the existence of $x_*\in \mathcal{M}$ satisfying
    \begin{equation*}
        \theta_{\inf}=\inf_{x\in \mathcal{M}}\theta_{\mathbb{P},m,c}(x)=\min _{x\in \mathcal{M}}\theta_{\mathbb{P},m,c}(x)=\theta_{\mathbb{P},m,c}(x_*).
    \end{equation*}
Now, by \hyperref[proposition 3.2]{\textcolor{blue} Proposition \ref{proposition 3.2}}, $\theta_{\mathbb{P},m,c}(x)>0$ for all $x \in \mathcal{M}$, which implies that $\theta_{\inf}>0$.
Finally, $$\inf_{x\in \mathcal{M}}\delta_{\mathbb{P},m,c}(x)= \sqrt{\theta_{\inf}}>0$$ follows from \eqref{eq:3.8}. 
This completes the proof. 
 
\hfill $\blacksquare$

\subsection*{A.8.}
\begin{lemma} \label{Lemma 3.2}
Let, $G(x,\theta):=\mathbb{E}_{\mathbb{P}} \left[\psi_c\left(\mathbb{S}_x(X)-\theta\right) \mathbbm{1}_{\{ \mathbb{S}_x(X)\leq F_x^{-1}(m)\}}\right].$
Under Assumptions (A3) and (A4), for every $x\in \mathcal{M}$ and every $\theta^*\in \mathbb{R}$ satisfying $G(x,\theta^*)=0$, the following holds
     \begin{equation*}
         -\frac{\partial}{\partial\theta} G(x,\theta^*) \geq h_{\inf} \left[ \min \left(F_{x}^{-1}(m), ~\theta^*+c \right)- \max\left(0, ~\theta^*-c\right)\right]>0,
     \end{equation*}
     and consequently,
    \begin{equation*}
        \frac{\partial}{\partial\theta} G(x,\theta^*)<0,
         \quad \text{for all}\quad x\in \mathcal{M}.
    \end{equation*}
\end{lemma}
\noindent \textbf{Proof:} 
As $\psi_c$ is Lipschitz and piecewise continuously differentiable, in view of DCT, we have
\begin{equation}\label{eq:L.8.1}
    \frac{\partial}{\partial\theta} G(x,\theta)=-\mathbb{E}_{\mathbb{P}}\left[\psi'_c\left(\mathbb{S}_x(X)-\theta\right)\cdot\mathbbm{1}_{\{\mathbb{S}_x(X)\leq F_{x}^{-1}(m)\}}\right].
\end{equation}
Moreover, due to the fact that $\psi^{'}_c(u)=\mathbbm{1}_{\{|u|\leq c\}}$ almost everywhere,
\begin{align}
    \frac{\partial}{\partial\theta} G(x,\theta)  
    &= -\mathbb{E}_{\mathbb{P}}\left[\mathbbm{1}_{\{|\left(\mathbb{S}_x(X)-\theta\right)|\leq c\}}\cdot \mathbbm{1}_{\{\mathbb{S}_x(X)\leq F_{x}^{-1}(m)\}}\right] \nonumber \\[1ex]  
    &= -\mathbb{E}_{\mathbb{P}}\left[\mathbbm{1}_{\{|\left(\mathbb{S}_x(X)-\theta\right)|\leq c,~\mathbb{S}_x(X)\leq F_{x}^{-1}(m)\}}\right] \nonumber \\[1ex]  
    &=-\mathbb{P}\left(|\left(\mathbb{S}_x(X)-\theta\right)|\leq c,~\mathbb{S}_x(X)\leq F_{x}^{-1}(m)\right) \nonumber  \\[1ex] 
    &= -\mathbb{P}\left[\max(0,\theta-c)\leq \mathbb{S}_x(X) \leq \min(F_{x}^{-1}(m),\theta+c)\right] \nonumber \\[1ex]  
     &= -\int\limits_{\max\left(0, ~\theta-c\right)}^{\min\left(F_{x}^{-1}(m), ~\theta+c\right)} h_{\mathbb{S}_x}(u)~du.
\end{align}

Now, let $\theta^*$ satisfies $G(x,\theta^*)=0$.
By \hyperref[proposition 3.2]{\textcolor{blue} Proposition \ref{proposition 3.2}}, $0<\theta^*\leq F_{x}^{-1}(m)$, and hence, $$\max(0,\theta^* - c) < \min(F_{x}^{-1}(m),\theta^* + c).$$ Therefore, by Assumption (A4),
\small
\begin{equation}
    -\frac{\partial}{\partial\theta} G(x,\theta^*) \geq h_{\inf}\left[ \min\left(F_{x}^{-1}(m), ~\theta^* + c\right)- \max\left(0, ~\theta^*-c\right)\right]>0,
\end{equation}
\normalsize
which implies
\begin{equation}
    \frac{\partial}{\partial\theta} G(x,\theta^*)<0,\quad \text{for all}\quad x\in \mathcal{M}.
\end{equation}
It completes the proof. 
 
\hfill $\blacksquare$


\subsection*{Appendix B} \label{Appendix B}
\subsubsection*{B.1.\quad \hyperref[Theorem 3.1]{\textcolor{blue}{Proof of Theorem \ref{Theorem 3.1}}}} \label{Proof of Theorem 3.1}
\noindent The proof proceeds in two steps. We first establish continuity of the map $x\mapsto\theta_{\mathbb{P},m,c}(x)$, where $\theta_{\mathbb{P},m,c}(x)$ is the unique solution of $\Phi_x(\theta)=0$, in view of Implicit Function Theorem (\textcite{rudin1976}, Theorem 9.28). Continuity of $\delta_{\mathbb{P},m,c}(x)$ then follows from \hyperref[proposition 3.2]{\textcolor{blue} Proposition \ref{proposition 3.2}} and \eqref{eq:3.8}.

 To study the continuity of the map $x\mapsto\theta_{\mathbb{P},m,c}(x)$, consider the function $G:\mathcal{M}\times \mathbb{R} \rightarrow \mathbb{R}$ by 
    \begin{equation*}
    G(x,\theta):=\mathbb{E}_{\mathbb{P}}\left[\psi_c\left(\mathbb{S}_x(X)-\theta\right) \mathbbm{1}_{\{ \mathbb{S}_x(X)\leq F_x^{-1}(m)\}}\right].
    \end{equation*}
    %
    The definition of conditional expectation implies that
    $G(x,\theta)=m~\Phi_x(\theta)$. Hence, $G(x,\cdot)$ and $\Phi_x(\cdot)$ have the same roots. Hence it is good enough to work on the function $G$.

    We first show that $G$ is jointly continuous on $\mathcal{M}\times \mathbb{R}$. Let $(x_n,\theta_n)\rightarrow (x,\theta)$ as $n \rightarrow \infty$ in $\mathcal{M}\times \mathbb{R}$, then by triangle inequality,
    \begin{equation*}
    \big|G(x_n,\theta_n)-G(x,\theta)\big|\leq \underbrace{\big|G(x_n,\theta_n)-G(x_n,\theta)\big|}_{\textbf{(E1)}}+\underbrace{\big|G(x_n,\theta)-G(x,\theta)\big|}_{\textbf{(E2)}}.
\end{equation*}
Let us now work on \textbf{(E1)}.
\begin{align*}
    \textbf{(E1)}:&\quad \big|G(x_n,\theta_n)-G(x_n,\theta) \big|\\
    &=\left|\mathbb{E}_{\mathbb{P}}\left[\psi_c\left(\mathbb{S}_{x_n}(X)-\theta_n\right) \mathbbm{1}_{\{\mathbb{S}_{x_n}(X)\leq F_{x_n}^{-1}(m)\}}\right]-\mathbb{E}_{\mathbb{P}}\left[\psi_c\left(\mathbb{S}_{x_n}(X)-\theta\right) \mathbbm{1}_{\{\mathbb{S}_{x_n}(X)\leq F_{x_n}^{-1}(m)\}}\right] \right|,\\[1ex]  
    &=\left|\mathbb{E}_{\mathbb{P}}\left[\psi_c\left(\mathbb{S}_{x_n}(X)-\theta_n\right)-\psi_c\left(\mathbb{S}_{x_n}(X)-\theta\right)\right]\mathbbm{1}_{\{\mathbb{S}_{x_n}(X)\leq F_{x_n}^{-1}(m)\}} \right|,\\[1ex] 
    &\leq \left|\mathbb{E}_{\mathbb{P}}\left[\left(\mathbb{S}_{x_n}(X)-\theta_n\right)-\left(\mathbb{S}_{x_n}(X)-\theta\right)\right] \right|,\qquad \qquad \qquad[\text{As $\psi_c$ is 1-Lipschitz}],\\[1ex]  
    &\leq \mathbb{E}_{\mathbb{P}}\left|\left[\left(\mathbb{S}_{x_n}(X)-\theta_n\right)-\left(\mathbb{S}_{x_n}(X)-\theta\right)\right] \right|,\qquad \qquad \qquad [\text{By Jensen's Inequality}],\\[1ex]  
    &=\left|\theta_n-\theta \right|\rightarrow 0, ~ \text{as}~ n\rightarrow \infty.
\end{align*}
Next, we work on \textbf{(E2)}. In \textbf{(E2)}, by Assumption (A4), $F_x$ is continuous and satisfies $F_x(F_x^{-1}(m))=m$. Moreover, as $\mathbb{S}_{x_n}(X) \rightarrow \mathbb{S}_x(X)$ almost surely, we have $F_{x_n}(t)\rightarrow F_{x}(t)$ at every continuity point $t$ of $F_x$. Consequently, the quantile convergence 
$F_{x_n}^{-1}(m) \rightarrow F_{x}^{-1}(m)$, which  follows from the result in
\textcite{vanderVaart1998}, Lemma 21.2. Therefore,
\begin{equation*}
 \psi_c\left(\mathbb{S}_{x_n}(X)-\theta\right) \mathbbm{1}_{\{\mathbb{S}_{x_n}(X)\leq F_{x_n}^{-1}(m)\}}\rightarrow \psi_c\left(\mathbb{S}_x(X)-\theta\right) \mathbbm{1}_{\{\mathbb{S}_x(X)\leq F_x^{-1}(m)\}},
\end{equation*}
almost surely as $n \rightarrow \infty$.

Now, since $|\psi_c(u)|\leq c$ for all $u$, the DCT leads to
$G(x_n,\theta)\rightarrow G(x,\theta)$ as $n\rightarrow\infty$.
Hence, $G$ is jointly continuous on $\mathcal{M}\times\mathbb{R}$. The continuity of $x\mapsto\theta_{\mathbb{P},m,c}(x)$ will follow from the Implicit Function Theorem. Having established the joint continuity of $G$, it remains to verify that the estimating equation is non-degenerate at its root. 

 The joint continuity of $G$, together with $\partial_{\theta}G<0$ at every root (\hyperref[Lemma 3.2]{\textcolor{blue} Lemma \ref{Lemma 3.2}}) shows that the hypotheses of the Implicit Function Theorem are satisfied at every $(x,\theta_{\mathbb{P},m,c}(x)) \in \mathcal{M}\times \mathbb{R}$. 
Hence, for every $x_0\in \mathcal{M}$, there exists a neighbourhood $U_{x_0}\subset \mathcal{M}$ of $x_0$ and a unique continuous map $\theta^*:U_{x_0} \rightarrow \mathbb{R}$ such that $G(x,\theta^*(x))=0$ for all $x\in U_{x_0}$ (see, \textcite{rudin1976}, Theorem 9.28). 

Finally, as the root of $G(x,\theta)=0$
is unique by \hyperref[proposition 3.1]{\textcolor{blue} Proposition \ref{proposition 3.1}}, it follows that $\theta^*(x)= \theta_{\mathbb{P},m,c}(x)$ for all $x\in U_{x_0}$. Hence, $x\mapsto \theta_{\mathbb{P},m,c}(x)$ is continuous on $U_{x_0}$. Moreover, since $x_0\in \mathcal{M}$ is chosen arbitrarily, it follows that the map $x\mapsto \theta_{\mathbb{P},m,c}(x)$ is continuous on $\mathcal{M}$. Finally, Continuity of $x\mapsto \delta_{\mathbb{P},m,c}(x)=\sqrt{\theta_{\mathbb{P},m,c}(x)}$ follows in view of continuity of $x\mapsto \theta_{\mathbb{P},m,c}(x)$.
It completes the proof. 
 
\hfill $\blacksquare$

\begin{corollary} \label{Corollary 3.1}
    Under Assumptions (A3) and (A4), there exists a constant $\beta>0$, independent of $x\in \mathcal{M}$, such that
    \begin{equation*}
        \frac{\partial}{\partial\theta} G(x,\theta_{\mathbb{P},m,c}(x))<-\beta,\quad \text{for all}\quad x\in \mathcal{M}.
    \end{equation*}
\end{corollary}


\noindent \textbf{Proof:} \hyperref[Theorem 3.1]{\textcolor{blue} Theorem \ref{Theorem 3.1}} implies that the map $x \mapsto \delta_{\mathbb{P},m,c}(x)$ is continuous on $\mathcal{M}$. As $G$ is jointly continuous on $\mathcal{M}\times \mathbb{R}$ and continuously differentiable with respect to $\theta$ (see, \eqref{eq:L.8.1}), the map
\begin{equation*}
    x \mapsto -\frac{\partial}{\partial\theta} G(x,\theta_{\mathbb{P},m,c}(x))
\end{equation*}
is continuous on $\mathcal{M}$. Moreover, \hyperref[Lemma 3.2]{\textcolor{blue} Lemma \ref{Lemma 3.2}} asserts that this function is strictly positive for every $x\in \mathcal{M}$. Further, as $\mathcal{M}$ is compact, by Extreme Value Theorem (\textcite{rudin1976}, Theorem 4.16), it attains its minimum and hence, there exists a constant $\beta>0$, independent of $x$, such that
\begin{equation*}
    \beta:=\min_{x\in \mathcal{M}}\left(-\frac{\partial}{\partial\theta} G\left(x,\theta_{\mathbb{P},m,c}(x)\right)\right) >0.
\end{equation*}
This completes the proof. 
 
\hfill $\blacksquare$

\subsubsection*{B.2.\quad \hyperref[Theorem 3.2]{\textcolor{blue}{Proof of Theorem \ref{Theorem 3.2}}}} \label{Proof of Theorem 3.2}
\noindent
We first establish a Lipschitz bound for the estimating equation
\begin{equation*}
    \Phi_x(\theta)= \mathbb{E}_{\mathbb{P}}\left[\psi_c\left(\mathbb{S}_x(X)-\theta\right)\mid \mathbb{S}_x(X)\leq F_x^{-1}(m)\right].
\end{equation*}
In order to establish it, for $x,y\in \mathcal{M}$, define the measurable subsets of $\mathcal{M}$ as follows,
\begin{align*}
    A_x:=&\{z\in \mathcal{M}:\mathbb{S}_x(z)\leq F_x^{-1}(m)\},\\
    A_y:=&\{z\in \mathcal{M}:\mathbb{S}_y(z)\leq F_y^{-1}(m)\}.
\end{align*}
For $X:\Omega \rightarrow \mathcal{M}$ with distribution $\mathbb{P}$, the corresponding events are on $\Omega$ are $\{\omega\in \Omega: X(\omega) \in A_x\}=\{\omega \in \Omega:\mathbb{S}_x(X(\omega))\leq F_x^{-1}(m)\}$,
and
$\{\omega\in \Omega: X(\omega) \in A_y\}=\{\omega \in \Omega:\mathbb{S}_y(X(\omega))\leq F_x^{-1}(m)\}$.
For brevity, we subsequently use 
$A_x:=\{\mathbb{S}_x(X)\leq F_x^{-1}(m)\}$ and $A_y:=\{\mathbb{S}_y(X)\leq F_y^{-1}(m)\}$. Under Assumption (A4), $\mathbb{P}(A_x)=m$, and hence,
$$\Phi_x(\theta)=\frac{1}{m} ~\mathbb{E}_{\mathbb{P}}\left[\psi_c\left(\mathbb{S}_x(X)-\theta\right)\mathbbm{1}_{A_x}\right].$$ Therefore,
\small
\begin{align}
    |\Phi_x(\theta)-\Phi_y(\theta)|&=  \frac{1}{m} \bigg|\mathbb{E}_{\mathbb{P}}\left[\psi_c\left(\mathbb{S}_x(X)-\theta\right) \mathbbm{1}_{A_x}\right]-\mathbb{E}_{\mathbb{P}}\left[\psi_c\left(\mathbb{S}_y(X)-\theta\right) \mathbbm{1}_{A_y}\right] \bigg|,\nonumber  \\[1ex]
    & = \frac{1}{m} \bigg| \mathbb{E}_{\mathbb{P}}  \Big[\psi_c\left(\mathbb{S}_x(X)-\theta \right) \mathbbm1_{A_x}-\psi_c\left(\mathbb{S}_y(X)-\theta\right) \mathbbm{1}_{A_y}\Big] \bigg|, \nonumber  \\[1ex] 
    &  \leq \frac{1}{m} \mathbb{E}_{\mathbb{P}} \Big| \psi_c \left(\mathbb{S}_x(X)-\theta \right) \mathbbm{1}_{A_x}-\psi_c\left(\mathbb{S}_y(X)-\theta\right) \mathbbm{1}_{A_y}\Big|, \quad \text{(By Jensen's inequality)}  \nonumber  \\[1ex] 
    &  =  \frac{1}{m} \left[\mathbb{E}_{\mathbb{P}} \left(\left| \psi_c\left(\mathbb{S}_x(X)-\theta \right) - \psi_c\left(\mathbb{S}_y(X)- \theta\right)  \right|\mathbbm{1}_{A_x} \right)+\mathbb{E}_{\mathbb{P}} \left(\big| \psi_c\left(\mathbb{S}_y(X)-\theta \right) \big| ~\big|\mathbbm{1}_{A_x}-\mathbbm{1}_{A_y}\big|\right)\right]. \nonumber  
\end{align}
\normalsize
As, $\psi_c$ is 1-Lipschitz and $|\psi_c|\leq c$, along with \hyperref[proposition 2.2]{\textcolor{blue} Proposition \ref{proposition 2.2}} (Assumption (A1)), we have
\begin{align}\label{eq:3.9}
    |\Phi_x(\theta)-\Phi_y(\theta)| &\leq \frac{1}{m} \mathbb{E}_{\mathbb{P}} \big| \mathbb{S}_x(X)-\mathbb{S}_y(X) \big| + \frac{c}{m} \mathbb{E}_{\mathbb{P}} \left| \mathbbm{1}_{A_x}- \mathbbm{1}_{A_y}\right|,\nonumber \\
    & \leq \frac{2D}{m} d(x,y)+\frac{c}{m}\underbrace{\mathbb{E}_{\mathbb{P}} \left| \mathbbm{1}_{A_x}- \mathbbm{1}_{A_y}\right|}_{\textbf{(E3)}}.
    \end{align}
    
    Observe that, the term $\big(\mathbbm{1}_{A_x}- \mathbbm{1}_{A_y}\big)$ is non-zero only on the symmetric difference, as follows
\begin{align*}
    A_x\triangle A_y=&\{\mathbb{S}_x(X) \leq  F_x^{-1}(m)\}\triangle \{\mathbb{S}_y(X) \leq  F_y^{-1}(m)\} \nonumber  \\[1ex]  
    = & \underbrace{\{\mathbb{S}_x(X) \leq  F_x^{-1}(m),~\mathbb{S}_y(X)> F_y^{-1}(m)\}}_{A_x \cap A^c_y} \cup \underbrace{\{\mathbb{S}_x(X)> F_x^{-1}(m),~\mathbb{S}_y(X) \leq  F_y^{-1}(m)\}}_{A_y\cap A^c_x}.
\end{align*}
\noindent From (\textbf{E3}), we have,
\begin{equation}\label{eq:3.10}
    \mathbb{E}_{\mathbb{P}} \left| \mathbbm{1}_{A_x}- \mathbbm{1}_{A_y}\right|  \leq \mathbb{E}_{\mathbb{P}}\left| \mathbbm{1}_{\left\{A_x \triangle A_y\right\}}\right| \leq \mathbb{P}\left( A_x \triangle A_y \right) = \mathbb{P}\left(A_x \cap A^c_y\right)+\mathbb{P}\left(A_y \cap A^c_x\right). 
\end{equation}
\noindent Consider the event $\left(A_x\cap A_y\right)$. On this event, $\mathbb{S}_x(X) \leq  F_x^{-1}(m)$ and $\mathbb{S}_y(X) >  F_y^{-1}(m)$. Therefore,
\begin{align*}
     & ( F_x^{-1}(m)-\mathbb{S}_x(X)) < \left(\mathbb{S}_y(X)- \mathbb{S}_x(X)\right) + ( F_x^{-1}(m) - F_y^{-1}(m)),\nonumber  \\[1ex]
     \text{i.e.}~ & \left| F_x^{-1}(m)-\mathbb{S}_x(X)\right| < \left|\mathbb{S}_y(X)- \mathbb{S}_x(X)\right| + \left| F_x^{-1}(m) - F_y^{-1}(m)\right|.
\end{align*}
Using \hyperref[proposition 2.2]{\textcolor{blue}{Proposition \ref{proposition 2.2}}} together with \hyperref[proposition 2.3]{\textcolor{blue}{Proposition \ref{proposition 2.3}}}, we have
\begin{equation*}
    A_x\cap A^c_y \subseteq \left\{\left| F_x^{-1}(m)-\mathbb{S}_x(X)\right| <  4D~d(x,y) \right\}.
\end{equation*}
Consequently,
\begin{align}\label{eq:3.11}
     \mathbb{P}\left(A_x \cap A^c_y\right) & \leq \mathbb{P}\left(\left| \mathbb{S}_x(X)-F_x^{-1}(m)\right| < 4D~d(x,y)\right), \nonumber\\[1ex]
     & \leq \mathbb{P}\left( F_x^{-1}(m)- 4D~d(x,y) < \mathbb{S}_x(X) < F_x^{-1}(m)+ 4D~d(x,y)\right).
\end{align} 
Under Assumption (A4), $\mathbb{S}_x(X):=[d(x,X)]^2$ admits a density $h_{\mathbb{S}_x(X)}$ satisfying $h_{\mathbb{S}_x} \leq h_{\sup}$ $\lambda$-almost everywhere on $\mathcal{S}_{\mu_x}$.
Hence, from \eqref{eq:3.11}, we have  
\begin{equation} \label{eq:3.12}
     \mathbb{P}\left(A_x \cap A^c_y\right)= \int\limits_{F_x^{-1}(m)-4D~d(x,y)}^{F_x^{-1}(m
     )+ 4D~d(x,y)}  h_{\mathbb{S}_x(X)}(t)~dt
     \leq 8D~ h_{\sup}  ~d(x,y).
\end{equation}
Arguing in a similar way, 
\begin{equation}\label{eq:3.12.1}
    \mathbb{P} \left(A_y \cap A^c_x\right) \leq  8D~ h_{\sup}  ~d(x,y).
\end{equation}
Substituting \eqref{eq:3.12} and \eqref{eq:3.12.1} into \eqref{eq:3.9} using (\ref{eq:3.10}), we obtain
\begin{equation} \label{eq:3.13}
    |\Phi_x(\theta)-\Phi_y(\theta)| \leq \frac{2D~(1+8c~h_{\sup})}{m} ~d(x,y)
    := L_1 ~d(x,y),
\end{equation}
for some constant $L_1>0$. As $\Phi_x(\theta_{\mathbb{P},m,c}(x))=0,$ and $  \Phi_y(\theta_{\mathbb{P},m,c}(y))=0$ it follows that
\begin{equation}\label{eq:3.14}
    |\Phi_x(\theta_{\mathbb{P},m,c}(y))-\Phi_x(\theta_{\mathbb{P},m,c}(x))|= |\Phi_x(\theta_{\mathbb{P},m,c}(y))-\Phi_y(\theta_{\mathbb{P},m,c}(y))|\leq L_1 ~d(x,y).
\end{equation}

Next, by \hyperref[Lemma 3.2]{\textcolor{blue} Lemma \ref{Lemma 3.2}}, $\Phi_x$ is continuously differentiable and strictly decreasing in $\theta$. Moreover, Assumption (A2) implies that, $\theta_{\mathbb{P},m,c}(x)\in \Theta=\left(0,\left(\frac{m}{a}\right)^{\frac{2}{b}} \right]$. Hence, there exist some $\xi\in \big(\min\{\theta_{\mathbb{P},m,c}(x),\theta_{\mathbb{P},m,c}(y)\},~\max\{\theta_{\mathbb{P},m,c}(x),\theta_{\mathbb{P},m,c}(y)\} \big)$ such that
 \begin{align}
      & \left|\frac{\partial}{\partial \theta} \Phi_x(\xi)\big(\theta_{\mathbb{P},m,c}(y)-\theta_{\mathbb{P},m,c}(x)\big)\right|
    \leq L_1 ~d(x,y).\nonumber
 \end{align}
 \noindent Then, \hyperref[Corollary 3.1]{\textcolor{blue} Corollary \ref{Corollary 3.1}} asserts that $\left|\frac{\partial}{\partial \theta}\Phi_x(\theta)\right|\geq \beta>0$, uniformly over $x\in \mathcal{M}$ and for all $\theta\in \Theta$, we have
\begin{equation}\label{eq:3.15}
    \left|\theta_{\mathbb{P},m,c}(x)-\theta_{\mathbb{P},m,c}(y)\right|
    \leq \frac{L_1}{\beta} ~d(x,y).
\end{equation}
Therefore, \eqref{eq:3.15} leads to the fact that $\theta_{\mathbb{P},m,c}$ is lipschitz.
\noindent Finally, \hyperref[proposition 3.2]{\textcolor{blue} Proposition \ref{proposition 3.2}} implies that ${\displaystyle\theta_{\inf}:=\inf_{x\in\mathcal{M}}\theta_{\mathbb{P},m,c}(x)>0}$. Hence,
\begin{align}
 \left|\delta_{\mathbb{P},m,c}(x)-\delta_{\mathbb{P},m,c}(y)\right| = \left|\sqrt{\theta_{\mathbb{P},m,c}(x)}-\sqrt{\theta_{\mathbb{P},m,c}(y)}\right| & = \frac{ \left|\theta_{\mathbb{P},m,c}(x)-\theta_{\mathbb{P},m,c}(y)\right|}{\sqrt{\theta_{\mathbb{P},m,c}(x)}+\sqrt{\theta_{\mathbb{P},m,c}(y)}} \nonumber\\
 & \leq \frac{1}{2\sqrt{\theta_{\inf}}} \left|\theta_{\mathbb{P},m,c}(x)-\theta_{\mathbb{P},m,c}(y)\right|.\nonumber
\end{align}
Therefore in view of \eqref{eq:3.15}, we have $\left|\delta_{\mathbb{P},m,c}(x)-\delta_{\mathbb{P},m,c}(y)\right|\leq L_*~d(x,y)$,
with $L_*:=\frac{L_1}{2\beta \sqrt{\theta_{\inf}}}>0$. This completes the proof. 

\hfill $\blacksquare$

\subsubsection*{B.3.\quad \hyperref[Theorem 4.1]{\textcolor{blue}{Proof of Theorem \ref{Theorem 4.1}}}} \label{Proof of Theorem 4.1}
\noindent 
In view of Propositions \hyperref[proposition 3.1]{\textcolor{blue}{\ref{proposition 3.1}}} and \hyperref[proposition 3.2]{\textcolor{blue}{\ref{proposition 3.2}}}, $\theta_{\mathbb{P},m,c}(x)$ and $\widehat{\theta}_{\mathbb{P}_n,m,c}(x)$ are well-defined and lie in the compact parameter space $\Theta:=\big[0,\left(\frac{m}{a}\right)^{\frac{2}{b}}\big]$. Now, by \hyperref[Theorem 3.1]{\textcolor{blue} Theorem \ref{Theorem 3.1}}, $x \mapsto \theta_{\mathbb{P},m,c}(x)$ is continuous on $\mathcal{M}$, and $(x,\theta)\mapsto L(x;\theta)$ is jointly continuous on $\mathcal{M}\times\Theta$. Hence, to prove the assertion of this theorem, it is enough to 
establish the uniform convergence of $L_n$ to $L$ on $\mathcal{M}\times\Theta$, after which uniform consistency of the minimizers follows by a standard argmin argument.

Observe that,
\small
 \begin{align}
    L_n(x;\theta)- L(x;\theta) &= \frac{1}{k}\sum_{j=1}^{k}\rho_c\left(\mathbb{S}_{x}(X_{(j)})- \theta\right)- \mathbb{E}\left[\rho_c\left(\mathbb{S}_x(X)-\theta\right)\mid \mathbb{S}_{x}(X)\leq F_x^{-1}(m)\right]\nonumber\\
    &= \frac{1}{k}\sum_{j=1}^{n}\rho_c\left(\mathbb{S}_{x}(X_{j})- \theta\right) \mathbbm{1}_{\{\mathbb{S}_{x}(X_{j}) \leq \mathbb{S}_{x}(X_{(k)}) \}} - \mathbb{E}\left[\rho_c\left(\mathbb{S}_x(X)-\theta\right)\mid \mathbb{S}_{x}(X)\leq F_x^{-1}(m)\right]\nonumber\\ 
    &= \frac{1}{k}\sum_{j=1}^{n}\rho_c\left(\mathbb{S}_{x}(X_{j})- \theta\right) \mathbbm{1}_{\{\mathbb{S}_{x}(X_{j}) \leq \mathbb{S}_{x}(X_{(k)}) \}} - \frac{1}{k}\sum_{j=1}^{n}\rho_c\left(\mathbb{S}_{x}(X_{j})- \theta\right) \mathbbm{1}_{\{\mathbb{S}_{x}(X_{j}) \leq F_x^{-1}(m) \}}\nonumber\\ 
    &\quad +\frac{1}{k}\sum_{j=1}^{n}\rho_c\left(\mathbb{S}_{x}(X_{j})- \theta\right) \mathbbm{1}_{\{\mathbb{S}_{x}(X_{j}) \leq F_x^{-1}(m) \}}-
    \mathbb{E}\left[\rho_c\left(\mathbb{S}_x(X)-\theta\right)\mid \mathbb{S}_{x}(X)\leq F_x^{-1}(m)\right]\nonumber\\
    &= \overbrace{\frac{1}{k}\sum_{j=1}^{n}\rho_c\left(\mathbb{S}_{x}(X_{j})- \theta\right) \left[\mathbbm{1}_{\{\mathbb{S}_{x}(X_{j}) \leq \mathbb{S}_{x}(X_{(k)}) \}}-
    \mathbbm{1}_{\{\mathbb{S}_{x}(X_{j}) \leq F_x^{-1}(m) \}}\right]}^{\textbf{(I)}}\nonumber\\
    &\quad +\underbrace{\frac{1}{k}\sum_{j=1}^{n}\rho_c\left(\mathbb{S}_{x}(X_{j})- \theta\right) \mathbbm{1}_{\{\mathbb{S}_{x}(X_{j}) \leq F_x^{-1}(m) \}}-\frac{1}{m}
    \mathbb{E}\left[\rho_c\left(\mathbb{S}_x(X)-\theta\right)\mathbbm{1}_{\{\mathbb{S}_{x}(X)\leq F_x^{-1}(m)\}}\right]}_{\textbf{(II)}}.  \label{eq:4.1}
\end{align}
\normalsize

Let us first work on \textbf{(I)}. Observe that under Assumption (A1),
$\mathbb{S}_{x}(X_{j})\in [0,D^2]$ and $ \theta \in \Theta=\left[0,\left(\frac{m}{a}\right)^{\frac{2}{b}}\right]$. Hence, $|\rho_c\left(\mathbb{S}_{x}(X_{j})- \theta\right)|\leq B^*:=\max\Big\{\frac{1}{2}\left(D^4+\left(\frac{m}{a}\right)^{\frac{4}{b}}\right),~  c\big(D^2+\left(\frac{m}{a}\right)^{\frac{2}{b}}+\frac{c^2}{2}\Big\}<\infty$ uniformly over $\mathcal{M}\times \Theta$. Therefore,
\small
\begin{align}
\sup_{\theta \in \Theta}|\textbf{(I)}| &\leq \left| \frac{nB^*}{k} \sum_{j=1}^{n} \Big[\mathbbm{1}\{\mathbb{S}_{x}(X_j)\leq \mathbb{S}_{x}(X_{(k)})\} -\mathbbm{1}\{\mathbb{S}_{x}(X_j)\leq F_x^{-1}(m)\} \Big] \right|\nonumber\\
& \leq\frac{nB^*}{k} \left| \widehat{F}_{n,x} \big(\mathbb{S}_{x}(X_{(k)})\big) - \widehat{F}_{n,x}\big(F_x^{-1}(m)\big)\right|\nonumber\\
&=\frac{nB^*}{k} \left| \frac{k}{n} - \widehat{F}_{n,x}\big(F_x^{-1}(m)\big)\right|\nonumber\\
&\leq \frac{nB^*}{k} \Bigg[ \left| \frac{k}{n}-F_x\big(F_x^{-1}(m)\big)\right| + \left| F_x\big(F_x^{-1}(m)\big) - \widehat{F}_{n,x}\big(F_x^{-1}(m)\big) \right| \Bigg]\quad [\text{By triangle inequality}] \nonumber\\
    &\leq  \frac{nB^*}{k} \Bigg[ \left| \frac{k}{n}-m\right| + \sup_{t\in \mathbb{R}}  \left|\widehat{F}_{n,x}(t)- F_x(t) \right| \Bigg],\nonumber\\
    \sup_{x\in \mathcal{M}}\sup_{\theta \in \Theta}|\textbf{(I)}|&\leq  \frac{nB^*}{k} \Bigg[ \left| \frac{k}{n}-m\right| + \sup_{x\in \mathcal{M}} \sup_{t\in \mathbb{R}}  \left|\widehat{F}_{n,x}(t)- F_x(t) \right| \Bigg]. 
    \label{eq:4.2}
\end{align}
\normalsize
Next, under Assumption (A1), $\mathcal{M} \subseteq \mathbb{R}^{\ell}$ is compact, and hence, the class $$\mathcal{F}=\{\mathbbm{1}_{\left\{[d(x,\cdot)]^2 \leq t\right\}}:~x\in \mathcal{M},t \geq 0 \}$$ has finite VC dimension in view of \textcite{pestov2012index}, Section 3.1. Afterwards, using the Glivenko-Cantelli theorem for VC classes (\textcite{van1996weak}, Chapter 2.4), we have
\begin{equation}
    \sup_{x\in \mathcal{M}} \sup_{t \in \mathbb{R}}\left|\widehat{F}_{n,x}(t)- F_x(t) \right|\xrightarrow{\text{a.s.}}0, ~\text{as}~n\rightarrow \infty.
    \label{eq:4.3}
\end{equation}
 Finally, as $k=\lceil mn \rceil$, we have $\left|\frac{k}{n}-m\right|\leq \frac{1}{n}\rightarrow 0$ as $n \rightarrow \infty$. Combining this with (\ref{eq:4.2}) and (\ref{eq:4.3}) yields
\begin{equation}\label{eq:4.4}
    \sup_{x\in \mathcal{M}} \sup_{\theta \in \Theta}|\textbf{(I)}|\xrightarrow{\text{a.s.}}0,~ \text{as}~n\rightarrow \infty.
\end{equation}

Next, let us work on \textbf{(II)}. First, we decompose term \textbf{(II)} as follows:
 \begin{align}
\textbf{(II)} &= \frac{1}{k} \sum_{j=1}^{n}\rho_c\big(\mathbb{S}_{x}(X_j)-\theta\big) \mathbbm{1}_{\{\mathbb{S}_{x}(X_j)\leq F_x^{-1}(m)\}} - \frac{1}{m} \mathbb{E} \left[ \rho_c\big(\mathbb{S}_{x}(X)-\theta\big) \mathbbm{1}_{\{\mathbb{S}_{x}(X)\leq F_x^{-1}(m)\}} \right] \nonumber\\
&= \frac{n}{k} \left[ \frac{1}{n} \sum_{j=1}^{n} \rho_c\big(\mathbb{S}_{x}(X_j)-\theta\big) \mathbbm{1}_{\{\mathbb{S}_{x}(X_j)\leq F_x^{-1}(m)\}} - \mathbb{E} \left[ \rho_c\big(\mathbb{S}_{x}(X)-\theta\big) \mathbbm{1}_{\{\mathbb{S}_{x}(X)\leq F_x^{-1}(m)\}} \right] \right]\nonumber \\
&\hspace{5.5cm} + \left( \frac{n}{k} - \frac{1}{m}\right) \mathbb{E} \left[ \rho_c\big(\mathbb{S}_{x}(X)-\theta\big) \mathbbm{1}_{\{\mathbb{S}_{x}(X)\leq F_x^{-1}(m)\}} \right]. \label{eq:4.5}
\end{align}
Hence,
\small
\begin{align}
 \sup_{x\in \mathcal{M}} \sup_{\theta\in\Theta}|\textbf{(II)}| & \leq \frac{n}{k} ~\sup_{x\in \mathcal{M}} \sup_{\theta\in\Theta} \left|\frac{1}{n} \sum_{j=1}^{n} \rho_c\big(\mathbb{S}_{x}(X_j)-\theta\big) \mathbbm{1}_{\{\mathbb{S}_{x}(X_j)\leq F_x^{-1}(m)\}} - \mathbb{E} \left[ \rho_c\big(\mathbb{S}_{x}(X)-\theta\big) \mathbbm{1}_{\{\mathbb{S}_{x}(X)\leq F_x^{-1}(m)\}} \right] \right| \nonumber\\
&\hspace{2.5cm} + \left| \frac{n}{k} - \frac{1}{m}\right|~ \sup_{x\in \mathcal{M}} \sup_{\theta\in\Theta} \left|\mathbb{E} \left[ \rho_c\big(\mathbb{S}_{x}(X)-\theta\big) \mathbbm{1}_{\{\mathbb{S}_{x}(X)\leq F_x^{-1}(m)\}} \right] \right|. \label{eq:4.6}
\end{align}
\normalsize

Now, define $h_{x,\theta}(z)=\rho_c\big(z-\theta\big) \mathbbm{1}_{\{z\leq F_x^{-1}(m)\}},$ for all $\theta\in \Theta$. Then for each fixed $\theta\in \Theta$, $\big\{h_{x,\theta}\big(\mathbb{S}_{x}(X_j)\big)\big\}_{j=1}^n$ is an i.i.d. random variables for each fixed $(x,\theta),$
 and $\{h_{x,\theta}:\theta\in \Theta\}$ is uniformly bounded by $B^*$ as $\frac{n}{k}\rightarrow \frac{1}{m}$, and $\|h_{x,\theta}\|_\infty\leq B^*$.
Therefore, it suffices to prove that
\begin{equation}
    \sup_{x\in \mathcal{M}} \sup_{\theta\in\Theta} \left|\frac{1}{n} \sum_{j=1}^{n} \rho_c\big(\mathbb{S}_{x}(X_j)-\theta\big) \mathbbm{1}_{\{\mathbb{S}_{x}(X_j)\leq F_x^{-1}(m)\}} - \mathbb{E} \left[ \rho_c\big(\mathbb{S}_{x}(X)-\theta\big) \mathbbm{1}_{\{\mathbb{S}_{x}(X)\leq F_x^{-1}(m)\}} \right] \right|\xrightarrow{\text{a.s.}}0,~\text{as}~n\rightarrow \infty. \label{eq:4.7}
\end{equation}
As $\mathcal{M}\times \Theta$ is compact and $(x,\theta)\mapsto h_{x,\theta}$ is continuous with respect to  $L^1(\mathbb{P})$-norm $\|f\|_{L^1(\mathbb{P})}=\int|f|~d\mathbb{P}$, the class $\mathcal{H}:=\{h_{x,\theta}: (x,\theta)\in \mathcal{M}\times \Theta\}$ is a continuous image of a compact set, and hence, it is compact in $L^1(\mathbb{P}):=\left\{f:\int |f|~d\mathbb{P} <\infty  \right\}$. 
Therefore, $\mathcal{H}$ is totally bounded in $L^1(P)$ in view of the fact that every compact subset of a metric space is totally bounded. Consequently, for every $\varepsilon>0$, there exists finitely many functions $h_{x_1,\theta_1},\cdots,h_{x_N,\theta_N}\in \mathcal{H}$ such that
\begin{equation}
   \mathcal{H}\subset \bigcup_{i=1}^N \left\{h_{x,\theta}:\left\|h_{x,\theta}-h_{{x_i,\theta_i}}\right\|_{L^1(P)}<\varepsilon\right\}. \label{eq:4.8}
\end{equation}
Equivalently, for every $h_{x,\theta}\in \mathcal{H}$ there exists $i \in \{1,\cdots,N\}$ such that
    $$\mathbb{E}[|h_{x,\theta}(\mathbb{S}_{x}(X))-h_{{x_i,\theta_i}}(\mathbb{S}_{x_i}(X))|]<\varepsilon.$$
(see \textcite{kosorok2008introduction}, Chapters 2 and 9, and \textcite{pollard1984convergence}, Chapter 4). Now, for each fixed $i$, the random variables $\big\{h_{x_i,\theta_{i}}\big(\mathbb{S}_{x}(X_j)\big)\big\}_{j=1}^n$ are i.i.d. and pointwise bounded. Hence, using strong law of large numbers,
\begin{equation*}
    \frac{1}{n}\sum_{j=1}^{n} h_{x_i,\theta_{i}}\big(\mathbb{S}_{x_i}(X_j)\big)- \mathbb{E}\left[h_{x_i,\theta_{i}}\big(\mathbb{S}_{x_i}(X)\big) \right] \xrightarrow{\text{a.s.}} 0,\quad i=1,\cdots,N.
\end{equation*}
Next, fix an arbitrary $(x,\theta)\in \mathcal{M}\times \Theta$ and let 
$h_{x_i,\theta_i}$ be its associated center from the covering described in \eqref{eq:4.8} satisfying
$\mathbb{E}|h_{x,\theta}-h_{{x_i,\theta_i}}|<\varepsilon$. Then using triangle inequality, we have
\begin{align*}
   \left|\frac{1}{n} \sum_{j=1}^{n} h_{x,\theta}\big(\mathbb{S}_{x}(X_j)\big) - \mathbb{E} \left[ h_{x,\theta}\big(\mathbb{S}_{x}(X)\big) \right] \right|
   \leq \left|\frac{1}{n} \sum_{j=1}^{n} h_{x,\theta}\big(\mathbb{S}_{x}(X_j)\big) - \frac{1}{n} \sum_{j=1}^{n} h_{x_i,\theta_{i}}\big(\mathbb{S}_{x_i}(X_j)\big) \right|\\
    + \left|\frac{1}{n} \sum_{j=1}^{n} h_{x_i,\theta_{i}}\big(\mathbb{S}_{x_i}(X_j)\big) - \mathbb{E} \left[ h_{x_i,\theta_{i}}\big(\mathbb{S}_{x_i}(X)\big) \right] \right| + \mathbb{E} \left[\left| h_{x_i,\theta_{i}}\big(\mathbb{S}_{x_i}(X)\big) - h_{x,\theta}\big(\mathbb{S}_{x}(X)\big) \right|\right].
\end{align*}
As it holds for every $(x,\theta)\in \mathcal{M}\times \Theta$, we have
\small
\begin{align}
  \sup_{x \in \mathcal{M}} \sup_{\theta \in \Theta}  \left|\frac{1}{n} \sum_{j=1}^{n} h_{x,\theta}\big(\mathbb{S}_{x}(X_j)\big) - \mathbb{E} \left[ h_{x,\theta}\big(\mathbb{S}_{x}(X)\big) \right] \right|
   \leq & \underbrace{\sup_{x \in \mathcal{M}} \sup_{\theta \in \Theta}  \frac{1}{n} \sum_{j=1}^{n} \left|h_{x,\theta}\big(\mathbb{S}_{x}(X_j)\big) - h_{x_i,\theta_{i}}\big(\mathbb{S}_{x_i}(X_j)\big) \right|}_{\textbf{(T1)}} \nonumber\\
     & + \underbrace{\max_{1\leq i\leq N} \left|\frac{1}{n} \sum_{j=1}^{n} h_{x_i,\theta_{i}}\big(\mathbb{S}_{x_i}(X_j)\big) - \mathbb{E} \left[ h_{x_i,\theta_{i}}\big(\mathbb{S}_{x_i}(X)\big) \right] \right|}_{\textbf{(T2)}} +\varepsilon. \label{eq:4.9}
\end{align}
\normalsize

Let us now work on \textbf{(T1)}. Note that the covering property described in \eqref{eq:4.8} ensures that $$\mathbb{E}[|h_{x,\theta}(\mathbb{S}_{x}(X))-h_{{x_i,\theta_i}}(\mathbb{S}_{x_i}(X))|]<\varepsilon,$$ 
for every $h_{x,\theta}\in \mathcal{H}$, and hence, by the strong law of large numbers $\textbf{(T1)} \xrightarrow{\text{a.s.}} 0$ as $n \rightarrow \infty$.

Next, for \textbf{(T2)}: $\big\{h_{x_i,\theta_{i}} \big(\mathbb{S}_{x}(X_j)\big) \big\}_{j=1}^n$ is i.i.d. random variables and these are uniformly bounded by $B^*<\infty$ for each $i=1,\cdots,n$. A further application to the strong law of large numbers therefore implies
\begin{equation*}
    \frac{1}{n} \sum_{j=1}^{n} h_{x_i,\theta_i}\left(\mathbb{S}_{x_i}(X_j)\right) \xrightarrow{\text{a.s.}} \mathbb{E} \left[ h_{x_i,\theta_i}\left(\mathbb{S}_{x_i}(X)\right) \right],\quad n\rightarrow \infty.
\end{equation*}
Afterwards, a union bound over the finite index set $\{ 1,\cdots,N\}$ yields $\textbf{(T2)}\xrightarrow{\text{a.s.}} 0$ as $n \rightarrow \infty$. Substituting the limiting values of $\textbf{(T1)}$ and $\textbf{(T2)}$ into \eqref{eq:4.9} yields,
\begin{equation} \label{eq:4.9.1}
    \limsup_{n\rightarrow \infty} \sup_{x \in \mathcal{M}} \sup_{\theta \in \Theta}  \left|\frac{1}{n} \sum_{j=1}^{n} h_{x,\theta}\big(\mathbb{S}_{x}(X_j)\big) - \mathbb{E} \left[ h_{x,\theta}\big(\mathbb{S}_{x}(X)\big) \right] \right|\leq 2\varepsilon,\quad \text{almost surely}.
\end{equation}
As $\varepsilon>0$ is chosen arbitrary, \eqref{eq:4.9.1} establishes \eqref{eq:4.7}. Finally, Substituting \eqref{eq:4.7} into \eqref{eq:4.6} and using $\frac{n}{k}\rightarrow \frac{1}{m}$, we have
\begin{equation}\label{eq:4.12}
    \sup_{x \in \mathcal{M}}\sup_{\theta \in \Theta}|\textbf{(II)}|\xrightarrow{\text{a.s.}}0,~ \text{as}~ n \rightarrow \infty. 
\end{equation}
At the end of this point, the decomposition \eqref{eq:4.1} together with \eqref{eq:4.4}) and \eqref{eq:4.12} establishes
\begin{equation}
    \sup_{x \in \mathcal{M}}\sup_{\theta \in \Theta}|L_n(x;\theta)- L(x;\theta)|\xrightarrow{\text{a.s.}}0. \label{eq:4.13}
\end{equation}

Now, let us work on the next part of the proof. For a fixed $\varepsilon>0$, define $$\mathcal{A}_{\varepsilon}=\{(x,\theta)\in \mathcal{M}\times \Theta:|\theta-\theta_{\mathbb{P},m,c}(x)|\geq \varepsilon\}.$$ Observe that continuity of $x\mapsto \theta_{\mathbb{P},m,c}(x)$ on the compact space $\mathcal{M} \times \Theta$ established in \hyperref[Theorem 3.1]{\textcolor{blue} Theorem \ref{Theorem 3.1}}, implies that $\mathcal{A}_{\varepsilon}$ is closed, and hence, compact subset of $\mathcal{M} \times \Theta$.
Now, by \hyperref[proposition 3.1]{\textcolor{blue} Proposition \ref{proposition 3.1}}, $\theta_{\mathbb{P},m,c}(x)$ is the unique minimizer of $L(x;\cdot)$ for every $x\in \mathcal{M}$, and
hence, $$L(x;\theta)- L(x;\theta_{\mathbb{P},m,c}(x))>0,~ \text{for all}~(x,\theta)\in \mathcal{A}_{\varepsilon}.$$
Moreover, since $(x,\theta)\mapsto [L(x;\theta)- L(x;\theta_{\mathbb{P},m,c}(x))]$ is continuous and strictly positive on $\mathcal{A}_{\varepsilon}$, it attains a minimum on $\mathcal{A}_{\varepsilon}$. Therefore,
\begin{equation}\label{eq:4.14}
    \gamma_{\varepsilon}:= \inf_{(x,\theta)\in \mathcal{A}_{\varepsilon}}[L(x;\theta)- L(x;\theta_{\mathbb{P},m,c}(x))]>0.
\end{equation}


 Now, consider $\mathcal{E}_{n,\varepsilon}:=\left\{\sup_{x\in \mathcal{M}}|\widehat{\theta}_{\mathbb{P}_n,m,c}(x)-\theta_{\mathbb{P},m,c}(x)|\geq \varepsilon  \right\}$. On $\mathcal{E}_{n,\varepsilon}$,
 there exists $x\in \mathcal{M}$ such that $|\widehat{\theta}_{\mathbb{P}_n,m,c}(x)-\theta_{\mathbb{P},m,c}(x)|\geq \varepsilon$, and hence, $\left(x,\theta_{\mathbb{P}_n,m,c}(x)\right)\in \mathcal{A}_{\varepsilon}$. Consequently, by (\ref{eq:4.14}), $$L(x,\widehat{\theta}_{\mathbb{P}_n,m,c}(x))\geq L(x,\theta_{\mathbb{P},m,c}(x))+\gamma_{\varepsilon}.$$
 Then, as $\widehat{\theta}_{\mathbb{P}_n,m,c}(x)$ minimizes $L_n(x;\cdot)$, we have $L_n(x;\widehat{\theta}_{\mathbb{P}_n,m,c}(x))\leq L_n(x;\theta_{\mathbb{P},m,c}(x))$. Therefore,
\begin{align*}
    &\gamma_{\varepsilon}\leq L(x,\widehat{\theta}_{\mathbb{P}_n,m,c}(x))- L(x,\theta_{\mathbb{P},m,c}(x)) = [L(x;\widehat{\theta}_{\mathbb{P}_n,m,c}(x))-L_n(x;\widehat{\theta}_{\mathbb{P}_n,m,c}(x))] \\
    & +[L_n(x;\widehat{\theta}_{\mathbb{P}_n,m,c}(x))-L_n(x;\theta_{\mathbb{P},m,c}(x))] +[L_n(x;\theta_{\mathbb{P},m,c}(x))-L(x;\theta_{\mathbb{P},m,c}(x))], \\& \leq 2 \sup_{x \in \mathcal{M}} \sup_{\theta \in \Theta} |L_n(x;\theta)-L(x;\theta)|.
\end{align*}
Consequently,
\begin{equation}\label{eq:4.14.1}
    \mathcal{E}_{n,\varepsilon} \subseteq \left\{ \sup_{x \in \mathcal{M}} \sup_{\theta \in \Theta} |L_n(x;\theta)-L(x;\theta))|>\frac{\gamma_{\varepsilon}}{2} \right\}.
\end{equation}
Observe that \eqref{eq:4.14.1} together with \eqref{eq:4.13} yields
\begin{equation}\label{eq:4.15}
    \text{Pr}\left(\mathcal{E}_{n,\varepsilon}\right)\leq \text{Pr}\left(\sup_{x \in \mathcal{M}} \sup_{\theta \in \Theta} |L_n(x;\theta)-L(x;\theta))|>\frac{\gamma_{\varepsilon}}{2}\right) \rightarrow 0,\quad n\rightarrow \infty, 
\end{equation}
which establishes,
\begin{equation}\label{eq:4.16}
    \sup_{x\in \mathcal{M}} |\widehat{\theta}_{\mathbb{P}_n,m,c}(x)-\theta_{\mathbb{P},m,c}(x)| \xrightarrow{P} 0.
\end{equation}
Finally, by \hyperref[proposition 3.2]{\textcolor{blue} Proposition \ref{proposition 3.2}}, we have ${\displaystyle\theta_{\min}:=\inf_{x\in\mathcal{M}} \theta_{\mathbb{P},m,c}(x)>0}$, 
and the mean value theorem yields
\begin{equation*}
    \|\widehat{\delta}_{\mathbb{P}_n,m,c}-\delta_{\mathbb{P},m,c} \|_{\infty}=\sup_{x\in \mathcal{M}} |\widehat{\delta}_{\mathbb{P}_n,m,c}(x)-\delta_{\mathbb{P},m,c}(x)| \leq\frac{1}{\sqrt{\theta_{\min}}} \sup_{x\in \mathcal{M}} |\widehat{\theta}_{\mathbb{P}_n,m,c}(x)-\theta_{\mathbb{P},m,c}(x)| \xrightarrow{P} 0,
\end{equation*}
in view of \eqref{eq:4.16}. It completes the proof. 
 
\hfill $\blacksquare$

\begin{lemma}\label{Lemma 4.1}
Let $(\mathcal{M},d)$ be a compact metric space, and let  $\mathbb{P}_\alpha = (1-\alpha)~\mathbb{P}_{s} + \alpha~ \mathbb{P}_{c}, ~ 0 \leq \alpha < \frac{1}{2},$
where $\mathbb{P}_{s}$ satisfies Assumption (A4). Let $F_{x,s}$ and $F_{x,\alpha}$ denote the distribution functions of $\mathbb{S}_x(X):=[d(x,X)]^2,~ \text{under}~
\mathbb{P}_{s}$ and $
\mathbb{P}_{\alpha}$,
respectively. Then, for every $m \in (0,1)$,
\begin{equation*}
    \sup_{x \in \mathcal{M}} \left| F^{-1}_{x,\alpha}(m) - F^{-1}_{x,s}(m) \right| \leq \frac{1}{h_{\inf}} \frac{\alpha}{(1-\alpha)}.
\end{equation*}
\end{lemma}
\noindent \textbf{Proof:} See the \hyperref[Proof of Lemma 4.1]{\textcolor{blue}{Proof of Lemma 8.2}} in \hyperref[Appendix B]{\textcolor{blue} Appendix B}.

\hfill $\blacksquare$

\subsubsection*{B.4.\quad \hyperref[Theorem 4.3]{\textcolor{blue}{Proof of Theorem \ref{Theorem 4.3}}}} \label{Proof of Theorem 4.3}

\noindent 
For $x\in \mathcal{M}$, suppose that $\theta_{\mathbb{P}_{s},m,c}(x)$ and $\theta_{\mathbb{P}_{\alpha},m,c}(x)$ denote the respective roots of the population estimating equations
$$\Phi_{\mathbb{P}_{s},x}\left(\theta_{\mathbb{P}_{s},m,c}(x)\right)=0 ~~\text{and}~~\Phi_{\mathbb{P}_{\alpha},x}\left(\theta_{\mathbb{P}_{\alpha},m,c}(x)\right)=0.$$


\noindent Note that, since the contaminating distribution $\mathbb{P}_c$ is unrestricted, the distribution function $F_{x,\alpha}$ of $\mathbb{S}_x(X):=[d(x,X)]^2$ under $\mathbb{P}_{\alpha}$ may not be continuous. Hence $\mathbb{P}_{\alpha}\left(\mathbb{S}_x(X)\leq F_{x,\alpha}^{-1}(m) \right) = m$ may not always be true. To overcome this issue,
define $m_{\alpha}(x):=\mathbb{P}_{\alpha}\left(\mathbb{S}_x(X)\leq F_{x,\alpha}^{-1}(m) \right)$, where $m_{\alpha}(x)\geq m$. The corresponding estimating equation under $\mathbb{P}_{\alpha}$ is therefore given by
\begin{align}\label{eq:4.24}
\Phi_{\mathbb{P}_{\alpha},x}(\theta) & =\mathbb{E}_{\mathbb{P}_{\alpha}}\left[\psi_c\left(\mathbb{S}_x(X)-\theta \right) \mid \mathbb{S}_x(X)\leq F_{x,\alpha}^{-1}(m) \right]\nonumber\\
&= \frac{1}{m_{\alpha}(x)} \int\limits_{\{\mathbb{S}_x(z)\leq F_{x,\alpha}^{-1}(m) \}} \psi_c\left(\mathbb{S}_x(z)-\theta \right)~d\mathbb{P}_{\alpha}(z),\nonumber\\
& = \frac{(1-\alpha)}{m_{\alpha}(x)}\int\limits_{\{\mathbb{S}_x(z)\leq F_{x,\alpha}^{-1}(m) \}} \psi_c\left(\mathbb{S}_x(z)-\theta \right)~d\mathbb{P}_{s}(z) \nonumber\\
     &\hspace{0.1cm} + \frac{\alpha}{m_{\alpha}(x)}\int\limits_{\{\mathbb{S}_x(z)\leq F_{x,\alpha}^{-1}(m) \}} \psi_c\left(\mathbb{S}_x(z)-\theta \right)~d \mathbb{P}_{c}(z).
\end{align}
Using \eqref{eq:4.24} in
$\Phi_{\mathbb{P}_{\alpha},x}(\theta_{\mathbb{P}_{\alpha},m,c}(x))=0$ yields
\small
\begin{align}\label{eq:4.25}
& \frac{(1-\alpha)}{m_{\alpha}(x)}\int\limits_{\{\mathbb{S}_x(z)\leq F_{x,\alpha}^{-1}(m) \}} \psi_c\left(\mathbb{S}_x(z)-\theta_{\mathbb{P}_{\alpha},m,c}(x) \right)~d\mathbb{P}_{s}(z) \nonumber\\
     =&- \frac{\alpha}{m_{\alpha}(x)}\int\limits_{\{\mathbb{S}_x(z)\leq F_{x,\alpha}^{-1}(m) \}} \psi_c\left(\mathbb{S}_x(z)-\theta_{\mathbb{P}_{\alpha},m,c}(x) \right)~d\mathbb{P}_{c}(z), \nonumber\\
     \text{i.e.}~ & 
     \left|~\int\limits_{\{\mathbb{S}_x(z)\leq F_{x,\alpha}^{-1}(m) \}} \psi_c\left(\mathbb{S}_x(z)-\theta_{\mathbb{P}_{\alpha},m,c}(x) \right)~d\mathbb{P}_{s}(z)\right| \leq c~\frac{\alpha}{(1-\alpha)}. 
\end{align}
\normalsize
The last inequality follows from 
the uniform bound $|\psi(u)| \leq c$ for all $u \in \mathbb{R}$.
As $\Phi_{\mathbb{P}_{\alpha},x}(\theta_{\mathbb{P}_{\alpha}, m,c}(x))=0$, it follows that
\small
\begin{align}\label{eq:4.26}
     &\left|\Phi_{\mathbb{P}_{\alpha},x}( \theta_{\mathbb{P}_{\alpha},m,c}(x))-\Phi_{\mathbb{P}_{s},x}( \theta_{\mathbb{P}_{\alpha},m,c}(x)) \right| = \left|\Phi_{\mathbb{P}_{s},x}( \theta_{\mathbb{P}_{\alpha},m,c}(x)) \right| \nonumber\\
     =& \frac{1}{m}  \left|~\int\limits_{\{\mathbb{S}_x(z)\leq F_{x}^{-1}(m) \}} \psi_c\left(\mathbb{S}_x(z)-\theta_{\mathbb{P}_{\alpha},m,c}(x) \right)~d \mathbb{P}_{s}(z)\right|.
\end{align}
\normalsize
\noindent Next, without loss of generality, assume $F_{x,\alpha}^{-1}(m)\leq F_{x}^{-1}(m)$. Then, we have
\small
\begin{align}\label{eq:4.27}
    \left|\Phi_{\mathbb{P}_{s}}( \theta_{\mathbb{P}_{\alpha},m,c}(x)) \right| &= \underbrace{\frac{1}{m}  \left|\int\limits_{\{\mathbb{S}_x(z)\leq F_{x,\alpha}^{-1}(m) \}} \psi_c\left(\mathbb{S}_x(z)-\theta_{\mathbb{P}_{\alpha},m,c}(x) \right)~d \mathbb{P}_{s}(z) \right|}_{(\textbf{B1})} 
     \nonumber\\ 
     & \qquad \qquad \quad +~\underbrace{\frac{1}{m} \left|\int\limits_{\{F_{x,\alpha}^{-1}(m)\leq\mathbb{S}_x(z)\leq F_{x}^{-1}(m) \}} \psi_c\left(\mathbb{S}_x(z)-\theta_{\mathbb{P}_{\alpha},m,c}(x) \right)~d\mathbb{P}_{s}(z)\right|}_{(\textbf{B2})}.
\end{align}
\normalsize
Note that using (\ref{eq:4.25}) we have, $(\textbf{B1})\leq \frac{c}{m}\cdot\frac{\alpha}{(1-\alpha)}$.

It therefore remains to bound  $(\textbf{B2})$. As $|\psi(u)| \leq c$ for all $u \in \mathbb{R}$, observe that
\begin{align}\label{eq:4.28}
    (\textbf{B2}) \leq \frac{c}{m}~\mathbb{P}_{s}\left(F_{x,\alpha}^{-1}(m)\leq\mathbb{S}_x(X)\leq F_{x}^{-1}(m) \right) & \leq \frac{c}{m}~\int\limits_{F_{x,\alpha}^{-1}(m)}^{F_{x}^{-1}(m)} h_x(t) ~dt \nonumber\\
    & \leq \frac{c}{m}~h_{\sup}\left[F_{x}^{-1}(m)-F_{x,\alpha}^{-1}(m)\right],
\end{align}
as under Assumption (A4),
$h_x(t) \leq h_{\sup}$ for all $t \in \mathcal{S}_{\mu_x}$. Moreover, when $F_{x,\alpha}^{-1}(m)> F_{x}^{-1}(m)$ the same bound can be obtained using the same argument. Consequently,
\begin{equation} \label{eq:4.29}
    \left|\Phi_{\mathbb{P}_{\alpha},x}( \theta_{\mathbb{P}_{\alpha},m,c}(x))-\Phi_{\mathbb{P}_{s},x}( \theta_{\mathbb{P}_{\alpha},m,c}(x)) \right| \leq 
    \frac{c}{m}\cdot\frac{\alpha}{(1-\alpha)}+\frac{c~h_{\sup}}{m}~\Big|F_{x,\alpha}^{-1}(m)-F_{x}^{-1}(m)\Big|.
\end{equation}


 Next, using \hyperref[Lemma 4.1]{\textcolor{blue} Lemma \ref{Lemma 4.1}} into \eqref{eq:4.29}, we have
\begin{equation}\label{eq:4.30}
    \left| \Phi_{\mathbb{P}_\alpha,x} \left(\theta_{\mathbb{P}_\alpha,m,c}(x)\right) - \Phi_{\mathbb{P}_{s},x} \left(\theta_{\mathbb{P}_\alpha,m,c}(x)\right) \right| \leq \left( \frac{c}{m} + \frac{c~h_{\sup}}{m~h_{\inf}} \right)~\frac{\alpha}{(1-\alpha)}\leq C_1~\frac{\alpha}{(1-\alpha)},
\end{equation}
for some constant $C_1 > 0$ depending only on $m, c, h_{\inf}$, and $h_{\sup}$. As 
\begin{equation}\label{eq:4.30.1}
    \Phi_{\mathbb{P}_{s},x} \left( \theta_{\mathbb{P}_{s},m,c}(x) \right)=0=\Phi_{\mathbb{P}_\alpha,x} \left(\theta_{\mathbb{P}_\alpha,m,c}(x)\right),
\end{equation}
combining \eqref{eq:4.30.1} with \eqref{eq:4.30} yields
\begin{equation}\label{eq:4.31}
    \left| \Phi_{\mathbb{P}_{s},x} \left( \theta_{\mathbb{P}_\alpha,m,c}(x) \right) - \Phi_{\mathbb{P}_{s},x} \left( \theta_{\mathbb{P}_{s},m,c}(x) \right) \right| \leq C_1~\frac{\alpha}{(1-\alpha)}.
\end{equation}
Now, as $\Phi_{\mathbb{P}_{s},x}$ is differentiable with respect to $\theta$, using the mean value theorem (\textcite{rudin1976}, Theorem 5.10, p.108), there exists $\xi_x \in \bigl(\min\{ \theta_{\mathbb{P}_{s},m,c}(x), \theta_{\mathbb{P}_\alpha,m,c}(x)\}, ~\max\{\theta_{\mathbb{P}_{s},m,c}(x),\theta_{\mathbb{P}_\alpha,m,c}(x)\} \bigr)$ such that
\begin{equation}\label{eq:4.32}
    \left| \frac{\partial}{\partial\theta} \Phi_{\mathbb{P}_{s},x}(\xi_x) \right| \left| \theta_{\mathbb{P}_\alpha,m,c}(x) - \theta_{\mathbb{P}_{s},m,c}(x) \right| \leq C_1~\frac{\alpha}{(1-\alpha)}.
\end{equation}
Next, \hyperref[Corollary 3.1]{\textcolor{blue} Corollary  \ref{Corollary 3.1}} therefore yields
\begin{equation}\label{eq:4.33}
    \left| \theta_{\mathbb{P}_\alpha,m,c}(x) - \theta_{\mathbb{P}_{s},m,c}(x) \right| \leq \frac{C_1}{\beta} \frac{\alpha}{(1-\alpha)}.
\end{equation}
Finally, \hyperref[proposition 3.2]{\textcolor{blue} Proposition \ref{proposition 3.2}} implies along with \eqref{eq:4.33}
\begin{equation} \label{eq:4.34}
    \left| \delta_{\mathbb{P}_\alpha,m,c}(x) - \delta_{\mathbb{P}_{s},m,c}(x) \right| \leq \frac{1}{\sqrt{\theta_{\min}}} \left| \theta_{\mathbb{P}_\alpha,m,c}(x) - \theta_{\mathbb{P}_{s},m,c}(x) \right| \leq \frac{C_1}{\beta\sqrt{\theta_{\min}}} \frac{\alpha}{(1-\alpha)}.
\end{equation}
Hence,
\begin{equation} \label{eq:4.35}
    B(\alpha) = \sup_{x\in \mathcal{M}} \left| \delta_{\mathbb{P}_\alpha,m,c}(x) - \delta_{\mathbb{P}_{s},m,c}(x) \right| \leq C^{*}~\frac{\alpha}{(1-\alpha)},\quad
\text{where}\quad C^{*} = \frac{C_1}{\beta~\sqrt{\theta_{\min}}}>0.
\end{equation}
This completes the proof. 
 
\hfill $\blacksquare$

\subsubsection*{B.5.\quad \hyperref[Lemma 4.1]{\textcolor{blue}{Proof of Lemma \ref{Lemma 4.1}}}}\label{Proof of Lemma 4.1}
\noindent Fix $x \in \mathcal{M}$. As $\mathbb{P}_\alpha=(1-\alpha)~\mathbb{P}_{s}+\alpha~ \mathbb{P}_{c}$, it follows that 
\begin{equation} \label{eq:4.36}
    F_{x,\alpha}(u)= (1-\alpha)~F_{x,s}(u) + \alpha ~F_{x,c}(u),
\end{equation}
where $F_{x, c}$ denotes the distribution function of $\mathbb{S}_x(X)$ under
$\mathbb{P}_{c}$. 
Consequently, $(1-\alpha)~F_{x,s}(u) \leq F_{x,\alpha}(u) \leq(1-\alpha)~F_{x,s}(u)+\alpha, ~\text{for all}~ u\geq 0.$
Using the property of quantile function $F_{x,\alpha}\left(F^{-1}_{x,\alpha}(m)\right) \geq m$, we obtain $m \leq F_{x,\alpha}\left(F^{-1}_{x,\alpha}(m)\right)  \leq (1-\alpha) F_{x,s} \left(F^{-1}_{x,\alpha}(m)\right) + \alpha,$
which implies
\begin{equation} \label{eq:4.39}
    F_{x,s} \left(F^{-1}_{x,\alpha}(m)\right) \geq \frac{m-\alpha}{(1-\alpha)}.
\end{equation}
Further, we also have
$F_{x,\alpha}\left(F_{x,\alpha}^{-1}(m)^-\right)\leq m$. Therefore,
\begin{align} \label{eq:4.40}
F_{x,\alpha}\left(F_{x,\alpha}^{-1}(m)^-\right) = \lim_{u\uparrow F_{x,\alpha}^{-1}(m)} F_{x,\alpha}(u)
= (1-\alpha) \lim_{u\uparrow F_{x,\alpha}^{-1}(m)} F_{x,s}(u) + \alpha \lim_{u\uparrow F_{x,\alpha}^{-1}(m)} F_{x,c}(u)\leq m.
\end{align}
\noindent Moreover, as $F_{x,s}$ is continuous by Assumption (A4), and $F_{x,c} \left(F_{x,\alpha}^{-1}(m)^-\right) \geq 0,$ it follows from \eqref{eq:4.40},
\begin{equation} \label{eq:4.41}
F_{x,s} \left( F_{x,\alpha}^{-1}(m) \right) \leq \frac{m}{1-\alpha}.
\end{equation}
Then, combining \eqref{eq:4.39} and \eqref{eq:4.41}, we obtain
\begin{equation*}
    \frac{m-\alpha}{1-\alpha}\leq F_{x,s} \left( F_{x,\alpha}^{-1}(m) \right) \leq \frac{m}{1-\alpha},
\end{equation*}
which further implies that
\begin{equation} \label{eq:4.42}
    \Big| F_{x,s}\left(F^{-1}_{x,\alpha}(m)\right) - m \Big|  \leq \frac{\alpha}{(1-\alpha)}.
\end{equation}
\noindent 
Furthermore, since $F_{x,s} \left(F^{-1}_{x,s}(m)\right) = m$, along with elementary calculus, Assumption (A4) implies that
\begin{align} \label{eq:4.43}
    \left| F_{x,s}\left(F^{-1}_{x,\alpha}(m)\right) - m \right|= \left| F_{x,s}\left(F^{-1}_{x,\alpha}(m)\right) - F_{x,s}\left(F^{-1}_{x,s}(m)\right) \right| & = \left| ~\int\limits_{F^{-1}_{x,s}(m)}^{F^{-1}_{x,\alpha}(m)} h_x(u)~du\right|,\nonumber\\
    & \geq h_{\inf} \left|F^{-1}_{x,\alpha}(m)- F^{-1}_{x,s}(m)\right|,
\end{align}
where the last inequality follows from Assumption (A4). Combining \eqref{eq:4.42} with \eqref{eq:4.43}, we obtain
\begin{equation} \label{eq:4.44}
    \left|F^{-1}_{x,\alpha}(m)- F^{-1}_{x,s}(m)\right| \leq \frac{1}{h_{\inf}} \left| F_{x,s}\left(F^{-1}_{x,\alpha}(m)\right) - m \right| \leq \frac{1}{h_{\inf}} \frac{\alpha}{(1-\alpha)}.
\end{equation}
Finally, as the bound in \eqref{eq:4.44} is uniform over $x \in \mathcal{M}$, we have
\begin{equation} \label{eq:4.45}
    \sup_{x\in \mathcal{M}}\left|F^{-1}_{x,\alpha}(m)- F^{-1}_{x,s}(m)\right| \leq \frac{1}{h_{\inf}} \frac{\alpha}{(1-\alpha)}.
\end{equation}
This completes the proof. 
 
\hfill $\blacksquare$

\subsubsection*{B.6.\quad \hyperref[Theorem 4.4]{\textcolor{blue}{Proof of Theorem \ref{Theorem 4.4}}}} \label{Proof of Theorem 4.4}
\noindent Recall 
$$\mathbb{VR}\Big(\mathcal{S}_{\mathbb{P}_\alpha},\delta_{\mathbb{P}_\alpha,m,c} \Big):=\bigg\{\mathbb{VR}_{t}\Big( \mathcal{S}_{\mathbb{P}_\alpha},\delta_{\mathbb{P}_\alpha,m,c} \Big)\bigg\}_{t\geq 0},$$ and $$\mathbb{VR}\Big(\mathcal{S}_{\mathbb{P}_s},\delta_{\mathbb{P}_s,m,c} \Big):=\bigg\{\mathbb{VR}_{t}\Big( \mathcal{S}_{\mathbb{P}_s},\delta_{\mathbb{P}_s,m,c} \Big)\bigg\}_{t\geq 0}$$ denote the weighted Vietoris–Rips filtrations induced by $\delta_{\mathbb{P}_{\alpha},m,c}$ on $\mathcal{S}_{\mathbb{P}_\alpha}$ and $\delta_{\mathbb{P}_{s},m,c}$ on $\mathcal{S}_{\mathbb{P}_s}$, respectively. As $\delta_{\mathbb{P}_s,m,c}:\mathcal{M} \rightarrow [0,\infty)$ is $L_*$-Lipschitz by \hyperref[Theorem 3.2]{\textcolor{blue} Theorem \ref{Theorem 3.2}}, and $\mathbb{P}_{c} \in \mathscr{P}(\mathcal{M})$ is arbitrary, application of the triangle inequality for the interleaving distance leads to
\small
\begin{align}
 d_I\big(\mathbb{VR}\left( \mathcal{S}_{\mathbb{P}_\alpha} ,\delta_{\mathbb{P}_\alpha,m,c} \right),  \mathbb{VR}\left(\mathcal{S}_{\mathbb{P}_s},\delta_{\mathbb{P}_s,m,c} \right) \big)
& \leq d_I\big(\mathbb{VR}\left( \mathcal{S}_{\mathbb{P}_\alpha} ,\delta_{\mathbb{P}_\alpha,m,c} \right),  \mathbb{VR}\left(\mathcal{S}_{\mathbb{P}_\alpha},\delta_{\mathbb{P}_s,m,c} \right)\big) \nonumber\\
&\quad + d_I\left(\mathbb{VR}\left(\mathcal{S}_{\mathbb{P}_\alpha},\delta_{\mathbb{P}_s,m,c} \right), \mathbb{VR}\Big(\mathcal{S}_{\mathbb{P}_s},\delta_{\mathbb{P}_s,m,c} \Big) \right).
\end{align}
\normalsize
Now, Propositions \hyperref[proposition 4.1]{\textcolor{blue}{\ref{proposition 4.1}}} and \hyperref[proposition 4.2]{\textcolor{blue}{ \ref{proposition 4.2}}} yield
\small
\begin{align}
    d_I\big(\mathbb{VR}\left( \mathcal{S}_{\mathbb{P}_\alpha} ,\delta_{\mathbb{P}_\alpha,m,c} \right),  \mathbb{VR}\left(\mathcal{S}_{\mathbb{P}_s},\delta_{\mathbb{P}_s,m,c} \right) \big) 
    \leq \|\delta_{\mathbb{P}_\alpha,m,c} - \delta_{\mathbb{P}_s,m,c} \|_{\infty} + (1+L_{*})~ d_H\left(\mathcal{S}_{\mathbb{P}_\alpha},\mathcal{S}_{\mathbb{P}_s}\right).
\end{align}
\normalsize Moreover, it follows from Theorem 4.4 in
\textcite{chazal2009proximity} that $ d_B \leq d_I$, and hence,
\begin{equation}\label{eq:7.47}
    d_B\left( \mathrm{Dgm}_p\big(\delta_{\mathbb{P}_{\alpha},m,c} \big), \mathrm{Dgm}_p\left(\delta_{\mathbb{P}_{s},m,c}\right)\right) \leq \|\delta_{\mathbb{P}_{\alpha},m,c}- \delta_{\mathbb{P}_{s},m,c} \|_{\infty} +(1+L_{*})~
    d_H\left(\mathcal{S}_{\mathbb{P}_{\alpha}},\mathcal{S}_{\mathbb{P}_s}\right)
\end{equation}

\noindent As \eqref{eq:7.47} holds for every $\mathbb{P}_{c} \in \mathscr{P}(\mathcal M)$, therefore from \hyperref[Theorem 4.3]{\textcolor{blue} Theorem \ref{Theorem 4.3}}, we obtain,
\begin{equation}\label{eq:7.48}
    d_B\left( \mathrm{Dgm}_p\big(\delta_{\mathbb{P}_{\alpha},m,c} \big), \mathrm{Dgm}_p\left(\delta_{\mathbb{P}_{s},m,c}\right)\right) \leq B(\alpha)
    +(1+L_{*})~ \sup_{\mathbb{P}_{c}\in \mathscr{P}(\mathcal M)} d_H\left(\mathcal{S}_{\mathbb{P}_{\alpha}},\mathcal{S}_{\mathbb{P}_s}\right).
\end{equation}

Observe that the assertion $B_{p}^{PD}(\alpha) \leq C^{**}~\frac{\alpha}{(1-\alpha)}$ follows from \eqref{eq:7.47} in view as \hyperref[Theorem 4.3]{\textcolor{blue} Theorem \ref{Theorem 4.3}} and due to the fact that ${\displaystyle \sup_{\mathbb{P}_{c}\in \mathscr{P}(\mathcal M)} d_H\left(\mathcal{S}_{\mathbb{P}_{\alpha}},\mathcal{S}_{\mathbb{P}_s}\right)=\mathcal{O}\left(\frac{\alpha}{1-\alpha}\right)}$. In particular, if ${\displaystyle \sup_{\mathbb{P}_{c}\in \mathscr{P}(\mathcal M)} d_H\left(\mathcal{S}_{\mathbb{P}_{\alpha}},\mathcal{S}_{\mathbb{P}_s}\right)}$ $\leq C_H\cdot\frac{\alpha}{1-\alpha}$ for some constant $C_H > 0$, we obtain $C^{**}:=C^* + (1+L_*)~C_H >0$. It completes the proof.


\hfill $\blacksquare$

\printbibliography
\addcontentsline{toc}{section}{References}
\end{document}